\documentclass[twoside,11pt,reqno]{amsart}

\usepackage[T1]{fontenc}
\usepackage{txfonts}
\DeclareSymbolFont{symbols}{OMS}{cmsy}{m}{n}
\usepackage[utf8]{inputenc}
\usepackage[T1]{fontenc}

\usepackage{amsmath,amssymb,amsthm}
\usepackage[utf8]{inputenc}
\usepackage{hyperref}
\usepackage[alphabetic,initials]{amsrefs}
\usepackage{stmaryrd}
\usepackage{enumerate}

\newcommand{\A}{\mathcal{A}}
\newcommand{\cB}{\mathcal{B}}
\newcommand{\cD}{\mathcal{D}}
\newcommand{\cE}{\mathcal{E}}
\newcommand{\C}{\mathcal{C}}
\newcommand{\cC}{\mathcal{C}}
\newcommand{\cM}{\mathcal{M}}

\newcommand{\cI}{\mathcal{I}}
\renewcommand{\O}{\mathcal{O}}

\newcommand{\cO}{\mathcal{O}}

\newcommand{\Hil}{\mathcal{H}}
\newcommand{\cK}{\mathcal{K}}

\newcommand{\M}{\mathcal{M}}
\newcommand{\N}{\mathcal{N}}

\newcommand{\RR}{\mathbb{R}}

\newcommand{\RRc}{\overline{\mathbb{R}}}
\newcommand{\CC}{\mathbb{C}}

\newcommand{\ZZ}{\mathbb{Z}}
\newcommand{\NN}{\mathbb{N}}
\newcommand{\MM}{\mathbb{M}}

\newcommand{\e}{\mathrm{e}}
\newcommand{\dd}{\mathrm{d}}
\newcommand{\ima}{\mathrm{i}}

\DeclareMathOperator{\Ad}{Ad}
\DeclareMathOperator{\aarg}{arg} 
\DeclareMathOperator{\B}{B}
\DeclareMathOperator{\Aut}{Aut}
\DeclareMathOperator{\Pic}{Pic}
\DeclareMathOperator{\id}{id}
\DeclareMathOperator{\DHR}{DHR}
\DeclareMathOperator{\Vir}{Vir}

\DeclareMathOperator{\End}{End}
\DeclareMathOperator{\tr}{tr}
\DeclareMathOperator{\Bim}{Bim}
\DeclareMathOperator{\Mod}{Mod}
\DeclareMathOperator{\Rep}{Rep}

\DeclareMathOperator{\Hom}{Hom}
\DeclareMathOperator{\Ob}{Ob}
\DeclareMathOperator{\Mor}{Mor}

\newcommand{\punkt}{\,\mathrm{.}}
\newcommand{\komma}{\,\mathrm{,}}

\newcommand{\SL}{\mathop{\mathsf{SL}}}

\newcommand{\PSL}{\mathop{\mathsf{PSL}}}
\newcommand{\SU}{\mathop{\mathsf{SU}}}
\newcommand{\PSU}{\mathop{\mathsf{PSU}}}
\newcommand{\U}{{\mathsf{U}}}

\newcommand{\Mob}{\mathsf{M\ddot ob}}

\newcommand{\Mobc}{\widetilde{\mathsf{M\ddot ob}}}

\newcommand{\gG}{\mathsf{G}}

\newcommand{\Sc}[1][]{\mathbb{S}^{1#1}}

\usepackage{xspace}
\DeclareRobustCommand{\eg}{e.g.\@\xspace}
\DeclareRobustCommand{\cf}{cf.\@\xspace}
\DeclareRobustCommand{\ie}{i.e.\@\xspace}
\makeatletter
\DeclareRobustCommand{\etc}{%
    \@ifnextchar{.}%
        {etc}%
        {etc.\@\xspace}%
}
\newcommand\local{commutative\@\xspace} 

\newcommand{\CFTnet}{conformal net\@\xspace}

\newcommand{\three}{III\@\xspace}

\newcommand{\threeone}{III$_1$\@\xspace}

\newcommand{\D}[2]{{}_{#1}\Delta_{#2}}
\newcommand{\NDN}{{}_\N\Delta_\N}
\newcommand{\NDM}{{}_\N\Delta_\M}
\newcommand{\MDN}{{}_\M\Delta_\N}
\newcommand{\MDM}{{}_\M\Delta_\M}

\newcommand{\LR}{{\mathrm{LR}}}
\newcommand{\loc}{{\mathrm{loc}}}
\newcommand{\s}[2]{{}_{#1}\mathcal{C}_{#2}}
\newcommand{\NNs}{\s{\N}{\N}}
\newcommand{\NMs}{\s{\N}{\M}}
\newcommand{\MNs}{\s{\M}{\N}}
\newcommand{\MMs}{\s{\M}{\M}}

\newcommand{\slot}{{~\cdot~}}

\usepackage{graphicx,tikz}
\usetikzlibrary{arrows}
\usetikzlibrary{calc}
\newcommand{\tikzmatht}[2][0.30]
{\vcenter{\hbox{\begin{tikzpicture}[scale=#1]#2
				 \end{tikzpicture}}}
}
\def\colM{black!10}

\def\colN{black!25}
\def\coN{white}
\def\coM{black!12}
\def\coP{black!18}

\newcommand{\mydot}[1]{\begin{scope}[shift={#1}] \fill[shift only] (0,0) circle (1.5pt); \end{scope}}

\newtheorem{thm}{Theorem}[section]
\newtheorem{lem}[thm]{Lemma}
\newtheorem{cor}[thm]{Corollary}
\newtheorem{prop}[thm]{Proposition}
\theoremstyle{definition}
\newtheorem{defi}[thm]{Definition}
\theoremstyle{remark}
\newtheorem{rmk}[thm]{Remark}
\newtheorem{example}[thm]{Example}

\renewcommand{\H}{\Hil}
\renewcommand{\N}{N}
\renewcommand{\M}{M}

\begin{document}
\date{\today}
\dateposted{\today}
\newcommand{\mytitle}{Characterization of 2D rational local
 conformal nets and its boundary conditions: the maximal case}
\title{\mytitle}
\author[M.\ Bischoff]{Marcel Bischoff}
\address[M.\ Bischoff]
{Vanderbilt University, Department of Mathematics, 1326 Stevenson
Center, Nashville, TN 37240, USA}
\email{marcel.bischoff@vanderbilt.edu}
\thanks{
The first author is supported by the German Research Foundation (Deutsche Forschungsgemeinschaft
(DFG)) by the DFG Research Training Group 1493 ``Mathematical Structures in
Modern Quantum Physics'' until August 2014.
}

\author[Y.\ Kawahigashi]{Yasuyuki Kawahigashi}
\address[Y.\ Kawahigashi]{
Department of Mathematical Sciences, The University of Tokyo, 
Komaba, Tokyo 153-8914, Japan
\newline
        Kavli IPMU (WPI), The University of Tokyo
5-1-5 Kashiwanoha, Kashiwa, 277-8583, Japan
}
\email{yasuyuki@ms.u-tokyo.ac.jp}
\thanks{
The second author is supported by the Grants-in-Aid for Scientific Research, JSPS
}
\author[R.\ Longo]{Roberto Longo}
\address[R.\ Longo]
{
Dipartimento di Matematica, Universit\'a di Roma ``Tor Vergata'',
Via della Ricerca Scientifica, 1, I-00133 Roma, Italy
}
\email{longo@mat.uniroma2.it}
\thanks{
       The third author is supported in part by the ERC Advanced Grant 669240 QUEST ``Quantum Algebraic Structures and Models'', PRIN-MIUR and GNAMPA-INdAM. 
}

\begin{abstract}
Let $\A$ be a completely rational local M\"obius covariant net on $S^1$, which
describes a set of chiral observables.  
We show that local M\"obius covariant nets $\cB_2$ on
2D Minkowski space which contains $\A$ as chiral left-right symmetry are in one-to-one
correspondence with Morita equivalence classes of Q-systems in the unitary
modular tensor category $\DHR(\A)$. The M\"obius covariant boundary conditions
with symmetry $\A$ of such a net $\cB_2$ are given by the Q-systems in the
Morita equivalence class or by simple objects in the module category modulo
automorphisms of the dual category. We generalize to reducible boundary conditions.

To establish this result we define the notion of Morita equivalence for Q-systems
(special symmetric $\ast$-Frobenius algebra objects) and non-degenerately braided
subfactors. We prove a conjecture by Kong and Runkel, namely that Rehren's construction (generalized Longo-Rehren
construction, $\alpha$-induction construction) coincides with the categorical
full center. This gives a new view and new results for the study of braided
subfactors.
\end{abstract}


\maketitle
\tableofcontents

\section{Introduction}
The subject of algebraic quantum field theory 
has led to many structural results and recently
also to interesting constructions and classifications in quantum field theory.
Conformal quantum field theory can be 
conveniently studied in this approach.
In particular there is the notion of a conformal QFT on Minkowski space 
and boundary conformal QFT on Minkowski half-plane $x>0$.

One can associate with a boundary conformal QFT (boundary theory) a conformal QFT
on Minkowski space (bulk theory), but in general several boundary theories can have the same
bulk theory, which correspond to different boundary conditions of the bulk
theory.

In a different framework Fuchs, Runkel and Schweigert 
gave a general construction, the so-called TFT construction, of 
a (euclidean) rational full conformal field theory (CFT).
The construction can be divided into two steps: 
first one chooses a certain vertex operator algebra (VOA),
whose representation category $\C$ is a modular tensor category
and which specifies chiral fields. This can be seen as the analytical part. 
Then with a choice of a special symmetric Frobenius algebra object $A\in \C$ one
can construct correlators on an arbitrary Riemann surface. 
The bulk field content depends on the Morita equivalence class of $A$, while 
$A$ itself fixes a boundary condition.

Carpi, and two of the authors gave a general procedure starting from 
an algebraic quantum field theory on the Minkowski space,
to obtain all locally isomorphic boundary conformal QFT nets, in other words to
find all possible boundary conditions (with unique vacuum).
The main purpose of this paper is to show that there is a similar classification for the
boundary conditions for maximal (full) 
(conformal) local nets on Minkowski space and its boundary conditions 
as in the afore mentioned TFT construction.

Let us consider more concretely a quantum field theory on Minkowski space.
By introducing new coordinates $x_\pm=t\mp x$
we identify the two-dimensional Minkowski space $\MM=\{(t,x) \in \RR^2\}$ with
metric $\dd
s^2=\dd t^2-\dd x^2$ with 
the product $L_+\times L_-$ of two light rays $L_\pm=\{(t,x) : t\pm x=0\}$
with metric $\dd s^2= \dd x_+ \dd x_-$.
The densities of conserved quantities (symmetries) are prescribed by left and
right moving chiral fields, \ie fields just depending on $x_+$ or $x_-$,
respectively.

For example for the stress-energy tensor holds $T_{00,01}=T_+(x_+) \pm T_-(x_-)$
and for the conserved $\U(1)$-current holds $j_{0,1}(t,x) = j_+(x_+)\pm j_-(x_-)$.
In the algebraic setting such conserved quantities are abstractly given by a net
$\A_2(O)=\A_+(I)\otimes\A_-(J)$.

In general, there can be other local
observables, so the net of observables is a local extension $\cB(O)\supset
\A_2(O)$ of $\A_2$. We ask this extension to be irreducible
($\cB(O)\cap\A_2(O)'=\CC\cdot 1$), which is for example true if we assume that
$\A_2$ contains the stress energy tensor of $\cB$.

We will also assume 
that the algebras of left and right moving chiral fields are
isomorphic, in other words $\A_2(O)=\A(I)\otimes\A(J)$ 
where $O=I\times J \subset L_+\times L_-$ and $\A$ is a local M\"obius covariant net
on $\RR$. So in this case symmetries are prescribed by the net $\A$.
 
We further assume $\A$ to be completely rational, this is for example true for the net
$\Vir_c$ generated by the stress energy tensor with central charge $c<1$,
$\SU(N)$ loop group models,  or conformal nets associated with even lattices
(lattice compactifications).
The category of Doplicher--Haag--Roberts superselection sectors of a completely
rational conformal
net is a unitary modular tensor category \cite{KaLoMg2001}.

Fixing $\A$ we are, as a first step, interested in classifying all nets $\cB$ ``containing the
symmetries described by $\A$'', \ie to classify all local extensions  $\cB_2\supset \A_2$.
It turns out that the maximal ones are classified by Morita equivalence classes of
chiral extensions $\A\subset \cB$.

Let us look a moment into nets defined on $\MM_+ = \{(t,x) \in M : x>0\}$,
\ie nets with a boundary at $x=0$. 
We are interested to prescribe boundary conditions of $\cB_2$ without 
flow of ``charges'' associated with $\A$. 
The vanishing of the chargeflow across the boundary 
of the charges associated with $\A$ is encoded in the 
algebraic framework via the trivial boundary net 
$\A_+(O) = \A(I)\vee \A(J)$ with $I\times J\in\MM_+$.
This net is locally isomorphic to $\A_2$ restricted to $\MM_+$.
In other words $\A_+$ prescribes the boundary condition of $\A_2$ such that
there is no charge flow across the boundary.

Now given a two-dimensional net $\cB_2$ which contains the given rational symmetries 
described by $\A$, \ie a local irreducible extension $\cB_2\supset \A_2$,
we are now interested in all boundary conditions with no charge flow associated with
$\A$ as above. Such a 
        boundary condition is abstractly given \cite{LoRe2004,CaKaLo2013} by a net $\cB_+\supset \A_+$ 
on $\MM_+$ which is locally isomorphic to $\cB_2$ such that this isomorphism 
restricts to an isomorphism of $\A_+\cong \A_2$.

A classification gets feasibile by operator algebraic methods.
Finite index subfactors $\N\subset\M$ are in one-to-one correspondence with
algebra objects (Q-systems)
in the unitary tensor category $\End(\N)$ of endomorphisms of $\N$.

Local irreducible extension $\cB\supset\A$ of nets with finite index 
give rise to nets of subfactors $\A(O)\subset \cB(O)$ and the
corresponding Q-system (up to isomorphism) is independent of $O$ and is in the category of localized DHR endomorphisms.
Conversely, every such Q-system gives a relatively local extension, which is
local if and only if the Q-system is commutative.
In particular, one has a one-to-one correspondence between Q-systems and
relatively local extensions.

This situation can be abstracted to the setting of braided subfactors, namely we
fix an interval $I$, set $N=A(I)$ and denote by $\NNs$ the category of 
localized DHR endomorphisms which are localized in $I$. 
We can start with a type III factor $N$ and a modular tensor category
$\NNs\subset \End(I)$ and look into subfactors $N\subset M$ such that the
corresponding Q-system is in $\NNs$.
We introduce the notion of Morita equivalence of such braided subfactors.
As a main technical result we show that a conjecture of Kong and Runkel \cite{KoRu2010}
is true. Namely, we show  in Prop.\ \ref{prop:main} 
that the generalized Longo--Rehren construction \cite{Re2000} coincides 
with the full center construction in the categorical literature (\eg \cite{FrFuRuSc2006,KoRu2008}).
We give some consequences on the study of 
braided subfactors and modular invariants. 
This result opens the possiblity to apply many results from the categorical literature
to the braided subfactor and conformal net setting. In particular, we make
use of the result that Q-systems are  
Morita equivalent if and only if they have the same full center \cite{KoRu2008}.

Going back to the conformal net setting we get the main result. Namely, 
maximal 2D extensions $\cB_2\supset \A_2$ are classified 
by Morita equivalence classes of Q-systems in $\Rep(\A)$ (see Prop.\ \ref{prop:class2Dmax}
 and 
	irreducible boundary conditions of $\cB_2$ are classified by 
equivalence classes of irreducible Q-systems in the Morita class (see Prop.\ref{prop:BC}).
We also treat reducible boundary conditions, which were not conisidered before in the literature, and show that we get a classification
by reducible Q-systems.

The article is structured as follows.

In Sec.\ \ref{sec:Pre} we give some background 
on the category of endomorphisms of a type III factor,
Q-systems, unitary modular tensor categories (UMTC), 
braided subfactors and the $\alpha$-induction construction.

In Sec.\ \ref{sec:Morita} we give a notion of Morita equivalence 
for subfactors and Q-systems in UMTCs. The Morita equivalence 
class of a subfactor in a UMTC can be described by 
irreducible sectors in the module category of the subfactor 
modulo automorphisms of some dual category.

In Sec.\ \ref{sec:FullCentre} we show that the $\alpha$-induction 
construction in subfactors coincide with the full center construction
in the categorical literature. This is the first main technical result.

In Sec.\ \ref{sec:Modular} we study maximal \local Q-systems in the category
$\NNs\boxtimes\overline{\NNs}$ (the Drinfel'd center of $\NNs$) 
and give a characterization of them.
We give some application to the study of modular invariants 
and examples of inequivalent extensions with same modular invariant,
\ie example of non-vanishing second cohomology.

In Sec.\ \ref{sec:CNet} we apply our former results to the study of conformal 
field theory on the Minkowski space in the operator algebraic (Haag--Kastler)
framework. We give a proof of a folk theorem about the representation theory of local extensions (Prop.\ \ref{prop:folk}). 
Given a completely rational conformal net $\A$, as the main result, we obtain a classification of
maximal local CFTs containing the chiral observables described by $\A$
and all its boundary conditions.
We also discuss reducible boundary conditions, \ie we drop the assumption
that the boundary condition possesses a unique vacuum.
Finally, we give a relation to the construction of adding a boundary in
\cite{CaKaLo2013}, which gives an alternative proof for the classification of boundary conditions.

\section{Preliminaries}
\label{sec:Pre}
\subsection{Endomorphisms of type III factors and Q-systems}
Let us look into the following strict 2--C$^\ast$-category $\C$. Its  $0$-cells
$\Ob(\C)=\{\N,\M,P,\ldots \}$ are given by a (finite) set of type III  factors.
The 1-cells are given for $\M,\N\in\Ob(\C)$  by $\Mor(\M,\N)$, \ie the set of
unital $\ast$-homomorphisms (morphism) from $\rho:\M \to \N$ with finite (statistical)
dimension $d\rho\equiv d_\rho =[\N:\rho(\M)]^{\frac12}$, where $[\N:\rho(\M)]$ denotes the
minimal index \cite{Jo1983,Ko1986}. The 2-cells are intertwiners, \ie for
$\lambda,\mu\in\Mor(\M,\N)$ we define $\Hom(\lambda,\mu)= \{ t\in\N:
t\lambda(m)=\mu(m)t \text{ for all } m\in\M\}$. Then $\Hom(\lambda,\mu)$ is a vector
space and we write $\langle\lambda,\mu\rangle= \dim\Hom(\lambda,\mu)$ for its
dimension.  Let $\rho\in\Mor(\M,\N)$. We call $\rho$ \textbf{irreducible} if
$\rho(\M)'\cap\N=\CC\cdot 1_\N$.  A sector is a unitary equivalence class
$[\rho]=\{\Ad U\circ\rho: U\in\N \text{ unitary}\}$.  We denote by
$\End(\N)=\Mor(\N,\N)$, which is a 2--C$^\ast$-category with only one $0$-cell,
so a C$^\ast$-tensor category. 

Let $\rho_1,\ldots, \rho_n \in\Mor(\M,\N)$, and let $r_i\in\N$ be generators of
the Cuntz algebra $\mathcal \cO_n$, \ie  $\sum_{i=1}^n r_ir_i^\ast = 1_\N$ and
$r_j^\ast r_i=\delta_{ij}\cdot 1_\N$.  The morphism 
$$
 \rho=\sum_{i=1}^n \Ad r_i\circ \rho_i \in\Mor(\M,\N),
$$
is called \textbf{direct sum} of $\rho_1,\ldots,\rho_n$ and we have $r_i\in
\Hom(\rho_i,\rho)$. The direct sum is unique on sectors and we write it as
\begin{align*}
        [\rho]&=:[\rho_1]\oplus\cdots\oplus[\rho_n]=:\bigoplus_{i=1}^n
[\rho_i]\komma\\
\intertext{and for the multiple direct sum we introduce the notation:
}
        n[\sigma]&:=\bigoplus_{i=1}^n[\sigma]\,, &n\in\NN, \sigma\in\Mor(M,N)
        \punkt
\end{align*}
We say that a full and replete subcategory $\C$ of $\Mor(\M,\N)$ has 
\textbf{subobjects}, if every object is a finite direct sum of irreducible sectors 
in $\C$. Similarly, we say it has \textbf{direct sums},
if $\rho_1,\ldots,\rho_n\in\C$ implies that also their direct sum is in $\cC$.  Let

us assume $\C$ has subobjects.  If $e\in\Hom(\rho,\rho)$ is a (not
necessarily orthogonal) projection (idempotent), then there exists a
$\rho'\in\C$ and $s\in\Hom(\rho',\rho)$ and $t\in\Hom(\rho,\rho')$ such
that $s\cdot t=e$ and $t\cdot s= 1_ {\rho'}\equiv 1_\N$.  We note that if we have
$e\in\Hom(\theta,\theta)$ we have an orthonormal projection $p=e(1+e-e^\ast)^{-1}
\in \Hom(\theta,\theta)$ with the same range.
If $[\rho]=\bigoplus_{i=1}^m [\rho_i]$ and $[\sigma
]=\bigoplus_{j=1}^n[\sigma_j]$ we can decompose $t\in\Hom(\rho,\sigma)$ as 
\begin{align*}
  t&=\bigoplus_{ij} t_{ij}:= s_i\cdot t_{ij}\cdot r_i^\ast, 
  &t_{ij}\in \Hom(\rho_i,\sigma_j)
  \komma
\end{align*}
where $r_i\in\Hom(\rho_i,\rho)$ and $s_j\in\Hom(\sigma_j,\sigma)$ are isometries
as above. Similarly, one can decompose $t\in\Hom(\rho,\sigma\tau)$ etc.

Let us briefly explain the graphical notation (string diagrams)
\cite{JoSt1991,BcEvKa1999,BcEvKa2000,Se2011,BaDoHe2011} 
 which we will use. The 0-cells $\N,\M,\ldots$ are drawn as shaded
two-dimensional regions, with different shadings for each factor. 
A 1-cell $\rho\in\Mor(\N,\M)$ is a vertical 
line (one dimensional) between the region $\M$ and $\N$ and composition of 1-cells 
correspond to	horizontal concatenation. The identity $\id_\N\in\End(\N)$ is 
not drawn. The 2-cells $t\in \Hom(\rho,\sigma)$ are 
vertices between two lines. Sometimes we draw also boxes and 
again the identity $1_\rho\equiv 1 \in \Hom(\rho,\rho)$ is in general not drawn.
The composition of intertwiners is vertical concatenation and the monoidal
product  horizontal concatenation. 

We use a Frobenius rotation invariant convention for trivalent vertices,  
namely for an isometry 
$e\in\Hom(\nu,\lambda\mu)$ we introduce the diagram  
$$
  \tikzmatht{
	  \fill[black!10] (-2,-2) rectangle (2,2);
    \fill[black!20] (1,2)--(1,1.5).. controls (1,1) .. (0,0) 
      .. controls (-1,1) ..  (-1,1.5)--(-1,2);
    \fill[black!30] (1,2)--(1,1.5).. controls (1,1) .. (0,0)--(0,-2)--
      (2,-2)--(2,2);
	  \draw[thick] (0,-2) node [below] {$\nu$}--(0,0) .. controls (1,1) ..
       (1,1.5)--(1,2) node [above] {$\mu$};
		\draw[thick] (0,0) node [right] {$e$} .. controls (-1,1) .. 
      (-1,1.5)--(-1,2) node [above] {$\lambda$};
        \mydot{(0,0)};
	}
  =:\sqrt[4]{\frac{d\lambda d\mu}{d\nu}} e
  \punkt
$$

Let $\C\subset \End(\N)$ and $\cD\subset \End(\M)$ be two full subcategories.
We define the \textbf{Deligne product} $\C\boxtimes \cD$ to be the completion
of $\C\otimes_\CC\cD$ under subobjects and direct sums \cf \cite[Appendix]{LoRo1997}.

A morphism $\bar\rho\colon\N\to \M$ 
is said to be a
  \index{conjugated morphism}
	\textbf{conjugate}
to $\rho\colon\M\to\N$ if there exist intertwiners 
$R\in(\id_\M,\bar\rho\rho)$ and $\bar R\in(\id_\N,\rho\bar\rho)$ 
such that the \textbf{conjugate equations} hold:
\begin{align}
  \label{eq:conj1}
  (1_\rho\otimes R^\ast)\cdot (\bar R\otimes 1_\rho) &\equiv \rho(R^\ast)
  \cdot \bar R =1_\rho
	\\
	\label{eq:conj2}
	(1_{\bar\rho}\otimes \bar R^\ast)\cdot (R\otimes 1_{\bar\rho})
	&\equiv \bar\rho(\bar R^\ast)\cdot R=1_{\bar \rho}
        \punkt
\end{align}
The 2--morphisms $R,\bar R$ will graphically be represented by 
$$
  \bar R=
  \tikzmatht{
		\fill[\colN] (-1.5,-1) rectangle (1.5,1);
		\fill[\colM] (-0.5,1)--(-0.5,0.5) arc (180:360:0.5)--(0.5,1)--cycle;
		\draw (-0.5,1)node [above]{$\rho$}--(-0.5,0.5) arc (180:360:0.5) -- (0.5,1) node[above]{$\bar\rho$};
		\node at (0,-1) [below]{$\id_\N$};
	}
\qquad
	R=
  \tikzmatht{
		\fill[\colM] (-1.5,-1) rectangle (1.5,1);
		\fill[\colN] (-0.5,1)--(-0.5,0.5) arc (180:360:0.5)--(0.5,1)--cycle;
		\draw (-0.5,1)node [above]{$\bar\rho$}--(-0.5,0.5) arc (180:360:0.5) -- (0.5,1) node[above]{$\rho$};
		\node at (0,-1) [below]{$\id_\M$};
	}
$$
and the above equations \eqref{eq:conj1}, \eqref{eq:conj2} are sometimes called
	\index{zig-zag identity}
	\textbf{zig-zag identities}, 
because in diagrams they are given by
$$
	\tikzmatht{
		\fill[\colN] (-2,-1.5) rectangle (2,1.5);
		\fill[\colM] (2,1.5)--(-1,1.5)--(-1,0) arc (180:360:0.5) arc (180:0:0.5)-- (1,-1.5)--(2,-1.5)--cycle;
		\draw (-1,1.5)node[above]{$\rho$}--(-1,0) arc (180:360:0.5) arc (180:0:0.5)-- (1,-1.5)node[below]{$\rho$};
		}
		=
			\tikzmatht{
				 \fill[\colN] (-1,-1.5) rectangle (0,1.5); 
				 \fill[\colM] (0,-1.5) rectangle (1,1.5); 
				 \draw (0,-1.5)node [below]{$\rho$}--(0,1.5)
					node [above]{$\rho$};
			}
		,\qquad    \tikzmatht{
		\fill[\colM] (-2,-1.5) rectangle (2,1.5);
		\fill[\colN] (2,1.5)--(-1,1.5)--(-1,0) arc (180:360:0.5) arc (180:0:0.5)-- (1,-1.5)--(2,-1.5)--cycle;
		\draw (-1,1.5)node[above]{$\bar\rho$}--(-1,0) arc (180:360:0.5) arc (180:0:0.5)-- (1,-1.5) node[below]{$\bar\rho$};
		}
		=
			\tikzmatht{
				 \fill[\colM] (-1,-1.5) rectangle (0,1.5); 
				 \fill[\colN] (0,-1.5) rectangle (1,1.5); 
				 \draw (0,-1.5)node [below]{$\bar\rho$}--(0,1.5)
					node [above]{$\bar\rho$};
			}
		\punkt
$$
If $\rho$ is irreducible we ask the solution $R,\bar R$ to be  
	\textbf{normalized},
\ie $\|R\|=\|\bar R\|$. In the case that $\rho$ is not irreducible 
we further ask that $R,\bar R$ is a 
	\textbf{standard}
solution of the conjugate equation, \ie $R$ (and similar $\bar R$) is of the form
\begin{align*}
	R&= \sum_i (\bar W_i\otimes  W_i) \cdot R_i \equiv\bigoplus_i R_i
	\komma
\end{align*}
where $R_i\in(\id_\M,\bar \rho_i \rho_i)$ is a normalized solution 
for an irreducible object $\rho_i\prec\rho$ and $W_i \in (\rho_i,\rho)$ and $\bar W_i\in(\bar\rho_i,\rho)$ 
are isometries expressing $\rho$ and $\bar \rho$ as direct sums of irreducibles.
We note that for the dimension
$d_\rho\equiv d\rho$ of $\rho$ we have $R^\ast R=d_\rho\cdot 1_\M$ and  
$d\rho=d\bar\rho$.
For $\N\neq \M$ we may always choose $\bar R_\rho=R_{\bar\rho}$.
If we have a subcategory $\NNs\subset\End(\N)$ 
we may choose a system $\NDN$ of representants for every sector in $\NNs$ 
and choose $R_\rho$ for every $\rho\in\NDN$
such that for $[\rho]\neq [\bar\rho]$ we have $\bar R_\rho=R_{\bar\rho}$. 
For $[\bar\rho]=[\rho]$ the intertwiners $R_\rho$ and $\bar R_\rho$ 
are intrinsically related, namely $\bar R_\rho=\pm R_\rho$ holds, where the sign $\pm1$ is called the Frobenius--Schur indicator. In this case the sector $[\rho]$ is 
called \textbf{real} for $+1$ and \textbf{pseudo-real} for $-1$.
Although $[\rho]$ and $[\bar \rho]$ might be represented by the same $\rho\in\NDN$ 
we still use $\bar\rho$ in the diagrammatically notation to distinguish 
between $R_\rho$ and $\bar R_\rho$.

A triple $\Theta=(\theta,w,x)$ with $\theta\in \End(\N)$ and isometries
$w\colon\id_N\to \theta$ and $x\colon\theta\to \theta^2$,
which we will graphically display as 
\begin{align*}
  \sqrt[4]{d\theta}\,w&=
  \tikzmatht{
    \useasboundingbox (-1,-1.5)--(1,2);
    \draw[thick] (0,0)--(0,1.5) node[above]{$\theta$};
    \mydot{(0,0)};
    \node at (0,0) [right] {$w$};
  }
  &
  \sqrt[4]{d\theta}\,x&=
  \tikzmatht{
    \draw[thick] (-1,1.5) node[above]{$\theta$}--(-1,1) arc (180:360:1)--(1,1.5)
      node[above]{$\theta$} (0,0)--(0,-1) node[below]{$\theta$};
    \mydot{(0,0)};
    \node at (0,0) [above] {$x$};
  }
\end{align*}
is called a \textbf{Q-sytem} (\cf \cite{Lo1994,LoRo1997}) if it fulfills
\begin{align*}
  xx&=\theta(x)x & (x\otimes 1_\theta)x &= (1_\theta\otimes x) x
  &\text{(associativity)}\\
  w^\ast x &= \theta(w^\ast)x=\lambda 1_\theta 
  &(w^\ast \otimes 1_\theta)x &= (1_\theta\otimes w^\ast)x=\lambda 1_\theta 
  &\text{(unit law)}
\end{align*}
where $\lambda=\sqrt{d\theta}^{-1}$. In graphical notation this reads:
\begin{align*}
  \tikzmatht{
    \draw[thick] (0,-1) node [below] {$\theta$}--(0,0);
    \mydot{(0,0)};
    \draw[thick] (-1,2.5) node [above] {$\theta$}--(-1,1) arc (180:360:1);
    \draw[thick] (0,2.5) node [above] {$\theta$}--(0,2) arc (180:360:1)--(2,2.5)
      node [above] {$\theta$};
    \mydot{(1,1)};
  }
  &=
  \tikzmatht{
    \begin{scope}[xscale=-1]
    \draw[thick] (0,-1) node [below] {$\theta$}--(0,0);
    \mydot{(0,0)};
    \draw[thick] (-1,2.5) node [above] {$\theta$}--(-1,1) arc (180:360:1);
    \draw[thick] (0,2.5) node [above] {$\theta$}--(0,2) arc (180:360:1)--(2,2.5)
      node [above] {$\theta$};
    \mydot{(1,1)};  
    \end{scope}
  }
  \,;
  &
  \tikzmatht{
    \draw[thick] (0,-1) node [below] {$\theta$}--(0,0);
    \draw[thick] (-1,2) node [above] {$\theta$}--(-1,1) arc (180:360:1)--
    (1,1.5);
    \mydot{(1,1.5)};
    \mydot{(0,0)};
  }
  &=
  \tikzmatht{
    \begin{scope}[xscale=-1]
    \draw[thick] (0,-1) node [below] {$\theta$}--(0,0);
    \draw[thick] (-1,2) node [above] {$\theta$}--(-1,1) arc (180:360:1)--
      (1,1.5);
    \mydot{(1,1.5)};
    \mydot{(0,0)};
    \end{scope}
  }
  = 
  \tikzmatht{
    \draw[thick] (0,-1) node [below] {$\theta$}--(0,2) node [above] {$\theta$};
  }
  \punkt
\end{align*}
Two Q-systems $\Theta=(\theta,w,x)$ and $\tilde\Theta=(\tilde\theta,\tilde
w,\tilde x)$ in $\End(\N)$ are called equivalent, if there is a unitary 
$u\in\Hom(\theta,\tilde\theta)$, such that
\begin{align*}
  \tilde x u&= (u\otimes u) x\equiv u\theta(u)x \,;& 
  u\tilde w&=w\\
\intertext{hold, or graphically:}
	\tikzmatht{
		\draw[thick] (1,3) node [above] {$\tilde\theta$}
		--(1,1.5) arc (360:180:1)--(-1,3) node [above] {$\tilde\theta$};
		\draw[thick] (0,-1.5) node [below] {$\theta$}--(0,0.5);
		\mydot{(0,0.5)};
		\node at (0,0.5) [above] {$\tilde x$};
		\fill[white](-0.5,0) rectangle (.5,-1);
		\draw (-.5,-0) rectangle node{$u$} (.5,-1);
	}
  &=
	\tikzmatht{
		\draw[thick] (1,3) node [above] {$\tilde\theta$}
		--(1,1.5) arc (360:180:1)--(-1,3) node [above] {$\tilde\theta$};
		\draw[thick] (0,-1.5) node [below] {$\theta$}--(0,0.5);
		\mydot{(0,0.5)};
		\node at (0,0.5) [below right] {$x$};
		\fill[white](-0.5,1.5) rectangle (-1.5,2.5);
		\draw (-0.5,1.5) rectangle node{$u$} (-1.5,2.5);
		\fill[white](0.5,1.5) rectangle (1.5,2.5);
		\draw (0.5,1.5) rectangle node{$u$} (1.5,2.5);
  }
  \,;
  &
	\tikzmatht{
    \draw[thick] (0,-2) node [below] {$\theta$}--(0,0);
		\fill[white](-0.5,-0.5) rectangle (.5,-1.5);
		\draw (-.5,-0.5) rectangle node{$u$} (.5,-1.5);
		\mydot{(0,0)};
    \node at (0,0) [above right] {$\tilde w^\ast$};
  }
  &=
	\tikzmatht{
    \draw[thick] (0,-2) node [below] {$\theta$}--(0,0);
		\mydot{(0,0)};
    \node at (0,0) [above right] {$w^\ast$};
  }
  \punkt
\end{align*}

A Q-system in a  C$^\ast$-tensor category automatically \cite{LoRo1997} fulfills
the ``Frobenius law'' 
\begin{align*}
  \begin{array}{rcccl}
  (x^\ast \otimes 1_\theta)(1_\theta \otimes x) 
  \equiv x^\ast\theta(x)
  &=& xx^\ast
  &=&(1_\theta\otimes x^\ast)(x\otimes 1_\theta)\equiv \theta(x^\ast)x\\
\intertext{or graphically:}
  \tikzmatht{
    \draw[thick] (-2,-2) node [below] {$\theta$}--(-2,0) arc (180:0:1) arc
(180:360:1)--(2,2) node [above] {$\theta$};
    \draw[thick] (-1,1)--(-1,2) node [above] {$\theta$};
    \draw[thick] (1,-1)--(1,-2) node [below] {$\theta$};
		\mydot{(1,-1)};
    \mydot{(-1,1)};
  }
  &=&
  \tikzmatht{
    \draw[thick] (-1,-2) node [below] {$\theta$} -- (-1,-1.5) arc (180:0:1)
      -- (1,-2) node [below] {$\theta$};
    \draw[thick] (0,-.5)--(0,.5);
    \draw[thick] (-1,2) node [above] {$\theta$} -- (-1,1.5) arc (180:360:1)
      -- (1,2) node [above] {$\theta$};
		\mydot{(0,-.5)};
    \mydot{(0,.5)};
  }
  &=&
  \tikzmatht{
    \begin{scope}[xscale=-1]
    \draw[thick] (-2,-2) node [below] {$\theta$}--(-2,0) arc (180:0:1) arc
      (180:360:1)--(2,2) node [above] {$\theta$};
    \draw[thick] (-1,1)--(-1,2) node [above] {$\theta$};
    \draw[thick] (1,-1)--(1,-2) node [below] {$\theta$};
		\mydot{(1,-1)};
    \mydot{(-1,1)};
    \end{scope}
  }
  \punkt
  \end{array}
\end{align*}
This means a Q-system is a special symmetric $\ast$-Frobenius algebra object,
but we prefer to use the name Q-system which is most common in the subfactor
context, (other names would be monoid, 
algebra object, monoidal algebra).
We say a Q-system $\Theta=(\theta,w,x)$ is \textbf{irreducible} (called haploid
in the Frobenius algebra context) if 
$\langle \id_\N,\theta\rangle=1$. 

\begin{defi}
  \label{defi:TrivialQSystem}
Every irreducible $a\in\Mor(\M,\N)$ defines an irreducible Q-system
\begin{align*}
  \Theta_a&=\left(\theta_a,w_a,x_a\right)
  := \left(a\bar a, \bar r_a,a(r_a)\right)
\end{align*}
in $\End(\N)$, where $r_a\colon\id_\M\to\bar a a$ and $\bar r_a\colon\id_\N\to a\bar a$ are 
isometries such that 
$\bar R_a=\sqrt{da}\cdot \bar r_a$ and $R_a=\sqrt{da}\cdot r_a$  
fulfill the conjugate equations (\ref{eq:conj1},\ref{eq:conj2}) for $a$. In graphical notation:
\begin{align*}
  \theta_a &= 
  \tikzmatht{
    \fill[\coN] (0.5,-1) rectangle (1,1); 
    \fill[\coM] (-0.5,-1) rectangle (0.5,1); 
    \fill[\coN] (-1,-1) rectangle (-0.5,1); 
    \draw (-0.5,-1)node [below]{$a$}--(-0.5,1)
      node [above]{$a$};
    \draw (0.5,-1)node [below]{$\bar a$}--(0.5,1)
      node [above]{$\bar a$};
  }
  \komma&
  \sqrt{da}\,w_a&=
  \tikzmatht{
    \fill[\coN] (-1,-1) rectangle (1,1);
    \fill[\coM] (-.5,1)--(-0.5,0.5) arc (180:360:.5)--(.5,1);
    \draw (-.5,1) node [above] {$a$}--(-0.5,0.5) arc (180:360:.5)--(.5,1)
      node [above] {$\bar a$};
  }
  \komma
  &\sqrt{da}\,x & =
  \tikzmatht{
    \fill[\coN] (-2,-1.5) rectangle (2,1.5);
    \path (-0.5,-1.5) coordinate (A);
    \path ($(0,1.15)+(235:1.5)$) coordinate (B);
    \path (-1.5,1.15) coordinate (C);
    \path (0.5,-1.5) coordinate (F);
    \path ($(0,1.15)+(-55:1.5)$) coordinate (E);
    \path (1.5,1.15) coordinate (D);
    \fill[\coM] (-1.5,1.5) rectangle (-0.5,1.145);
    \fill[\coM] (1.5,1.5) rectangle (0.5,1.145);
    \fill[\coM] (A) .. controls +(90:0.5) and +(-35:0.5) .. (B)
      arc (235:180:1.5) --++(1,0) arc (180:360:0.5) -- (D)
      arc (0:-55:1.5) .. controls +(-135:0.5) and +(90:0.5) .. (F) -- cycle;
    \draw  (A) node[below]{$a$}.. controls +(90:0.5) and +(-35:0.5) 
      .. (B) arc (235:180:1.5) ++(1,0) arc (180:360:0.5) 
      (D) arc (0:-55:1.5) .. controls +(-135:0.5) and +(90:0.5) .. (F)
      node[below]{$\bar a$};
    \draw (-1.5,1.5) node[above]{$a$}--(-1.5,1.145);
    \draw (-.5,1.5) node[above]{$\bar a$}--(-.5,1.145);
    \draw (.5,1.5) node[above]{$a$}--(.5,1.145);
    \draw (1.5,1.5) node[above]{$\bar a$}--(1.5,1.145);
  }\punkt
\end{align*}
\end{defi}
We remark that up to this point everything can abstractly be defined 
in a 2--C$^\ast$-category. 

Consider now a finite index irreducible subfactor $\N\subset\M$ 
with inclusion $\iota\colon\N\to\M$ then $\Theta:=\Theta_{\bar\iota}$ 
gives \textbf{dual canonical Q-system} of $\N\subset\M$ 
(and $\Gamma=\Theta_\iota$ the canonical Q-system). The
endomorphism $\theta\equiv\bar\iota\iota \in\End(\N)$ is called the \textbf{dual canonical 
endomorphism} of $\N\subset\M$ ($\gamma\equiv\iota\bar\iota\in\End(\M)$ is
called the canonical endomorphism).

Conversely, starting from an irreducible Q-system $\Theta$ in $\End(\N)$, 
there is a subfactor
$\N_1\subset \N$, where $\N_1$ is defined to be the image $\N_1:=E(\N)$ of the 
conditional expectation $E(\slot)=x^\ast\theta(\slot)x$
and there is subfactor (extension) $\N\subset \M$ defined 
by the Jones basic construction $\N_1\subset\N\subset \M$ (\cf \cite{LoRe1995}).
One can make the construction of $\M$ explicit (\cf \cite{BiKaLoRe2014-2})  and obtains this way a dual
morphism
$\bar \iota\colon \M \to \N$ of the inclusion $\iota\colon\N\to\M$ such 
that $\Theta =\Theta_{\bar\iota}$.

The upshot of this discussion is that there is a one-to-one correspondence (\cf \cite{Lo1994}) of 
\begin{itemize}
  \item Q-systems in $\End(\N)$ up to equivalence.
  \item Irreducible finite index subfactors $\N\subset \M$ up to conjugation.
\end{itemize}

\begin{rmk}
\label{rmk:2ndCom}
We note that $\theta$ alone does not fix $\N\subset \M$, which can be seen
as a cohomological obstruction. Izumi and Kosaki \cite{IzKo2002} define
the \textbf{second cohomology}
$H^2(\N\subset\M)$
to be all equivalence classes of Q-systems $\Theta=(\theta,w,x)$ 
with $\theta$ the dual canonical endomorphism of $\N\subset \M$ (their definition
uses actually the canonical endomorphism). We say the second cohomology of
$N\subset M$
vanishes if there up to equivalence is just one Q-system $\Theta=(\theta,x,w)$, 
where $\theta$ is the dual canonical endomorphism of $N\subset M$.
\end{rmk}

We finally note that $\Theta$ is a Q-system in the full $C^\ast$-tensor
subcategory with subobjects generated by $\theta$.
The Q-system becomes ``trivial'', \ie is of the form $\Theta_{\bar\iota}$, in the
2--C$^\ast$-category formed of 0-cells $\{N,M\}$ and full and replete
subcategories $\s{L}{P}\subset
\Mor(P,L)$ with subobjects and direct sums, which is generated by $\{\iota,\bar\iota\}$.
We remark that this is actually a general feature of Frobenius algebra object in
rigid tensor categors, in particular the obtained 2--C$^\ast$-category
together with the 1-morphisms $\iota\colon\N\to \M$ and
$\bar\iota\colon\M\to\N$ appears in \cite{Mg2003} under the name \textbf{Morita
context}. In the general situation having a special symmetric Frobenius algebra
$A$ in a rigid tensor category $\cC$ one can find a bicategory $\tilde  \cC\supset \cC$ giving a Morita context in which the Frobenius
algebra becomes trivial, \cf \cite{Mg2003} for details.

\subsection{UMTCs in \texorpdfstring{$\End(\N)$}{End(N)} and braided subfactors}
\label{subsec:UMTC}

Let us fix a type III factor $\N$ and 
write $\NNs \subset \End(\N)$ for a full and replete 
subcategory $\NNs$ of $\End(\N)$,
such that each object is a finite direct sum of irreducible objects
and $\NNs$ is closed under taking finite direct sums. We use this notation
to stress that it is a category of $\N$-$\N$ morphisms.
We may choose an endomorphism for each irreducible sector and denote the set of these 
endomorphisms by $\NDN$.
Let us assume the following properties:
\begin{enumerate}
	\item $\id_\N\in\NDN$.
	\item There are only finitely many irreducible sectors in $\NNs$, \ie $|\NDN|< \infty$.
	\item If $\sigma\in\NDN$ then also a conjugate (dual) $\bar\sigma\in\NDN$.
	\item If $\rho,\sigma \in \NDN$, then $\rho\circ\sigma\in\NNs$, in other words 
	we have that 
	\begin{align*}
	[\mu\circ\nu]&=\bigoplus N_{\mu\nu}^\rho[\rho],
	&N_{\mu\nu}^\rho&=\langle \rho,\mu\nu\rangle,
\end{align*}
where $N_{\mu\nu}^\rho$ are called \textbf{fusion rule coefficients}.
	\newcounter{enumTemp}
	\setcounter{enumTemp}{\theenumi}
\end{enumerate}
This means that $\NNs$ is a finite rigid C$^\ast$--tensor category
\cite{LoRo1997}, \ie a \textbf{unitary 
fusion category}. 
We associated with $\NNs$ a finite dimensional vector space $K_0(\NNs)\otimes_{\ZZ}\CC
\cong \CC^{|\NDN|}$, where $|\NDN|$ denotes the cardinality of the system $\NDN$
and $K_0(\NNs)$ is the Grothendieck group of the monoidal category $\NNs$.

We define the \textbf{global dimension} $\dim\NNs$ of $\NNs$ to be
$$
\dim\NNs=\sum_{\rho\in\NDN} (d\rho)^2\punkt
$$
 
We remark that for convenience we assume $\NNs$ to be a subcategory of
$\End(\N)$. But it turns out that this is not a lost of generality, because 
every countable generated rigid C$^\ast$--tensor can be embedded 
in $\End(\N)$ by the result of \cite{Ya2003}.

We will need more structure on $\NNs$, in particular we additionally assume:
\begin{enumerate}
	\setcounter{enumi}{\theenumTemp}
	\item There is a natural family $\{\varepsilon(\mu,\nu)\in\Hom(\mu\nu,\nu\mu):\mu,\nu\in \NNs\}$ 
	fulfilling:
	\begin{align*}
		\varepsilon(\lambda,\mu\nu)&=
		(1_\mu\otimes \varepsilon(\lambda,\nu))\cdot
		(\varepsilon(\lambda,\mu)\otimes 1_\nu)
		\equiv \mu(\varepsilon(\lambda,\nu))\cdot
		\varepsilon(\lambda,\mu)
		\\
		\varepsilon(\lambda\mu,\nu)&=(\varepsilon(\lambda,\nu)\otimes
		1_\mu)\cdot(1_\lambda\otimes \varepsilon(\mu,\nu))
		\equiv \varepsilon(\lambda,\nu)\cdot
		\lambda(\varepsilon(\mu,\nu)).
	\end{align*}
	Naturality means, that for $s\colon\sigma\to\sigma'$ and $t\colon\tau\to\tau'$
	\begin{align*}
		(t\otimes s)\cdot \varepsilon(\sigma,\tau)
		&\equiv t\cdot\tau (s)\cdot\varepsilon(\sigma,\tau)
		\\= 
		\varepsilon(\sigma',\tau')\cdot(s\otimes t)
		&\equiv\varepsilon(\sigma',\tau')\cdot
		s\cdot\sigma(t).
	\end{align*}
	We note that this family is determined by $\{\varepsilon(\mu,\nu)\in\Hom(\mu\nu,\nu\mu):\mu,\nu\in \NDN\}$.
	\setcounter{enumTemp}{\theenumi}
\end{enumerate}
That means that $\NNs$ is a \textbf{braided unitary fusion category} 
which has automatically the structure of a 
\textbf{unitary ribbon fusion category}.
We then say that $\NNs\subset \End(\N)$ is a \textbf{URFC}.
The braiding $\varepsilon^+(\lambda,\mu):=\varepsilon(\lambda,\mu)$ 
always comes along with an opposite braiding 
$\varepsilon^-(\lambda,\mu):=\varepsilon(\mu,\lambda)^\ast$
which in general is different from	
 $\varepsilon^+(\lambda,\mu)$. We will graphically denote the braiding by:
\begin{align*}  
  \varepsilon^+(\lambda,\nu)&=
  \tikzmatht{
    \useasboundingbox (-1.5,-3) rectangle (1.5,3);
    \draw[thick]  (-1,2) node [above] {$\nu$} --(-1,1.5)
  	  .. controls (-1,0) and (1,0) .. (1,-1.5)--(1,-2) node [below] {$\nu$};
    \draw[thick] (-1,-2) node [below] {$\lambda$} 	--(-1,-1.5) .. controls 
       (-1,0) and (1,0) ..  (1,1.5) --(1,2) node [above] {$\lambda$};
    \draw [double,ultra thick,white] (-1,-1.5) .. controls (-1,0) and 
      (1,0) ..  (1,1.5);
    \draw[thick] (-1,-2) node [below] {$\lambda$} 	--(-1,-1.5) .. controls 
      (-1,0) and (1,0) ..  (1,1.5) --(1,2) node [above] {$\lambda$};
  }
  &
  \varepsilon^-(\lambda,\nu)&=
  \tikzmatht{
    \useasboundingbox (-1.5,-3) rectangle (1.5,3);
    \draw[thick] (-1,-2) node [below] {$\lambda$} 	--(-1,-1.5) .. controls 
      (-1,0) and (1,0) ..  (1,1.5) --(1,2) node [above] {$\lambda$};
    \draw[thick] (-1,-2) node [below] {$\lambda$} 	--(-1,-1.5) .. controls 
       (-1,0) and (1,0) ..  (1,1.5) --(1,2) node [above] {$\lambda$};
    \draw [double,ultra thick,white] (1,-1.5) .. controls (1,0) and 
      (-1,0) ..  (-1,1.5);
    \draw[thick]  (-1,2) node [above] {$\nu$} --(-1,1.5)
  	  .. controls (-1,0) and (1,0) .. (1,-1.5)--(1,-2) node [below] {$\nu$};
  }
  \punkt
\end{align*}

We denote by $\overline{\NNs}$ the braided category obtained by interchanging
the braiding with the opposite braiding.

Finally, most of the time we will also use the following additional assumption:
\begin{enumerate}
	\setcounter{enumi}{\theenumTemp}
	\item The braiding is non-degenerate, \ie
	$\varepsilon^+(\lambda,\mu)=\varepsilon^-(\lambda,\mu)$
	for all  $\mu\in\NDN$ implies $[\lambda]=[\id_\N]$.
	\setcounter{enumTemp}{\theenumi}
\end{enumerate}
We then say $\NNs$ is \textbf{modular}. In other words $\NNs$ is a 
\textbf{unitary modular tensor category} (UMTC).

We define (see \cite{BcEvKa1999}) for $\lambda,\mu\in\NDN$
\begin{align}
	Y_{\lambda\mu}&=
	\tikzmatht{
	\begin{scope}[yscale=-1]
		\draw[thick] (-2.5,0) arc (180:60:1.5);
		\draw[thick] (-2.5,0) node [left] {$\bar\lambda$} arc (180:390:1.5);
		\draw[thick]  (2.5,0) node [right] {$\bar\mu$} arc (0:210:1.5);
		\draw[thick] (2.5,0) arc (360:240:1.5);
	\end{scope}
	}
  \,;
  &
  \omega_\lambda\cdot 1_\lambda &=
	\tikzmatht{
	\begin{scope}[yscale=-1]
		\draw[thick]  (0,-1.5)node [above] {$\lambda$}--(0,-0.7)
  		(0,0)--(0,1.5) node [below] {$\lambda$};
		\draw [thick] (0,0) arc  (180:490:.7);
	\end{scope}
	}
        \notag
\intertext{and the following $|\NDN|\times|\NDN|$-matrices}
  \label{eq:ST}
  S_{\lambda\mu}&=(\dim \NNs)^{-\frac 12} Y_{\lambda,\mu}\komma 
  &T_{\lambda\mu}&= \e^{-\pi \ima c/12} \delta_{\lambda\mu}\omega_\lambda\komma
\end{align}
where
\begin{align*}
  z&=\sum_{\rho\in\NDN} (d\rho)^2\omega_\rho\,;&
  c&=4\aarg(z)/\pi\punkt
\end{align*}
They obey the relations of the \textbf{partial Verlinde modular algebra}:
  $TSTST=S$,
  $CTC=T$, and $CSC=S$,
where $C_{\mu\nu}=\delta_{\mu,\bar\nu}$ is the \textbf{charge conjugation matrix}.

The property (\theenumTemp) is equivalent to:
\begin{enumerate}
	\setcounter{enumi}{\theenumTemp}
	\item[(6')] $Z(\NNs)\cong\NNs{} \boxtimes {}\overline{\NNs}$,
where $Z(\NNs)$ is the Drinfeld center of $\NNs$ \cite[Corollary 7.11]{Mg2003II} and
	\item[(6'')] the matrix $S=(S_{\lambda\mu})$ is unitary.
\end{enumerate}

In particular, in the modular case we have (\cite[Prop.~2.5]{BcEvKa1999}):
\begin{align*}
  S^\ast S&=T^\ast T =1\komma&
  (ST)^3&=S^2=C\komma&
  CTC=T 
  \komma
\end{align*}
\ie $S$ and $T$ define a unitary representation of $\SL(2,\ZZ)\cong
\ZZ_6\ast_{\ZZ_2}\ZZ_4$ on
$\CC^{|\NDN|}$ if and only if $\NNs$ is modular.

\subsection{Braided subfactors and \texorpdfstring{$\alpha$}{alpha}-induction}
Let $\N$ be a type III factor, $\NNs\subset \End(\N)$ a URFC and let
$\iota(\N)\subset \M$ be an irreducible subfactor such that
$\theta\equiv\bar\iota\iota\in\NNs$.  We call the data $(\iota(\N)\subset
\M,\NNs)$ a \textbf{braided subfactor}. If $\NNs\subset\End(\N)$ happens to be
a UMTC we call the braided subfactor a \textbf{non-degenerately braided}.
There is an obvious
one-to-one correspondence between (the equivalence classes of)  braided subfactors in $\NNs$ and Q-systems in
$\NNs$.

For $\rho \in\NNs$ we define its \textbf{$\alpha$-induction} by
\begin{align*}
	\alpha^\pm_\lambda=\bar\iota^{-1}\circ \Ad (\varepsilon^\pm(\lambda,\theta))  
  \circ \lambda\circ \bar \iota \in \End(\M)
  \punkt
\end{align*}
We define the \textbf{module category} $\NMs$ to be the full subcategory with subobjects and direct sums of
$\Mor(\M,\N)$, which is  generated by $\NNs\bar\iota\equiv\{\rho\bar\iota:\rho\in\NNs\}$
and choose a set of representatives of irreducible sectors $\NDM$.
In the same way we define $\MNs$ and the \textbf{dual category} $\MMs$ generated by $\iota\NNs$ and
$\iota\NNs\bar \iota$, respectively. Finally we define $\MMs^\pm$ to be generated by
$\alpha^\pm(\NNs)$, respectively, and the \textbf{ambichiral category}
$\MMs^0=\MMs^+\cap \MMs^-$. Again we choose a set of representatives of irreducible sectors $\MDN,\MDM,\MDM^\pm,\MDM^0$
in the respective categories.

It turns out  that $\MMs^\pm\subset\MMs$ and that
$\MMs^+\cup\MMs^-$ generates $\MMs$ \cite[Thm.\ 5.10]{BcEvKa1999}.  It will be convenient to work in the 2-category generated by $\NNs\cup\NMs\cup\MNs\cup\MMs$. 

As shown in \cite[Prop.~3.1]{BcEvKa1999}, we have for $a\in\NMs$, 
$\lambda\in\NNs$:
\begin{align*}
  \varepsilon^\pm(\lambda,a\iota)&\in \Hom(\lambda a,a\alpha^\pm_\lambda)&
  \cE^\pm(\lambda,\bar a)\in\Hom(\alpha^\pm_\lambda\bar a,\bar a \lambda)
  \komma
\end{align*}
where
$\cE^\pm(\lambda,\bar a):=T^\ast \iota(\varepsilon^\pm(\lambda,\bar \nu))
\alpha^\pm_\lambda(T)$ for $a\in\NMs$ with $\bar a\prec \bar \iota \nu$ for some
$\nu\in\NNs$ and $T\in(\bar a,\bar\iota \nu)$ an isometry. The definition does
not depend on the choice of $\nu$ and $T$.  We set $\cE^\pm(\bar a,\lambda):=(\cE^\mp(\lambda,\bar
a))^\ast$.  We represent this graphically---where we use 
thin lines for morphisms in $\MNs$ and $\NMs$, normal lines
for endomorphisms in $\NNs$ and thick lines for endomorphisms in $\MMs$---as follows:
\begin{align*}  
	\varepsilon^+(\lambda,a\iota)&=
  \tikzmatht{
    \useasboundingbox (-1.5,-3) rectangle (1.5,3);
    \fill[\coM] (1.5,-2)--(1.5,2)--(-1,2)--(-1,1.5).. controls (-1,0)
      and (1,0) ..  (1,-1.5)--(1,-2);
    \fill[\coN] (-1.5,-2)--(-1.5,2)--(-1,2)--(-1,1.5).. controls 
      (-1,0) and (1,0) ..  (1,-1.5)--(1,-2);
	  \draw  (-1,2) node [above] {$a$} --(-1,1.5)
  	  .. controls (-1,0) and (1,0) .. (1,-1.5)--(1,-2) node [below] {$a$};
    \draw[thick] (-1,-2) node [below] {$\lambda$} 	--(-1,-1.5) .. controls 
      (-1,0) and (1,0) ..  (1,1.5) --(1,2) node [above] {$\alpha^+_\lambda$};
  	\draw [double,ultra thick,\coN] (-1,-1.5) .. controls (-1,0) and 
      (1,0) ..  (1,1.5);
    \draw [double,ultra thick,\coM] (0,0) .. controls (0.5,.375) and
      (1,0.75) ..  (1,1.5)--(1,2);
	  \draw [thick] (-1,-1.5) .. controls (-1,0) and (1,0) ..	(1,1.5);
    \draw[ultra thick] (0,0)..controls (0.5,.375) and (1,0.75)..(1,1.5)--(1,2);
	}
  \,;
  &
  \cE^+(\lambda,\bar a)&=
	\tikzmatht{
    \useasboundingbox (-1.5,-2.5) rectangle (1.5,2.5);
    \begin{scope}[rotate=180]
    \fill[\coM] (1.5,-2)--(1.5,2)--(-1,2)--(-1,1.5).. controls (-1,0)
       and (1,0) ..  (1,-1.5)--(1,-2);
    \fill[\coN] (-1.5,-2)--(-1.5,2)--(-1,2)--(-1,1.5).. controls 
      (-1,0) and (1,0) ..  (1,-1.5)--(1,-2);
	  \draw (-1,2) node [below] {$\bar a$} --(-1,1.5)
    	.. controls (-1,0) and (1,0) .. (1,-1.5)--(1,-2) node [above] {$\bar a$};
    \draw[thick]
      (-1,-2) node [above] {$\lambda$} 	--(-1,-1.5)
      .. controls (-1,0) and (1,0) ..  (1,1.5) --(1,2) node [below] 
      {$\alpha^+_\lambda$};
  	\draw[double,ultra thick, \coN] (-1,-1.5) .. controls (-1,0) and 
      (1,0)..(1,1.5); 
    \draw[double,ultra thick, \coM] (0,0) .. controls (0.5,.375) and
      (1,0.75) ..  (1,1.5)--(1,2);
	  \draw[thick] (-1,-1.5) .. controls (-1,0) and (1,0) ..	(1,1.5);
    \draw[ultra thick] (0,0)..controls (0.5,.375) and (1,0.75)..(1,1.5)--(1,2);
  \end{scope}
	}
  \punkt
\end{align*}
The \textbf{intertwining braided fusion equations} (IBFE's) 
\cite[Prop.~3.3]{BcEvKa1999} hold, namely 
\begin{align*}
  \rho(t)\,\varepsilon^\pm(\lambda,\rho)&=\varepsilon^\pm(a\iota,\rho)\,
  a(\cE^\pm (\bar b,\rho))\,t\,,
  \\
  t\,\varepsilon^\pm(\rho,\lambda)&=a(\cE^\pm(\rho,\bar b))\,
  \varepsilon^\pm(
  \rho,a\iota)\,\rho(t)\,,
  \\
  \rho(y)\,\varepsilon^\pm,(a\iota,\rho)&=\varepsilon^\pm(\lambda,\rho)\,
  \lambda( \varepsilon^\pm(b\iota,\rho))\,y\,,
  \\
  y\,\varepsilon^\pm(\rho,a\iota)&=\lambda(\varepsilon^\pm(\rho,b
  \iota)\,\varepsilon^\pm(\rho,\lambda))\,\rho(y)\,,
  \\
  \alpha^\mp(Y)\,\cE^\pm(\bar a,\rho)&=\cE^\pm(\bar b,\rho)\,
  \bar b(\varepsilon^\pm(\lambda,\rho))\,Y\,,
  \\
  Y\,\cE^\pm(\rho,\bar a)&=\bar b (\varepsilon^\pm(\rho,\lambda))\,
  \cE^\pm(\rho, \bar b)\,\alpha^\pm_\rho\rho(Y)
  \komma
\end{align*}
where $\lambda,\rho\in\NNs$, $a,b\in\NMs$ with conjugates $\bar a,\bar b\in\MNs$;
$t\in\Hom(\lambda,a\bar b)$, $y\in\Hom(a,\lambda b)$ and $Y\in\Hom(\bar a,\bar
b\lambda)$. The IBFE's have simple graphical interpretation, \eg the first and sixth
equations are represented by:
\begin{align*}
  \tikzmatht{
    \useasboundingbox (-3,-4.5) rectangle (3,4.5);
	  \fill[\coN] (-3,-3) rectangle (3,3);
		\draw[thick](0,-3) node [below] {$\lambda$}--(0,0);
		\fill[\coM] (1,3)--(1,-1) .. controls (1,-1.5) ..(0,-2.5)
       .. controls  (-1,-1.5)  .. (-1,-1.5)--(-1,3);
		\draw (1,3) node [above] {$\bar b$}--(1,-1) .. controls (1,-1.5) ..
      (0,-2.5) node [above] {$t$} .. controls  (-1,-1.5)  .. (-1,-1)--(-1,3)
		node [above] {$a$};
    \draw[thick,-round cap] (2,-3) node [below] {$\rho$}
      ..controls (2,-2)..(1.1,-1.1);
    \draw[thick,round cap-] (-1.1,1.1)..controls (-2,2)..(-2,3)
      node [above] {$\rho$};
    \draw[ultra thick, line cap=round] (0.85,-0.85)-- node
      [above] {$\alpha^-_\rho$} (-0.85,0.85);
		\mydot{(0,-2.5) };
	}
  &=
  \tikzmatht{
    \useasboundingbox (-3,-4.5) rectangle (3,4.5);
	  \fill[\coN] (-3,-3) rectangle (3,3);
		\draw[thick](0,-3) node [below] {$\lambda$}--(0,1);
		\fill[\coM] (1,3)--(1,2.5) .. controls (1,2) ..(0,1)
       .. controls  (-1,2)  .. (-1,2.5)--(-1,3);
		\draw (1,3) node [above] {$\bar b$}--(1,2.5) .. controls (1,2) ..
      (0,1) node [above] {$t$} .. controls  (-1,2)  .. (-1,2.5)--(-1,3)
		node [above] {$a$};
    \draw[thick,-round cap] (2,-3) node [below] {$\rho$}
      ..controls (2,-2)..(1.1,-1.1)--(.15,-.15);
    \draw[thick,round cap-]  (-.15,.15)--(-1.1,1.1)..controls
      (-2,2)..(-2,3) node [above] {$\rho$};
		\mydot{(0,1) };
	}
  \,;&
  \tikzmatht{
    \useasboundingbox (-3,-3.5) rectangle (3,3.5);
	  \fill[\coN] (-3,-3) rectangle (3,3);
		\fill[\coM] (-3,3)--(-3,-3)--(0,-3)--(0,-2.5).. controls 
       (-1,-1.5)  .. (-1,-1)--(-1,3);		
    \draw[thick] (1,3) node [above] {$\lambda$}--(1,-1) .. controls (1,-1.5) ..(0,-2.5);
    \draw (0,-3) node [below] {$\bar a$} -- (0,-2.5) node [above] {$Y$} 
      .. controls  (-1,-1.5)  .. (-1,-1)--(-1,3) node [above] {$\bar b$};
    \draw[thick, -round cap] (2,-3) node [below] {$\rho$}
      ..controls (2,-2).. (1.1,-1.1);
    \draw[thick, line cap=round] (.9,-.9)--(-.9,.9);
    \draw[ultra thick,round cap-] (-1.15,1.15)..controls (-2,2)..(-2,3) node 
      [above] {$\alpha^-_\rho$};
		\mydot{(0,-2.5) };
	}
  &=
  \tikzmatht{
    \useasboundingbox (-3,-3.5) rectangle (3,3.5);
	  \fill[\coN] (-3,-3) rectangle (3,3);
		\fill[\coM] (-3,3)--(-3,-3)--(0,-3)--(0,1).. controls 
       (-1,2)  .. (-1,2.5)--(-1,3);		
    \draw[thick] (1,3) node [above] {$\lambda$}--(1,2.5) .. controls (1,2) ..(0,1);
    \draw (0,-3) node [below] {$\bar a$} -- (0,1) node [above] {$Y$} 
      .. controls  (-1,2)  .. (-1,2.5)--(-1,3) node [above] {$\bar b$};
    \draw[thick, -round cap] (2,-3) node [below] {$\rho$}
      ..controls (2,-2).. (1.1,-1.1)--(.1,-.1);
    \draw[ultra thick,round cap-] (-.15,.15)..controls (-2,2)..(-2,3) node
      [above] {$\alpha^-_\rho$};
		\mydot{(0,1) };
	}
  \punkt
\end{align*}
For details we refer to \cite[Sect.~3.3]{BcEvKa1999}.

There is a \textbf{relative braiding} \cite[p.\ 738]{BcEvKa2000}
\begin{align}
 \cE_\mathrm{r}(\beta_+,\beta_-)&:=S^\ast\alpha_{\mu}(T^\ast)\varepsilon(\lambda,\mu)\alpha^+_\lambda(S)T
	\in \Hom(\beta_+\beta_-,\beta_+\beta_-)\komma
\end{align}
where for fixed $\beta_\pm\in\MMs^\pm$, we choose $\lambda,\mu\in \NNs$, such that $\beta_+\prec \alpha^+_\lambda$, $\beta_-\prec\alpha^-_\nu$and isometries $S,T$, such that $T\in\Hom(\beta_+,\alpha^+\mu)$
and $S\in\Hom(\beta_-,\alpha^-_\mu)$. The definition is independent of the particular choice of  $\lambda,\mu,S,T$.

The relative braidings give a non-degenerate braiding $\varepsilon(\slot,\slot):=\cE_\mathrm{r}(\slot,\slot)$ on $\MMs^0$ by \cite[Sec.\ 4]{BcEvKa2000}, so in particular
$\MMs^0$ becomes a UMTC.

In general for two braided subfactors $\iota_a(\N)\subset\M_a$ and
$\iota_b(\N)\subset \M_b$ in $\NNs$ we define $\s{\M_a}{\M_b}$ as a full
subcategory of $\Mor(\M_b,\M_a)$ with subobjects and direct sums generated by
$\iota_a\NNs\bar \iota_b$.

\section{Morita equivalence for braided subfactors}
\label{sec:Morita}
\subsection{Module categories, modules and bimodules} 
In this section we give the notion of Morita equivalent non-degenerately braided
subfactors. 

We adapt the following definitions from \cite{Os2003}.
\begin{defi}
  A (strict) \textbf{module category} over a tensor category $\cC$ is a category
  $\cM$ together with an exact bifunctor $\otimes\colon\cC\times \cM\to\cM$ such that
  $(X\otimes Y)\otimes M = X\otimes(Y\otimes M)$ for all $X,Y\in\cC$ and
  $M\in\cM$.

  Let $\cM_1,\cM_2$ be two module categories over $\cC$.  A (strict)
  \textbf{module functor} from $\cM_1$ to $\cM_2$ is a functor $F\colon\cM_1 \to 
  \cM_2$ such that $F(X\otimes M)=X\otimes F(M)$.

  Two module categories $\cM_1$ and $\cM_2$ over $\cC$ are called
  \textbf{isomorphic} if there exist a module functor, which is an isomorphism
  of categories.
\end{defi}
Let $\NNs \subset \End(\N)$ be a UFC and let
$\Theta=(\theta,w,x)$ be a Q-system in $\NNs$ corresponding to $\N\subset \M$.
A (right) $\Theta$-module (\cf \cite{EvPi2003}) is a pair $(\rho,r)$ with
$\rho \in \NNs$ and $\tilde r \in \Hom(\rho\circ\theta,\rho)$, such that 
$r^\ast$ is an isometry and $\tilde r=\sqrt[4]{d\theta}\, r$ satisfies
\begin{align*}
  \tilde r\cdot (1_\rho\otimes m) 
    &= \tilde r \cdot (\tilde e \otimes 1_\theta ) &\Leftrightarrow&& \tilde
r\cdot \rho(m)
    &= \tilde \rho(\tilde r^2)\\ 
      \tilde r\cdot (1_\rho\otimes r)
    &=1_\rho &\Leftrightarrow&& \tilde r\cdot \rho(e)
    &=1_\rho 
\end{align*} 
where $m = \sqrt[4]{d\theta} x^\ast$ the
multiplication and $e=\sqrt[4]{d\theta} w$ the unit of 
the (Frobenius) algebra object corresponding to $\Theta$.  
Graphically this means:
\begin{align*}
  \tikzmatht{
    \draw[thick] (0,-4) node [below] {$\rho$}--(0,1) node [above] {$\rho$};
    \draw[thick] (3,-4) node [below] {$\theta$}--(3,-3.5) arc (0:180:1)--
      (1,-4) node [below] {$\theta$};
    \draw[thick] (2,-2.5) .. controls (2,-2) .. (0,0);
    \mydot{(0,0)};
    \node at (0,0) [left] {$r$};
    \mydot{(2,-2.5)};
    \node at (2,-2.5) [above right] {$x^\ast$};
  }
  &=
  \tikzmatht{
    \draw[thick] (0,-4) node [below] {$\rho$}--(0,1) node [above] {$\rho$};
    \draw[thick] (2,-4) node [below] {$\theta$}--(2,-3.5)--
      (2,-2.5) .. controls (2,-2) .. (0,0);
    \draw[thick] (1,-4) node [below] {$\theta$}--(1,-3.5)..controls (1,-3)..
      (0,-2);
    \mydot{(0,0)};
    \node at (0,0) [left] {$r$};
    \mydot{(0,-2)};
    \node at (0,-2) [left] {$r$};
  }\,;
  &
  \tikzmatht{
    \draw[thick] (0,-3) node [below] {$\rho$}--(0,2) node [above] {$\rho$};
    \draw[thick] (1,-2)--(1,-1.5) .. controls (1,-1) .. (0,0);
    \mydot{(0,0)};
    \node[thick] at (0,0) [left] {$r$};
    \mydot{(1,-2)}
    \node at (1,-2) [right] {$w$};
  }
  &=
  \tikzmatht{
    \draw[thick] (0,-3) node [below] {$\rho$}--(0,2) node [above] {$\rho$};
  }
  \punkt
\end{align*}

A left $\Theta$-module can be defined similarly. We note that because we are
working in C$^\ast$-categories and ask $r^\ast$ to be an isometry, 
that a module is also a co-module by the action
$r^\ast$.
The endomorphism $\rho\theta$ with $\rho\in\NNs$ has the structure of a right
$\Theta$-module, where the action is given by $\tilde r=1_\rho\otimes m \equiv
\rho(m)\equiv \sqrt[4]{d\theta}\cdot
\rho(x^\ast)\in\Hom(\rho\theta\theta,\rho\theta)$ in other words
$r=\rho(x^\ast)$, graphically:
\begin{align*}
  \tikzmatht{
    \draw[thick] (0,-4) node [below] {$\rho\theta$}--(0,0) node [above] {$\rho\theta$};
    \draw[thick] (1,-4) node [below] {$\theta$}--(1,-3.5)..controls (1,-3)..
      (0,-2);
    \mydot{(0,-2)};
    \node at (0,-2) [left] {$r$};
  }
  :=
  \tikzmatht{%
    \draw[thick] (0,-4) node [below] {$\rho$}--(0,0) node [above] {$\rho$};
    \draw[thick] (3,-4) node [below] {$\theta$}--(3,-3.5) arc (0:180:1)--
      (1,-4) node [below] {$\theta$};
    \draw[thick] (2,-2.5) .. controls (2,-1) and (1,-1.5).. (1,0)
        node [above] {$\theta$};
    \mydot{(2,-2.5)};
    \node at (2,-2.5) [above right] {$x^\ast$};
  }
  \punkt
\end{align*}
It is called the 
  \textbf{induced module}.
Any irreducible right $\Theta$-module is equivalent to a submodule of an induced module \cf
\cite{Os2003}.

The $\Theta$-modules form a category with $\Hom_\Theta(\rho,\sigma)\equiv
\Hom_\Theta((\rho,r),(\sigma,s))=\{
t\in\Hom(\rho,\sigma) : tr=st\}$, so the arrows are arrows of the objects which
intertwine the actions. There is a correspondence between projections $p\in
\Hom_\Theta(\rho,\rho)$ and submodules, namely we can choose $\rho_p$ and
$t\in\Hom(\rho_p,\rho)$ with $t^\ast t=1_{\rho_p}$, $tt^\ast=p$ and define
$r_p=t^\ast r t$.

Let $\Theta_a=(\theta_a,w_a,x_a)$ and $\Theta_b=(\theta_b,w_b,x_b)$ be two
Q-systems in $\NNs$. A $\Theta_a$-$\Theta_b$ bimodule is a triple
$(\rho,r_a,r_b)$ with $\rho\in\NNs$ and $\rho_a\in \Hom(\theta_a\rho,\rho)$ and
$\rho_b\in \Hom(\rho\theta_b,\rho)$, such that $(\rho,r_a)$ is a left
$\Theta_a$-module and $(\rho,r_b)$ is a (right) $\Theta_b$-module and which
commute, \ie $$r_a\cdot \theta_a(r_b) = r_b\cdot r_a.$$ We can define:
\begin{align*}
	r:=r_a\cdot (1_{\theta_a}\otimes r_b) &= r_b\cdot (r_a \otimes 1_{\theta_a})
	\in (\theta_a\circ\rho\circ\theta_b,\rho).
\end{align*}
Let $\rho=(\rho,r_a,r_b)$ and $\sigma=(\sigma,s_a,s_b)$ be two
$\Theta_a$--$\Theta_b$ bimodules.  An intertwiner $t\colon\rho \to \sigma$ is an
$\Theta_a$--$\Theta_b$ bimodule intertwiner, if $t$ intertwines the actions $r$ and $s$, \ie 
\begin{align*}
	 tr &=s (1_{\theta_a}\otimes t \otimes 1_{\theta_b})\equiv s\theta_a(t)
\punkt
\end{align*}
Let us denote by $\Bim(\Theta_a,\Theta_b)$ the category 
of bimodules with $\Hom_{\Theta_a-\Theta_b}(\rho,\sigma)$
$\Theta_a$-$\Theta_b$ bimodule intertwiner.
We note that one can give Q-systems, bimodules and intertwiners
the structure of a bicategory, by introducing a relative tensor product between
bimodules.

We set $\Mod(\Theta)=\Bim(1,\Theta)$ to be the category of 
(right) $\Theta$-modules.

The category $\Mod(\Theta)$ has a natural structure of a (strict)
left $\NNs$ module category, where the functor 
$\NNs\times \Mod(\Theta)$ is given by $(\mu, \rho) \mapsto \mu\rho$
where $\mu\rho$ is a right-module with $r_{\mu\sigma}=\mu(r_\rho)$
and $\Hom_{\Mod(\Theta)}(\rho,\sigma) \ni T\mapsto 
\mu(T)\in\Hom_{\Mod(\Theta)}(\mu\rho,\mu\sigma)$.

\begin{prop}[{\cite[Lemma 3.1.]{EvPi2003}}] 
	\label{prop:EquivCat}
        Let $\NNs$ be a UMTC and $\Theta_a,\Theta_b$ irreducible Q-systems in
$\NNs$.
	The category of $\Theta_a$-$\Theta_b$ bimodules is equivalent to the  
  category $\s{\M_a}{\M_b}$. The functor $\Phi$ maps $\beta \in \s{\M_a}{\M_b}$ to
  $\bar \iota_a\circ\beta\circ\iota_b$ and $t\in \Hom(\beta,\beta')$ to 
  $\bar \iota_a(t)\in \Hom_{\Theta_a\text-\Theta_b}(\Phi(\beta),\Phi(\beta'))$.
\end{prop}
\begin{proof}
In \cite[Lemma 3.1.]{EvPi2003} is shown that the functor $\Phi$ is fully
faithful. It is also shown that is is essentially surjective, so it gives an
equivalence of categories.
\end{proof}
The functor $\Phi$ is graphically given as follows, where $\rho=\Phi(\beta)$
$\tilde r\in\Hom(\theta_a\rho\theta_b,\rho)$ the action:
$$
    \Phi\colon
    \tikzmatht[.42]{
      \fill[\coM] (-1,-1) rectangle (0,1);
      \fill[\coP] (0,-1) rectangle (1,1);
      \draw[ultra thick] (0,-1) node [below] {$\beta$}--(0,1) node [above]
        {$\beta'$};
      \fill[white] (-.5,-.5) rectangle (.5,.5);
      \draw (-.5,-.5) rectangle node {$t$} (.5,.5);
    }
    \mapsto
    \tikzmatht[.42]{
      \fill[\coM] (-1,-1) rectangle (0,1);
      \fill[\coP] (0,-1) rectangle (1,1);
      \draw[ultra thick] (0,-1) node [below] {$\beta$}--(0,1) node [above]
        {$\beta'$};
      \draw (-1,-1) node [below] {$\bar\iota_a$}--(-1,1) node 
        [above] {$\bar\iota_a$};
      \draw (1,-1) node [below] {$\iota_b$}--(1,1) node 
        [above] {$\iota_b$};
      \fill[white] (-.5,-.5) rectangle (.5,.5);
      \draw (-.5,-.5) rectangle node {$t$} (.5,.5);
  }
  \komma
  \qquad 
  \tilde r=
  \tikzmatht[.42]{
    \draw[thick] (0,-1) node [below] {$\rho$}--(0,1) node [above] 
      {$\rho$};
    \draw[thick] (-1,-1) node [below] {$\theta_a$} .. controls (-1,-.5) .. 
      (0,0);
    \draw[thick] (1,-1) node [below] {$\theta_b$} .. controls (1,-.5) .. 
      (0,0);
    \mydot{(0,0)};
  }
  :=
    \tikzmatht[.42]{
      \fill[\coP] (0,1)--(0,-1)--(3,-1) arc (0:90:1.5)
        .. controls +(-.5,0)..(1,1);
      \fill[\coM] (0,1)--(0,-1)--(-3,-1) arc (180:90:1.5)
        .. controls +(.5,0)..(-1,1);
      \fill[\coN] (1,-1) arc (180:0:.5);
      \fill[\coN] (-1,-1) arc (0:180:.5);
      \draw[ultra thick] (0,-1) node [below] {$\beta$}--(0,1) node [above]
        {$\beta$};
      \draw (-3,-1) node[below] {$\bar \iota_a$}arc (180:90:1.5).. 
        controls +(.5,0)..  (-1,1) node [above] {$\bar\iota_a$};
      \draw (3,-1) node[below] {$\iota_b$}arc (0:90:1.5).. controls +(-.5,0)..
        (1,1) node [above] {$\iota_b$};
      \draw (1,-1) node [below] {$\iota_b$} arc (180:0:.5) node [below]
         {$\bar \iota_b$};
      \draw (-1,-1) node [below] {$\bar\iota_a$} arc (0:180:.5) node [below]
         {$ \iota_a$};
    }
  \punkt
$$

\begin{rmk}
        \label{rmk:alphaInduction}
Let $\Theta=(\theta,w,x)$ be a Q-system in a UMTC $\NNs$ with corresponding
subfactor $\iota(N)\subset M$.
The bimodule $\Phi(\alpha^\pm_\lambda)\equiv\bar\iota\alpha^\pm_\lambda\iota\equiv
\bar\iota\iota\lambda$ is the object
$\theta\lambda$ with left
action the induced action $x^\ast$ and right action by
$x^\ast\varepsilon^\pm(\lambda,\theta)$, namely for the $+$-case: 
$$
  \tilde r=
  \tikzmatht[.42]{
    \draw[thick] (0,-1) node [below] {$\theta$}--(0,1) node [above] 
      {$\theta$};
    \draw[thick] (-1,-1) node [below] {$\theta$} .. controls (-1,-.5) .. 
      (0,0);
    \draw[thick] (1,-1) node [below] {$\theta$} .. controls (1,-.5) .. 
      (0,0);
    \mydot{(0,0)};
    \draw[ultra thick,double, white] (0.5,-1) --(0.5,1);  
    \draw[thick] (0.5,-1) node [below] {$\lambda$}--(0.5,1) node [above] 
      {$\lambda$};
  }
  =
    \tikzmatht[.42]{
      \fill[\coM] (0,1)--(0,-1)--(3,-1) arc (0:90:1.5)
        .. controls +(-.5,0)..(1,1);
      \fill[\coM] (0,1)--(0,-1)--(-3,-1) arc (180:90:1.5)
        .. controls +(.5,0)..(-1,1);
      \fill[\coN] (1,-1) arc (180:0:.5);
      \fill[\coN] (-1,-1) arc (0:180:.5);
      \draw[ultra thick] (0,-1) node [below] {$\alpha^+_\lambda$}--(0,1) node [above]
        {$\alpha^+_\lambda$};
      \draw (-3,-1) node[below] {$\bar \iota$}arc (180:90:1.5).. 
        controls +(.5,0)..  (-1,1) node [above] {$\bar\iota$};
      \draw (3,-1) node[below] {$\iota$}arc (0:90:1.5).. controls +(-.5,0)..
        (1,1) node [above] {$\iota$};
      \draw (1,-1) node [below] {$\iota$} arc (180:0:.5) node [below]
         {$\bar \iota$};
      \draw (-1,-1) node [below] {$\bar\iota$} arc (0:180:.5) node [below]
         {$ \iota$};
    }
\komma
$$
where equality can
be seen easily using $\iota\lambda = \alpha^+_\lambda\iota$,
$\Theta=\Theta_{\bar\iota}$ and the IBFEs by pulling the $\lambda$-string
between $\bar \iota$ and $\iota$. The $-$-case works analogous using the opposite
        braiding.
The obtained bimodules coincide with the notion of
$\alpha$-induction in the categorical literature.
\end{rmk}
The category $\Bim(\Theta,\Theta)$ becomes a tensor category, where
$\rho\otimes_\Theta\sigma$ is the object associated to the projection in
$P_{\rho\otimes_\Theta\sigma}\in\Hom(\rho\sigma,\rho\sigma)$ given by:

$$
P_{\rho\otimes_\Theta\sigma}
  =\frac1{\sqrt{d\theta}}
  \tikzmatht{
    \draw[thick] (-2,1)--(-2,-3) node [below] {$\rho$};
    \draw[thick] (2,1)--(2,-3) node [below] {$\sigma$};
    \draw[thick] (-1,-1) arc (180:360:1)  .. controls (1,-.5) .. 
      (2,0);
    \draw[thick] (-1,-1) .. controls (-1,-.5) .. 
      (-2,0);
    \mydot{(2,0)};
    \mydot{(-2,0)};
    \mydot{(0,-2)};
    \draw[thick] (0,-2)--(0,-2.5);
    \mydot{(0,-2.5)};
  }
\punkt
$$
and it is easy to check that $\Phi$ is a tensor functor. Thus,
$\Bim(\Theta,\Theta)$ and $\MMs$ are equivalent as tensor categories.
We note that this category is non-strict. We can define the categories 
$\Bim^\pm(\Theta,\Theta)$ to be the image of $\MMs^\pm$ under $\Phi$ and
$\Bim^0(\Theta,\Theta)=\Bim^+(\Theta,\Theta)\cap\Bim^-(\Theta,\Theta)$.

In the special case
$\M_a=\N$ and $\M_b=\M$ and $\theta_a=\theta$
we have an equivalence of the category $\s{\N}{\M}$ and the category $\Mod(\Theta)$ of right $\Theta$-modules 
given by $\bar a \mapsto \bar a \iota$.
The category of right $\Theta$-modules $\Mod(\Theta)$ 
becomes a module category over $\NNs$ using the monoidal structure
inherent from $\End(\N)$. The same is true for $\NMs$.

In particular, it follows:
\begin{prop}
  \label{prop:IsoCats}
	Let $\NNs\subset \End(\N)$ be a UMTC and $\Theta$ be a Q-system in $\NNs$
	with corresponding subfactor $\N\subset \M$.
	Then  $\Mod(\Theta)$ and $\NMs$ are equivalent as module categories.
\end{prop}
\begin{proof}
It follows directly from the properties of the monoidal structure,
 that 
the functor $\Phi$
(in the case of $\M_a=\N$ and $\M_b=\M$ and $\theta_a=\theta$)
 in the proof of Prop.\ \ref{prop:EquivCat}
is a module functor, so in particular a module isomorphism,
between the two module categories $\Mod(\Theta)$ and $\NMs$ over $\NNs$.
\end{proof}

We remark that in general in the definition of module it is not assumed that $r$ is
a (multiple) of an isometry, because the existence of a unitary structure is not
assumed. But since every module in the general sense is equivalent to a
submodule of an induced module and the submodule can chosen to have a multiple
of an isometry as action, we can  
without lost of generality restrict to modules where $r$ is a
multiple of an isometry. This can also be shown directly \cite{BiKaLoRe2014-2}.

Let $a\in \NMs$ be irreducible and consider the subfactor $\N\subset \M_a$ given
by the Q-system $\Theta_a$ (see Def.\ \ref{defi:TrivialQSystem}).
Let $\M_a$ be the factor which is given by Jones basic construction 
$a(\M)\subset \N\subset \M_a$ 
and denote the inclusion map $\iota_a\colon\N\hookrightarrow \M_a$. Because the
subfactors $\bar\iota_a(\M_a)\subset \N$ and $a(\M)\subset \N$ have by
definition the same Q-system and thus are conjugated by a unitary in $\N$, we
may and do choose $\bar\iota_a\colon\M_a \to \N$, such that $\bar\iota_a(\M_a)= a(\M)$.
This implies that $\alpha=\bar\iota_a^{-1}\circ a\colon\M\to \M_a$ is an isomorphism
with conjugate $\alpha^{-1}=a^{-1}\circ\bar\iota_a\colon\M_a\to\M$.
\begin{lem}[\cf \cite{LoRe2004,Ev2002}]
	\label{lem:DHROrbit}
	Let $\NNs\subset \End(\N)$ be a UMTC and $\Theta$ be a Q-system in $\NNs$
	with corresponding subfactor $\N\subset \M$.

	For $a\in\NMs$ irreducible let $\Theta_a$ be the canonical Q-system
	$(\Theta_a=a\bar a,w_a,x_a)$ and $\N\subset \M_a$ the corresponding 
	subfactor.
	Then $\NMs$ and $\s{\N}{\M_a}$ 
	are isomorphic as module categories of $\NNs$.
	The isomorphism is given by 
	$\Psi\colon b \mapsto b\circ a^{-1}\circ\iota_a$ and 
	$\Hom_{\s{\N}{\M}}(b,c) \ni t \mapsto t \in \Hom_{\s{\N}{\M_a}}(\Psi(b),\Psi(c))$. 
\end{lem}

\begin{rmk}
  \label{rmk:Inductive}
  Given $a\in\NMs$ we have the Q-systen $\Theta_a$ with
  $\theta_a=a \bar a$.
  Let $\beta=\Phi(a)\in\Mod(\Theta)$, then $\bar\beta$ is a $\Theta$ left 
  module and there is another way to construct a Q-system \cite{KoRu2008}
  denoted by $\bar \beta\otimes_\Theta \beta$, and it is 
  easy to check that $\bar \beta \otimes_\Theta \beta \cong \bar a a$ 
  and that the obtained Q-systems are equivalent.
\end{rmk}

\subsection{The Morita equivalence class of a braided subfactor}
\label{sec:MoritaEquivalenceClass}
In the following we use the definition of 
Morita equivalence for module categories as in \cite[Def.\ 3.3]{Os2003}.
  Let $\NNs\subset \End(\N)$ be a UMTC. We remember that we
	call a pair 
	$(\N\subset \M,\NNs)$ where $\N\subset \M$ is a subfactor
	whose Q-system $\Theta$ is in $\NNs$ a 
non-degenerately braided subfactor. 
\begin{defi} 
  \label{defi:MoritaEquivalence}
	Let $\NNs\subset \End(\N)$ be a UMTC.
	Two irreducible Q-systems $\Theta_a$ and $\Theta_b$ in $\NNs$ are called 
	\textbf{Morita equivalent} if 
	one of the following equivalent statements hold:
	\begin{itemize}
		\item $\Mod(\Theta_a)$ and $\Mod(\Theta_b)$ are equivalent
		as module categories over $\NNs$.
		\item $\s{\N}{\M_a}$ and $\s{\N}{\M_b}$ are 
		equivalent as module categories over $\NNs$,
		where $\N\subset \M_\bullet$ is corresponding
		to	$\Theta_\bullet$.
	\end{itemize}
	We say that the subfactors $\N\subset\M_a$ and 
	$\N\subset\M_b$ are Morita equivalent if
	their Q-systems $\Theta_a$ and $\Theta_b$, respectively, are
	Morita equivalent.
\end{defi}
Let $(\iota(\N)\subset \M,\NNs)$ be a non-degenerately braided subfactor. 
It follows directly that for $a,b\in\NMs$ irreducible
$\Theta_a$ and $\Theta_b$ are Morita equivalent and in particular 
are Morita equivalent to $\Theta_{\bar\iota}$. 
But it can also happen that $\Theta_a$ and $\Theta_b$ are equivalent for $[a]\neq [b]$.
If $\cC$ is a UTFC, we denote by $\Pic(\cC)$ the full and replete 
subcategory (2-group) with objects 
$\{\rho\in\C: d\rho=1\}$ (not completed under direct sums).
\begin{prop}[\cite{GrSn2012}]
	\label{prop:Picard}
  Given two irreducible objects $a,b\in \NMs$. 
  Then the Q-systems $\Theta_a$ and $\Theta_b$ are equivalent 
  if and only if there is an automorphism $\beta\in\Pic(\MMs)$ such 
  that $b\beta=a$.
\end{prop}
Now we can give a characterization of the Morita equivalence class 
of a non-degenerately braided subfactor.
\begin{prop}
  \label{prop:ClassDHROrbit}
	Let $\NNs\subset \End(\N)$ be a UMTC and let 
	$\Theta$ be a Q-system in $\NNs$.
	Then there is a one-to-one correspondence between 
	\begin{enumerate}
		\item equivalence classes $[\Theta_a]$ of irreducible 
		Q-systems Morita equivalent to $\Theta$,
		\item irreducible sectors $[a]$ with $a\in\NMs$ up 
		the identification:
		$[a]\sim [b]$ if there is an automorphism $\beta\in{}_\M\mathcal X_{\M}$,
		such that $[a]=[\beta b]$,
	\item elements in $\NDM/\Pic(\MMs)$. 
	\end{enumerate}
\end{prop}
\begin{proof}
	Statement (3) is just a reformulation of (2).
	Let $a\in{}_\N\mathcal X_{\M}$ then we obtain a canonical Q-system $\Theta_a$ in 
  $\NNs$ which is Morita equivalent to $\Theta$ by Lemma \ref{lem:DHROrbit}. 
  Conversely given a Q-system $\Theta_a$ Morita equivalent to $\Theta$ then
  $\NMs$ is equivalent to $\s{\N}{\M_a}$.
  The element  $a\in \NMs$ corresponding 
  to $\iota_a \in\s{\N}{\M_a}$ under this equivalence 
  is the corresponding element in $\NMs$, \cf 
\cite[Remark 3.5]{Os2003}. 
  The rest follows by Prop.\ \ref{prop:Picard}.
\end{proof}

\section{\texorpdfstring{$\alpha$}{alpha}-induction construction and the full center}
\label{sec:FullCentre}
\subsection{The full center and Rehren's construction coincide}
Let $\N$ be a type III factor and $\NNs\subset \End(\N)$ a UMTC.
As before let $\NDN=\{\id_\N,\rho_1,\ldots,\rho_n\}$ a set of 
representatives for each sector.

Given $\nu,\lambda,\mu\in\NDN$, we can choose a set of isometries
$B(\nu,\lambda\mu):=\{e_i\}_{i=1,\ldots,\langle\nu,\lambda\mu\rangle}$ with $ e_i \in\Hom_{\NNs}(\nu,\lambda\mu)$, such that $\{e_i\}$ form an orthonormal
basis with respect to the scalar product $(e,f)=\Phi_\nu(e^\ast f)$ defined by
the left inverse $\Phi_\nu$ of $\nu$ \cite{LoRo1997} or equivalently defined by
$(e,f)\cdot 1_\nu = e^\ast f$.
We define for an isometry $e\in\Hom_{\NNs}(\nu,\lambda\mu)$ 
an isometry $\bar e \in \Hom_{\overline{\NNs}}(\bar\nu,\bar\lambda\bar\mu)$ by  
$$
    \tikzmatht{
	    \draw[thick] (0,-2) node [below] {$\bar\nu$}--(0,0) .. controls (1,1) ..
         (1,1.5) --(1,2.5) node [above] {$\bar\mu$};
	  	\draw[thick] (0,0) node [right] {$\bar e$} .. controls (-1,1) .. 
        (-1,1.5) --(-1,2.5)node [above] {$\bar\lambda$};
  		\mydot{(0,0)};
	  }
    :=
    \tikzmatht{
	  	\draw[thick] (0,0) node [left] {$e^\ast$} .. controls (-1,-1) .. 
        (-1,-1.5) arc(180:360:1) -- (1,1.5) node [above] {$\bar\lambda$};
	    \draw[ultra thick,white,double] (0.2,-0.2).. controls (1,-1.5) .. (2,-1.5)
        arc (270:360:1);
	    \draw[thick] (-2,-3) node[below]{$\bar\nu$} 
          --(-2,0) arc (180:0:1) .. controls (1,-1.5) ..
         (2,-1.5) arc (270:360:1) --(3,1.5) node [above] {$\bar\mu$};
  		\mydot{(0,0)};
	  }
  \punkt
$$
\begin{defi}[Longo--Rehren construction]
  \label{defi:LR}
  Let $\NNs\subset \End(\N)$ a URFC.
  There is a Q-system 
  $\Theta_\LR=(\theta_\LR,w_\LR,x_\LR)$ in $\NNs\boxtimes \overline{\NNs}$ 
  given by:
  \begin{align*}
  	[\theta_\LR]&=\bigoplus_{\rho\in\NNs} [\rho\boxtimes \bar\rho],
  	& 
  	x_\LR&= \frac1{\sqrt{d\theta}}\bigoplus_{\lambda\mu\nu}\sum_{e\in B(\nu,\lambda\mu)}
  	\sqrt{\frac{d\lambda d\mu}{d\nu d\theta}} \, e\boxtimes \bar e
    \komma
    \\&&&=\bigoplus_{\lambda\mu\nu} \sum_{e\in B(\nu,\lambda\mu)}
    \tikzmatht{
	    \draw[thick] (0,-1) node [below] {$\nu$}--(0,0) .. controls (1,1) ..
         (1,1.5) node [above] {$\mu$};
  		\draw[thick] (0,0) node [right] {$e$} .. controls (-1,1) .. 
        (-1,1.5) node [above] {$\lambda$};
		  \mydot{(0,0)};
  	}
    \boxtimes
    \tikzmatht{
	    \draw[thick] (0,-1) node [below] {$\bar\nu$}--(0,0) .. controls (1,1) ..
         (1,1.5) node [above] {$\bar\mu$};
	  	\draw[thick] (0,0) node [right] {$\bar e$} .. controls (-1,1) .. 
        (-1,1.5) node [above] {$\bar\lambda$};
  		\mydot{(0,0)};
	  }
    \punkt
\intertext{%
  More general, for an equivalence of braided categories $\phi\colon \NNs \to \NNs'$,
  we define the Q-system
  $\Theta^\phi_\LR=(\theta_\LR^\phi,w_\LR^\phi,x^\phi_\LR)$ in 
  $\NNs\boxtimes \overline{\NNs'}$ by
  }
  	[\theta_\LR^\phi]&=\bigoplus_{\rho\in\NNs} [\rho\boxtimes \phi(\bar\rho)],
  	& 
  	x_\LR^\phi&= \bigoplus_{\lambda\mu\nu}\sum_{e\in B(\nu,\lambda\mu)}
  	\sqrt{\frac{d\lambda d\mu}{d\nu d\theta}} e\boxtimes \overline{\phi(e)}
        \punkt
\end{align*}
\end{defi}
\begin{defi} 
  \label{defi:LocalQSystem}
  Let $\NNs\subset \End(\N)$ be a URFC.
  A Q-system $\Theta=(\theta,w,x)$ in $\NNs$ is called \textbf{\local}
if $\varepsilon(\theta,\theta)x=x$. Diagrammatically:
\begin{align*}		
	\tikzmatht{
	\draw[thick] (1,4.5) node [above] {$\theta$}
	--(1,1.5) arc (360:180:1)--(-1,4.5) node [above] {$\theta$};
	\draw[thick] (0,-0.5) node [below] {$\theta$}--(0,0.5) node [below right] {$$};
	\mydot{(0,0.5)};
	}
	=
	\tikzmatht{
	\draw[thick]  (-1,4.5) node [above] {$\theta$} 
	.. controls (-1,3) and (1,3) .. (1,1.5) 
	arc (360:180:1)
	.. controls (-1,3) and (1,3) ..  (1,4.5) node [above] {$\theta$};
	\draw[thick] (0,-0.5) node [below] {$\theta$}--(0,0.5) node [below right] {$$};
	\draw [double,ultra thick, white] (-1,1.5) .. controls (-1,3) and (1,3) ..  (1,4.5) node [above] {$\theta$};
	\mydot{(0,0.5)};
	\draw [thick] (-1,1.5) .. controls (-1,3) and (1,3) ..	(1,4.5) node [above] {$\theta$};
	\mydot{(0,0.5)};
	}
        \punkt
\end{align*}
\end{defi}
\begin{prop}[\cite{LoRe1995}] 
  \label{prop:LRlocal}
  The Q-system obtained by the Longo--Rehren construction
  is \local.
\end{prop}

\begin{defi}[Product Q-system]
  \label{defi:ProductQSystem}
  Let $\Theta_i=(\theta_i,w_i,x_i)$ with $i=1,2$ be two Q-systems in a URFC
  category $\NNs$. 
  Then we define two Q-systems 
  $\Theta_1\circ^\pm\Theta_2=(\theta_1\circ\theta_2,w_1w_2,x_\pm)$ in $\NNs$, 
  where $x_\pm= \theta_1(\varepsilon^\pm(\theta_1,\theta_2))x_1\theta_1(x_2)$,
  graphically:
  \begin{align*}
  \tikzmatht[0.52]{
    \draw[thick] (-1,1.5) node[above]{$\theta_1\theta_2$}--(-1,1) arc (180:360:1)--(1,1.5)
      node[above]{$\theta_1\theta_2$} (0,0)--(0,-1) node[below]{$\theta_1\theta_2$};
    \mydot{(0,0)};
    \node at (0,0) [below left] {$x_+$};
  }
  =
  \tikzmatht[0.52]{
    \begin{scope}[shift={(1,0)}]
    \draw[thick] (-1,1.5) node[above]{$\theta_2$}--(-1,1) arc (180:360:1)--(1,1.5)
      node[above]{$\theta_2$} (0,0)--(0,-1) node[below]{$\theta_2$};
    \mydot{(0,0)};
    \node at (0,0) [below right] {$x_2$};
    \end{scope}
    \draw[double, ultra thick,white] (-1,1.5) node[above]{$\theta_1$}--(-1,1) arc (180:360:1)--(1,1.5)
      node[above]{$\theta_1$} (0,0)--(0,-1) node[below]{$\theta_1$};
    \draw[thick] (-1,1.5) node[above]{$\theta_1$}--(-1,1) arc (180:360:1)--(1,1.5)
      node[above]{$\theta_1$} (0,0)--(0,-1) node[below]{$\theta_1$};
    \mydot{(0,0)};
    \node at (0,0) [below left] {$x_1$};
  }
  \punkt
  \end{align*}
\end{defi}

\begin{defi}
  \label{defi:LocalProjector}
  For $\Theta\equiv(\theta,w,x)$ a Q-system in $\NNs$ and $\rho\in\NNs$, we define
  $$
  	P^\mathrm{l}_\Theta(\rho)=\frac1{\sqrt{d\theta}}\cdot
  	\tikzmatht{
    	\draw [thick] (0,4) node [above] {$\rho$}--(0,0.6) (0,0)--(0,-2) 
        node [below] {$\rho$};
  		\draw [thick](-1,-2) node [below] {$\theta$}--(-1,-1.414);
	  	\draw [thick](-1,4) node [above] {$\theta$}--(-1,2)   arc (90:180:.5)
		  	arc(0:-180:.5) arc(180:0:1.5) --+(0,0) arc(0:-45:1.707)
  			arc(135:180:1.707) arc (0:-180:0.5) arc (180:135:1.707)
		  	arc (-45:0:1.707) arc(0:90:0.5);
	  	\mydot{(-1,2)} ;
  		\mydot{(-1,-1.414)} ;
		  \draw[ultra thick,double,white] (-1,4)--(-1,2.5); 
	  	\draw [thick](-1,4)--(-1,2.5); 
  		\draw[ultra thick,double,white] (0,3)--(0,1); 
  		\draw [thick](0,3)--(0,1); 
  	}
        \equiv
  	\tikzmatht{
    	\draw [thick] (0,4) node [above] {$\rho$}--(0,-0.15) (0,-.6)--(0,-2) 
        node [below] {$\rho$};
  		\draw [thick](-1,-2) node [below] {$\theta$}--(-1,4)  node
[above] {$\theta$};
                \draw [thick] (-1,1) circle (1.707);
  		\mydot{(-1,-.707)} ;
		  \draw[ultra thick,double,white] (-1,4)--(-1,2.5); 
	  	\draw [thick](-1,4)--(-1,2.5); 
  		\draw[ultra thick,double,white] (0,3)--(0,1); 
  		\draw [thick](0,3)--(0,1); 
  	}
  	\in \Hom(\theta\rho,\theta\rho)
  $$
  and $P^\mathrm{l}_\Theta:=P^\mathrm{l}_\Theta(\id_\N)$. 
  Similarly, we define $P^\mathrm{r}_\Theta(\rho)$ and $P^\mathrm{r}_\Theta$ by interchanging the braiding 
  with the opposite braiding.
\end{defi}
\begin{lem} 
  \label{lem:Projection}
  $P^\mathrm{l/r}_\Theta(\rho)$ is a projection.
\end{lem}
\begin{proof} 
  That $P^\mathrm{l}_\Theta(\rho)^2=P^\mathrm{l}_\Theta(\rho)$ is proven as in
  \cite[Lemma 5.2]{FuRuSc2002}, see also \cite{BiKaLoRe2014-2}. We just remark that we have a prefactor due to 
  another normalization and that one can check that $P^\mathrm{l}_\Theta(\rho)$
  is selfadjoint.
\end{proof}
\begin{prop}[Sub-Q-system \cf \cite{BiKaLoRe2014-2}]
  \label{prop:SubQSystem} 
	Let $p\in\Hom(\theta,\theta)$ 
	be an orthogonal projection satisfying 
	$p\theta (p)xp=\theta(p)xp =pxp =p\theta(p)x$ and  $w^\ast p=w^\ast$.
	Let $\theta_p\prec\theta$ 
	corresponding to $p$, \ie there a isometry 
  $s\in\Hom(\theta_p,\theta)$,
	such that $s^\ast s=1_{\theta_p}$ and $s s^\ast  = p$.
	Then $\Theta_p=(\theta_p,w_p,x_p)$ with 
	\begin{align*}
		w_p&:=s^\ast w, & x_p&:=\sqrt{\frac{d\theta}{d\theta_p}}\cdot
     s^\ast\theta(s^\ast)x s
	\end{align*}
	is a Q-system.
\end{prop}
Graphically, the conditions are given by:
\begin{align*}
	\tikzmatht[.42]{
		\draw[thick] (1,3) node [above] {$\theta$}
		--(1,1.5) arc (360:180:1)--(-1,3) node [above] {$\theta$};
		\draw[thick] (0,-1.5) node [below] {$\theta$}--(0,0.5);
		\mydot{(0,0.5)};
		\fill[white](-0.5,1.5) rectangle (-1.5,2.5);
		\draw (-0.5,1.5) rectangle node{$p$} (-1.5,2.5);
		\fill[white](-0.5,0) rectangle (.5,-1);
		\draw (-.5,-0) rectangle node{$p$} (.5,-1);
		\fill[white](0.5,1.5) rectangle (1.5,2.5);
		\draw (0.5,1.5) rectangle node{$p$} (1.5,2.5);
	}
  &=
	\tikzmatht[.42]{
		\draw[thick] (1,3) node [above] {$\theta$}
		--(1,1.5) arc (360:180:1)--(-1,3) node [above] {$\theta$};
		\draw[thick] (0,-1.5) node [below] {$\theta$}--(0,0.5);
		\mydot{(0,0.5)};
		\fill[white](-0.5,0) rectangle (.5,-1);
		\draw (-.5,-0) rectangle node{$p$} (.5,-1);
		\fill[white](0.5,1.5) rectangle (1.5,2.5);
		\draw (0.5,1.5) rectangle node{$p$} (1.5,2.5);
	}
  =
	\tikzmatht[.42]{
		\draw[thick] (1,3) node [above] {$\theta$}
		--(1,1.5) arc (360:180:1)--(-1,3) node [above] {$\theta$};
		\draw[thick] (0,-1.5) node [below] {$\theta$}--(0,0.5);
		\mydot{(0,0.5)};
		\fill[white](-0.5,1.5) rectangle (-1.5,2.5);
		\draw (-0.5,1.5) rectangle node{$p$} (-1.5,2.5);
		\fill[white](0.5,1.5) rectangle (1.5,2.5);
		\draw (0.5,1.5) rectangle node{$p$} (1.5,2.5);
	}
  =
	\tikzmatht[.42]{
		\draw[thick] (1,3) node [above] {$\theta$}
		--(1,1.5) arc (360:180:1)--(-1,3) node [above] {$\theta$};
		\draw[thick] (0,-1.5) node [below] {$\theta$}--(0,0.5);
		\mydot{(0,0.5)};
		\fill[white](-0.5,1.5) rectangle (-1.5,2.5);
		\draw (-0.5,1.5) rectangle node{$p$} (-1.5,2.5);
		\fill[white](-0.5,0) rectangle (.5,-1);
		\draw (-.5,-0) rectangle node{$p$} (.5,-1);
	}
        \komma
        &
	\tikzmatht[.42]{
            \draw[thick] (0,-2) node [below] {$\theta$}--(0,0);
		\fill[white](-0.5,-0.5) rectangle (.5,-1.5);
		\draw (-.5,-0.5) rectangle node{$p$} (.5,-1.5);
		\mydot{(0,0)};
        }
  =
	\tikzmatht[.42]{
    \draw[thick] (0,-2) node [below] {$\theta$}--(0,0);
		\mydot{(0,0)};
  }
  \punkt
\end{align*}
\begin{rmk}
        \label{rmk:IntermediateSubfactor} 
        The notion of sub-Q-system $\Theta_p$ of $\Theta$ corresponds to
        the notion of intermediate subfactor $L$ with $\N\subset L \subset \M$ 
        where $\Theta$ is the
        dual canonical Q-system of $\N\subset\M$. Namely, the properties of 
        the sub-Q-system are just a reformulation of \cite[Corollary~3.10]{IzLoPo1998}.
        Namely, they consider subspaces $K_\rho\subset\Hom(\iota,\iota\rho)$ for
each $\rho \in\NDN$,
        which correspond to a projection $p\in\Hom(\theta,\theta)$ if we identify
the Hilbert spaces $\Hom(\rho,\theta)$ and $\Hom(\iota,\iota\rho)$ by Frobenius
reciprocity. 
\end{rmk}
\begin{rmk}[\cf \cite{BiKaLoRe2014-2}]
        \label{rmk:SubQSystem} 
        If one drops the condition $w^\ast p = w^\ast$ in Prop.\
\ref{prop:SubQSystem} then we obtain a more general
``sub'' Q-system $\Theta_p=(\theta_p,w_p,x_p)$ with 
	\begin{align*}
		w_p&:=\lambda^{-1}\cdot 
                s^\ast w, & x_p&:=\lambda\cdot \sqrt{\frac{d\theta}{d\theta_p}}\cdot
     s^\ast\theta(s^\ast)x s \,
	\end{align*}
	where $\lambda = \sqrt{ w^\ast p w}$. 
\end{rmk}
\begin{defi} 
  \label{defi:LeftCenter}
  We denote by $C_\mathrm{l}(\Theta)=(C_\mathrm{l}(\theta),C_\mathrm{l}(w),
  C_\mathrm{l}(x))$ the \textbf{left center} of $\Theta$, which is defined to be the sub-Q-sytem
  associated with the projection $P^\mathrm{l}_\Theta\in\Hom(\theta,\theta)$.
  Analogously, the \textbf{right center} $C_\mathrm{r}(\Theta)$ is defined
  using $P^\mathrm{r}_\Theta$.
\end{defi}

\begin{rmk}[{\cite[Lemma 2.30]{FrFuRuSc2006}}]
  \label{rmk:LRCenter}
  The Q-system $C_\mathrm{l/r}(\Theta)$ is a maximal \local sub-Q-system of
  $\Theta$.
\end{rmk}
\begin{rmk}
  \label{rmk:MaximalChiral}
  The intermediate factor $\N\subset \M_+\subset \M$ defined 
  in \cite{BcEv2000} 
  is given by the Q-system $C_\mathrm{l}(\Theta)$. Namely, the characterization of
  $P^\mathrm{l}_\Theta$ in \cite[Lemma 2.30]{FrFuRuSc2006} is the characterization in
  \cite[Lemma 4.1]{BcEv2000} in terms of subspaces $\H_\rho\subset
  \Hom(\iota,\iota\rho)$ of ``charged intertwiners''.
  Similarly,  $\N\subset \M_-\subset \M$ is given by 
  $C_r(\Theta)$.
\end{rmk}

\begin{defi}[\cf \cite{FjFuRuSc2008}] 
  \label{defi:FullCenter}
  Let $\NNs$ be a UMTC. The \textbf{full center} of a 
  Q-system $\Theta$ is defined to be the Q-system $Z(\Theta)
\equiv (Z(\theta),Z(w),Z(x)) =  C_\mathrm{l}( 
  (\Theta\boxtimes \id_\N)\circ^+\Theta_\LR)$ in $\NNs\boxtimes \overline{\NNs}$.
\end{defi}
In particular we have $Z(\id_\N)=\Theta_\LR$.
\begin{defi} 
  \label{defi:LocalIntertwiners}
  Let $\NNs$ be a URFC and 
	$\Theta=(\theta,w,x)$ a Q-system in $\NNs$. We define
	\begin{align*}
		\Hom_\loc(\theta\rho,\sigma)&=\{t \in\Hom(\theta\rho,\sigma): 
		t\cdot P^\mathrm{l}_\theta(\rho)=t\}\komma\\
		\Hom_\loc(\sigma,\theta\rho)&=\{t^\ast \in\Hom(\sigma,\theta\rho): 
		P^\mathrm{l}_\theta(\rho)\cdot t^\ast=t^\ast \}\punkt
	\end{align*}
\end{defi}
In particular, the spaces $\Hom_\loc(\theta\rho,\sigma)$ and 
$\Hom_\loc(\sigma,\theta\rho)$ are anti-isomorphic, due to the self-adjointness
of $P^\mathrm{l}_\theta(\rho)$.
\begin{lem}
  \label{lem:Isometry}
	The isometry $\psi\in\Hom\left( Z(\theta),(\theta\boxtimes \id_\N)\theta_\LR \right)$
	with $\psi\psi^\ast =P^\mathrm{l}_{(\Theta\boxtimes
\id_\N)\circ^+\Theta_\LR}$
	and $\psi^\ast\psi=1$ is of the form:
	\begin{align*}
		\psi=\bigoplus_{\lambda_1,\lambda_2\in\NDN} 
	\bigoplus_{m\in B(\theta\lambda_2,\lambda_1)_\loc} 
	m^\ast \boxtimes \id_{\lambda_2} \in \Hom\left( Z(\theta),(\theta\boxtimes
\id_\N)\theta_\LR \right)\komma
	\end{align*}
  where the sum over $m$ goes over an ONB of
$\Hom_\loc(\theta\lambda_2,\lambda_1)$.
	In particular:
	\begin{align*}
		[Z(\theta)]=\bigoplus_{\lambda_1,\lambda_2\in\NDN}
\langle\theta\lambda_2,\lambda_1\rangle_\loc 
		\left[\lambda_1 \boxtimes \overline\lambda_2\right]
	\komma
	\end{align*}
  where $\langle\slot,\slot\rangle_\loc=\dim\Hom_\loc(\slot,\slot)$. 
\end{lem}
\begin{proof}
	We first note that $u\in  \Hom\left(
	R(\theta),(\theta\boxtimes 1)\theta_\LR \right)$ given by
	\begin{align*}
		u&:=\bigoplus_{\lambda_1,\lambda_2\in\NDN} 
	\bigoplus_{m\in  B(\theta\lambda_2,\lambda_1)} 
	m^\ast \boxtimes \id_{\lambda_2} \in \Hom\left(
	R(\theta),(\theta\boxtimes \id_\N)\theta_\LR \right)\komma\\
	R(\theta)&:=\bigoplus_{\lambda_1,\lambda_2\in\NDN}
	\langle\theta\lambda_2,\lambda_1\rangle \lambda_1\boxtimes\overline{\lambda}_2
	\end{align*}
	is a unitary interwiner.
	It can be shown that 
	$$
	P^\mathrm{l}_{(\Theta\boxtimes \id_\N)\circ^+\Theta_\LR}\cdot u =
	P_{\Theta\boxtimes \id_\N}^\mathrm{l}(\theta_\LR) \cdot u
	\equiv
		\left(\bigoplus_{\lambda\in\NDN} P^\mathrm{l}_{\Theta}
	(\lambda)\boxtimes1_{\overline\lambda}\right)\cdot u\punkt
	$$
	The equality is the statement \cite[Prop.\ 3.14(i)]{FrFuRuSc2006},
	namely it is proven that $C_\mathrm{l}((\Theta\boxtimes \id_\N)\circ^+ \Theta_\LR)$
	which is associated with $P^\mathrm{l}_{(\Theta\boxtimes \id_\N)\circ^+
\Theta}$
	is associated with the projection
	$P^\mathrm{l}_{\Theta\boxtimes \id_\N}(C_\mathrm{l}(\theta_\LR))
	\equiv P^\mathrm{l}_{\Theta\boxtimes \id_\N}(\theta_\LR)$.
	We can conclude by eventually choosing another basis that a maximal 
	isometry invariant w.r.t.\ $P^\mathrm{l}_{(\Theta\boxtimes
\id_\N)\circ^+\Theta_\LR}$ is given by summing just over 
	ONB's of $\Hom_\loc(\theta\lambda_2,\lambda_1)$.
\end{proof}
		\newcommand{\picalpha}{ 
		\tikzmatht[.42]{
			\fill[\coM,rounded corners] (-2,-2) rectangle (2,2);
			\fill[\coM] (-0.8,-1)--(0,-0.2)--(0,0)--(-0.2,0)--
			(-1.2,-1);
			\fill[\coN] 
				(0,-1) .. controls (-0.8,-1) and (-0.8,-1) .. (0,-0.2)
				--(0,0.2) .. controls (-1,-0.8) and (-1,-0.8) .. (-1,0)  
				arc(180:-90:1);
			\draw  (0,-1) .. controls (-0.8,-1) and (-0.8,-1) .. (0,-0.2)
				--(0,0.2) .. controls (-1,-0.8) and (-1,-0.8) .. (-1,0)  
				arc(180:100:1)
				(80:1)
				arc(80:-90:1);
			\draw [ultra thick] (0,-2) node [below] {$\alpha^-_\rho$}--(0,-1.1);
			\draw [thick] (0,-0.9)--(0,0);
			\draw [thick] (0,0)--(0,1);
			\draw [ultra thick] (0,1)--(0,2) node [above] {$\alpha^+_\sigma$};
            \mydot{(0,0)};
		}}
		\newcommand{\picsigma}{
		\tikzmatht[.42]{
			\fill [\coN,rounded corners] (-2,-2) rectangle (2,2);
			\fill [\coM]
				(-1,-2)--(-1,0) arc (180:-90:1) .. controls (-.8,-1) and (-.8,-1)..(-0.8,-2);
			\draw[ultra thick] (0,1)--(0,-0.85);
			\fill [white] (-0.2,-0.2) rectangle (.2,.2);
			\draw (-0.2,-0.2) rectangle (.2,.2);
			\draw (-1,-2) node [below] {$\theta$}--(-1,0) 
			arc(180:100:1) (80:1) arc(80:-90:1)
			.. controls (-.8,-1) and (-.8,-1) ..(-0.8,-2);
			\draw[thick] (0,2) node [above] {$\sigma$}--(0,1);
			\draw[thick] (0,-2) node [below] {$\rho$}--(0,-1.15);
		}}
		\newcommand{\picsigmaa}{
		\tikzmatht[.42]{
			\begin{scope}[yscale=-1]
			\fill [\coN,rounded corners] (-2,-2) rectangle (2,2);
			\fill [\coM]
				(-1,-2)--(-1,0) arc (180:-90:1) .. controls (-.8,-1) and (-.8,-1)..(-0.8,-2);
			\draw[ultra thick] (0,1)--(0,-0.85);
			\fill [white] (-0.2,-0.2) rectangle (.2,.2);
			\draw (-0.2,-0.2) rectangle (.2,.2);
			\draw (-1,-2) node [above] {$\theta$}--(-1,0) 
			arc(180:100:1) (80:1) arc(80:-90:1)
			.. controls (-.8,-1) and (-.8,-1) ..(-0.8,-2);
			\draw[thick] (0,2) node [below] {$\sigma$}--(0,1);
			\draw[thick] (0,-2) node [above] {$\rho$}--(0,-1.15);
			\end{scope}
		}}
		\newcommand{\piccircle}{
			\draw[ultra thick, \coN] (0,-0.5)--(0,-0.9);
			\fill[even odd rule,\coM] 
				(-1.2,-1)--(-0.8,-0.6)--(-0.8,0) arc (180:-90:0.8)
				-- (-0.6,-0.8)--(-0.8,-1)
				(0,-.6) .. controls (-0.4,-.6) and (0,-0.2) 
				.. (0,0)
				.. controls (0,0.2) and (-.6,-0.4) .. (-.6,0)  
				arc (180:-90:.6);
			\draw  (0,-.6) .. controls (-0.4,-.6) and (0,-0.2) 
				.. (0,0)
				.. controls (0,0.2) and (-.6,-0.4) .. (-.6,0)  
				arc (180:100:.6)
				(80:.6)
				arc(80:-90:.6);
			\draw (-1.2,-1)--(-0.8,-0.6)--(-0.8,0) arc (180:100:0.8)
			(80:0.8) arc (80:-90:0.8)
			-- (-0.6,-0.8)--(-0.8,-1);
			\draw[ultra thick] (0,0.6)--(0,0.8);
			\draw[ultra thick] (0,-0.67)--(0,-0.73);
            \mydot{(0,0)};
		}
Given a Q-system $\Theta$ in $\NNs$ and $\iota(\N)\subset \M$ its associated subfactor 
with the inclusion map $\iota\colon\N\to\M$, we will constantly use that the 
Q-system $\Theta$ is of the form $\Theta_{\bar\iota}$ as in 
Def.~\ref{defi:TrivialQSystem}, in other 
words the Q-system $\Theta$ becomes trivial in the 2--C$^\ast$-category generated by
$\NNs,\iota,\bar\iota$. This simplifies many graphical proofs.
\begin{lem}
  \label{lem:IsomorphicHilbertSpaces}
	Let $\NNs\subset \End(\N)$ be a UMTC, $\Theta$ a Q-system 
	in $\NNs$ and $\N\subset \M$ the corresponding subfactor.
	Let $\rho,\sigma \in \NNs$ be irreducible. 
	The spaces $\Hom_\loc(\theta\rho,\sigma)$ 
	and $\Hom(\alpha_\rho^-,\alpha_\sigma^+)$ are isomorphic by
	the map: 
	\begin{align*}
		\Hom_\loc(\theta\rho,\sigma) & \longrightarrow \Hom(\alpha_\rho^-,\alpha_\sigma^+)\\
		\tikzmatht[.42]{
			\fill[\coN,rounded corners] (-2,-1) rectangle (1,1);
			\fill[\coM,rounded corners] (-0.8,-1)--(0,-0.2)--(0,0.2)--(-1.2,-1);
			\draw[rounded corners] (-0.8,-1)--(0,-0.2)--(0,0.2)--(-1.2,-1)
				node [below] {$\theta$};
			\draw[thick] (0,1) node [above] {$\sigma$}--(0,-1) node [below] {$\rho$};
            \mydot{(0,0)};
		}
		&\longmapsto \frac1{\sqrt[4]{d\theta}}
		\picalpha
    \\
		\frac1{\sqrt[4]{d\theta}}
		\picsigma 
		&\longmapsfrom		
		\tikzmatht[.42]{
			\fill[\coM,rounded corners] (-1,-1) rectangle (1,1);
			\draw[ultra thick] (0,1) node [above] {$\alpha_\sigma^+$}--(0,-1)
				node [below] {$\alpha_\rho^-$};
			\fill [white] (-0.2,-0.2) rectangle (.2,.2);
			\draw (-0.2,-0.2) rectangle (.2,.2);
		}
		\punkt
	\end{align*}
	In the same way $\Hom_\loc(\rho,\theta\sigma)$ is isomorphic to 
	$\Hom(\alpha_\rho^+,\alpha_\sigma^-)$. This gives a unitary equivalence  
	between the Hilbert spaces $\Hom_\loc(\rho,\theta\sigma)$ with scalar product $(e,f)=\Phi_\sigma(e^\ast f)$ 
	and $\Hom(\alpha_\rho^+,\alpha_\sigma^-)$ with scalar product $(e',f')=\Phi_{\alpha^+_\sigma}(e'^\ast f')$, where
	$\Phi_\sigma$ and $\Phi_{\alpha^+_\sigma}$ denote the unique left inverse and unique standard left inverse, 
	respectively.
\end{lem}
\begin{proof}
	We first check that the map is well defined, namely
	the image is an element in $\Hom_{\loc}(\theta\rho,\sigma)$ and we have
        (``$=$'' denotes the trivial intertwiner identifying $\theta=\bar\iota\iota$)
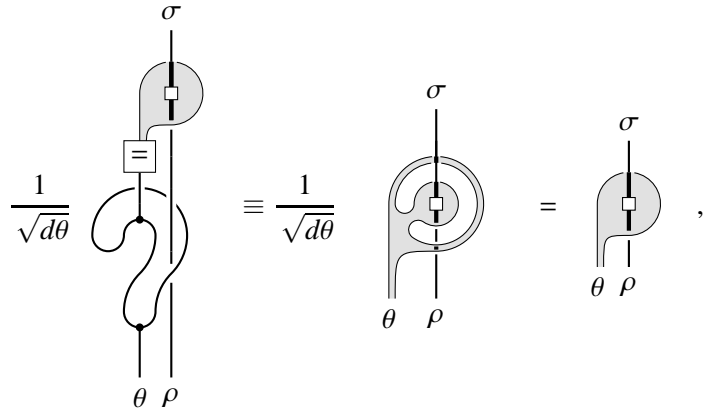
	$$
		\frac1{\sqrt{d\theta}}\,  
	\tikzmatht[.42]{
		\fill[\coN,rounded corners] (-3,-3) rectangle (2,8);
		 \draw [thick] (0,4)--(0,0.6) (0,0)--(0,-3) node [below] {$\rho$};
		\draw [thick](-1,-3) node [below] {$\theta$}--(-1,-1.414);
		\draw [thick](-1,4)--(-1,2)   arc (90:180:.5)
			arc(0:-180:.5) arc(180:0:1.5) --+(0,0) arc(0:-45:1.707)
			arc(135:180:1.707) arc (0:-180:0.5) arc (180:135:1.707)
			arc (-45:0:1.707) arc(0:90:0.5);
		\mydot{(-1,2)} ;
		\mydot{(-1,-1.414)} ;
		\draw[ultra thick,double,\coN] (-1,4)--(-1,2.5); 
		\draw [thick](-1,4)--(-1,2.5); 
		\draw[ultra thick,double,\coN] (0,3)--(0,1); 
		\draw [thick](0,3)--(0,1); 
		\begin{scope}[shift={(0,6)}]
			 \fill [\coM]
			(-1,-2)--(-1,0) arc (180:-90:1) .. controls (-.8,-1)and (-.8,-1) ..(-0.8,-2);
			\draw[ultra thick] (0,1)--(0,-.85);%
			\draw[ultra thick] (0,1)--(0,-.85);%
			\fill [white] (-0.2,-0.2) rectangle (.2,.2);
			\draw (-0.2,-0.2) rectangle (.2,.2);
			\draw 
			(-1,-2)--(-1,0) 
			arc(180:100:1) (80:1) arc(80:-90:1)
			.. controls (-.8,-1)and (-.8,-1) ..(-.8,-2);
			\draw[thick] (0,2) node [above] {$\sigma$}--(0,1);
			\draw[thick] (0,-2)--(0,-1.15);
			\fill [white] (-0.5,-1.5) rectangle (-1.5,-2.5);
			\draw(-0.5,-1.5) rectangle node {$=$} (-1.5,-2.5);
		\end{scope}
	}
	\equiv	
		\frac1{\sqrt{d\theta}}\,  
		\tikzmatht[.42]{
			\fill[\coN,rounded corners] (-3,-3) rectangle (3,3);
			\draw[ultra thick] (0,-1.3)--(0,-1.4) (0,-0.6)--(0,0);
			\draw [thick] (0,-0.9)--(0,0);
			\draw [thick] (0,0)--(0,1);
			\mydot{(0,0)};
			\fill [\coM]	
			(-1.5,-3) node [below] {$\theta$}--(-1.5,0) 
			arc(180:-90:1.5)
			.. controls (-1.3,-1.5)and (-1.3,-1.5) ..(-1.3,-3);
			\draw[ultra thick] (0,1)--(0,-.6);%
			\fill [white] (-0.2,-0.2) rectangle (.2,.2);
			\draw (-0.2,-0.2) rectangle (.2,.2);
			\fill [\coN](-0.7,0) arc (180:-125:0.7)
				arc (55:235:0.3)
				arc (-125:180:1.3)
				arc (180:360:0.3);
			\draw (-0.7,0) arc (180:100:0.7)
				(80:0.7) arc (80:-125:0.7)
				arc (55:235:0.3)
				arc (-125:83:1.3)
				(97:1.3)
				arc (97:180:1.3)
				arc (180:360:0.3);
			\draw [thick] (0,2)--(0,0.5);
			\draw [thick] (0,-1.2)--(0,-0.8);
			\draw 
			(-1.5,-3) node [below] {$\theta$}--(-1.5,0) 
			arc(180:95:1.5) (85:1.5) arc(85:-90:1.5)
			.. controls (-1.3,-1.5)and (-1.3,-1.5) ..(-1.3,-3);
			\draw[thick] (0,-3) node [below] {$\rho$} --(0,-1.6);
			\draw[thick] (0,3) node [above] {$\sigma$} --(0,1.5);
			\draw [ultra thick] (0,1.3)--(0,1.5);
			\draw [ultra thick] (0,-1.35)--(0,-1.45);
		}	
	= 
	\picsigma
  \komma
	$$
  where we used in the first equation that $\Theta$ is of the form 
  $\Theta_{\bar\iota}$ and in the second equation that the closed string
  can be contracted which cancels the prefactor. 
  So we conclude that the image is actually in $\Hom_\loc(\theta\rho,\sigma)$.
	
  We have to show that both maps are inverse to each other:
	\begin{align*}
		\tikzmatht[.42]{
			\fill[\coN,rounded corners] (-2,-1) rectangle (1,1);
			\fill[\coM,rounded corners] (-0.8,-1)--(0,-0.2)--(0,0.2)--(-1.2,-1);
			\draw[rounded corners] (-0.8,-1)--(0,-0.2)--(0,0.2)--(-1.2,-1)
				node [below] {$\theta$};
			\draw[thick] (0,1) node [above] {$\sigma$}--(0,-1) node [below]
				{$\rho$};
            \mydot{(0,0)}; 
		}
		&\longmapsto
		\picalpha
		\longmapsto
		\frac1{\sqrt{d\theta}}
		\tikzmatht[.42]{
			\fill[\coN,rounded corners] 
			(-2,-2) rectangle (2,2);
			\fill [\coM]
			(-1.2,-2)--(-1.2,0) arc (180:-90:1.2) .. controls (-1,-1.2)and (-1,-1.2) ..(-1,-2);
			\draw[ultra thick] (0,0)--(0,-.85);%
			\fill [white] (-0.2,-0.2) rectangle (.2,.2);
			\draw (-0.2,-0.2) rectangle (.2,.2);
			\draw 
			(-1.2,-2)--(-1.2,0) 
			arc(180:100:1.2) (80:1.2) arc(80:-90:1.2)
			.. controls (-1,-1.2)and (-1,-1.2) ..(-1,-2) node [below] {$\theta$};
			\draw[thick] (0,2) node [above] {$\sigma$}--(0,1);
			\draw[thick] (0,-2) node [below] {$\rho$}--(0,-1.30);
			\fill[\coM] (-0.8,-1)--(0,-0.2)--(0,0)--(-0.2,0)--
			(-1.2,-1);
			\fill[\coN] 
				(0,-1) .. controls (-0.8,-1) and (-0.8,-1) .. (0,-0.2)
				--(0,0.2) .. controls (-1,-0.8) and (-1,-0.8) .. (-1,0)  
				arc(180:-90:1);
			\draw  (0,-1) .. controls (-0.8,-1) and (-0.8,-1) .. (0,-0.2)
				--(0,0.2) .. controls (-1,-0.8) and (-1,-0.8) .. (-1,0)  
				arc(180:100:1)
				(80:1)
				arc(80:-90:1);
			\draw [thick] (0,-0.9)--(0,0);
			\draw [thick] (0,0)--(0,1);
			\draw [ultra thick] (0,1)--(0,1.2);
            \mydot{(0,0)}; 
		}
		=
		\tikzmatht[.42]{
			\fill[\coN,rounded corners] (-2,-1) rectangle (1,1);
			\fill[\coM,rounded corners] (-0.8,-1)--(0,-0.2)--(0,0.2)--(-1.2,-1);
			\draw[rounded corners] (-0.8,-1)--(0,-0.2)--(0,0.2)--(-1.2,-1)
				node [below] {$\theta$};
			\draw[thick] (0,1) node [above] {$\sigma$}--(0,-1) node [below]
				{$\rho$};
            \mydot{(0,0)}; 
		}
		\\
		\tikzmatht[.42]{
			\fill[\coM,rounded corners] (-1,-1) rectangle (1,1);
			\draw[ultra thick] (0,1) node [above] {$\alpha_\sigma^+$}--(0,-1)
				node [below] {$\alpha_\rho^-$};
			\fill [white] (-0.2,-0.2) rectangle (.2,.2);
			\draw (-0.2,-0.2) rectangle (.2,.2);
		}
		&\longmapsto		
		\frac1{\sqrt[4]{d\theta}}	
		\picsigma
		\longmapsto			
		\frac1{\sqrt{d\theta}}
		\tikzmatht[.42]{
			\fill[\coM,rounded corners] (-2,-2) rectangle (2,2);
			\draw[ultra thick] (0,-2)	node [below] {$\alpha_\rho^-$}--(0,-1.4) (0,-0.6)--(0,0);
			\draw [thick] (0,-0.9)--(0,0);
			\draw [thick] (0,0)--(0,1);
			\draw [ultra thick] (0,1)--(0,2) node [above] {$\alpha_\sigma^+$};
            \mydot{(0,0)}; 
			\fill [rounded corners, \coM] (-0.8,-2)--(-.8,-1)--(1,-1)
				--(1,1)--(-1,1)--(-1,-2);
			\draw[ultra thick] (0,1)--(0,-.6);%
			\fill [white] (-0.2,-0.2) rectangle (.2,.2);
			\draw (-0.2,-0.2) rectangle (.2,.2);
			\fill [\coN](-0.7,0) arc (180:-125:0.7)
				arc (55:235:0.3)
				arc (-125:180:1.3)
				arc (180:360:0.3);
			\draw (-0.7,0) arc (180:100:0.7)
				(80:0.7) arc (80:-125:0.7)
				arc (55:235:0.3)
				arc (-125:80:1.3)
				(100:1.3)
				arc (100:180:1.3)
				arc (180:360:0.3);
			\draw [thick] (0,2)--(0,0.5);
			\draw [thick] (0,-1.2)--(0,-0.8);
		}	
		=
		\tikzmatht[.42]{
			\fill[\coM,rounded corners] (-1,-1) rectangle (1,1);
			\draw[ultra thick] (0,1) node [above] {$\alpha_\sigma^+$}--(0,-1)
				node [below] {$\alpha_\rho^-$};
			\fill [white] (-0.2,-0.2) rectangle (.2,.2);
			\draw (-0.2,-0.2) rectangle (.2,.2);
		}\komma
	\end{align*}
  where the last equation in the first line is exactly the fact that
  the intertwiner is in $\Hom_{\loc}(\theta\rho,\sigma)$, namely the diagram
  can be deformed to obtain $P^\mathrm{l}_\Theta(\rho)$ which can be omitted; in the last
  equation of the second line the closed string can again be contracted to a 
  dimension cancelling the prefactor.

	Finally, unitarity can be seen as follows:
	\begin{align*}
		\left\| \frac1{\sqrt[4]{d\theta}}~ 
			\picsigmaa
		~ \right\|^2=\frac1{\sqrt{d\theta} \,d\sigma}\,
		\tikzmatht[.42]{
			\fill [\coN,rounded corners] (-2.5,-7) rectangle (1.5,3);
			\begin{scope}[shift={(0,-4)},rotate=180,xscale=-1]
			\fill [\coM]
				(-1,-2)--(-1,0) arc (180:-90:1) .. controls (-.8,-1) and (-.8,-1)..(-0.8,-2);
			\draw[ultra thick] (0,1)--(0,-0.85);
			\fill [white] (-0.2,-0.2) rectangle (.2,.2);
			\draw (-0.2,-0.2) rectangle (.2,.2);
			\draw (-1,-2) node [right] {$\theta$}--(-1,0) 
			arc(180:100:1) (80:1) arc(80:-90:1)
			.. controls (-.8,-1) and (-.8,-1) ..(-0.8,-2);
			\draw[thick] (0,1.5) node [right] {$\sigma$}--(0,1);
			\draw[thick] (0,-2) node [right] {$\rho$}--(0,-1.15);
			\end{scope}
			\begin{scope}
			\fill [\coM]
				(-1,-2)--(-1,0) arc (180:-90:1) .. controls (-.8,-1) and (-.8,-1)..(-0.8,-2);
			\draw[ultra thick] (0,1)--(0,-0.85);
			\fill [white] (-0.2,-0.2) rectangle (.2,.2);
			\draw (-0.2,-0.2) rectangle (.2,.2);
			\draw (-1,-2) node [below] {}--(-1,0) 
			arc(180:100:1) (80:1) arc(80:-90:1)
			.. controls (-.8,-1) and (-.8,-1) ..(-0.8,-2);
			\draw[thick] (0,1.5) node [above left] {}--(0,1);
			\draw[thick] (0,-2) node [below] {}--(0,-1.15);
			\end{scope}
			\draw[thick] (0,-5.5)arc (0:-180:1)--(-2,1.5) arc (180:0:1);
		}
			=
		\left\|~
			\tikzmatht[.42]{
				\begin{scope}[yscale=-1]
				\fill[\coM,rounded corners] (-1,-1) rectangle (1,1);
				\draw[ultra thick] (0,1) node [below] {$\alpha_\sigma^+$}--(0,-1)
					node [above] {$\alpha_\rho^-$};
				\fill [white] (-0.2,-0.2) rectangle (.2,.2);
				\draw (-0.2,-0.2) rectangle (.2,.2);
				\end{scope}
			}
		~\right\|^2
    \komma
	\end{align*}
  where in the last equation we use that the string diagram can be deformed to
  give the standard left inverse for 
  $\alpha_\sigma^+$ (\cf \cite[Lemma~2.2]{Re2000}).
\end{proof}

\begin{defi}[$\alpha$-induction construction \cite{Re2000}]
  \label{defi:AlphaInductionConstruction}
  For a braided subfactor $\iota(\N)\subset \M$ in $\NNs$ there is 
  a Q-system $\Theta_\M=(\theta_\M,w_\M,x_\M)$ in $\NNs\boxtimes \overline
	{\NNs}$ given by:
	\begin{align*}
		[\theta_\M]&=\bigoplus_{\rho,\sigma\in\NDN} Z_{\mu\nu}
		[\mu\boxtimes\bar\nu],
		\\
		Z_{\mu\nu}&=\langle \alpha^+_\mu,\alpha^-_\nu\rangle\\
		x_M &= \bigoplus_{lmn} \sum_{e_1,e_2}
		\sqrt{\frac{d\lambda_2d\mu_2}{d\theta_\M d\nu_2}}
		\Phi^1_{\nu_1}[\iota({e_1}^\ast)(\phi_l^\ast \otimes \phi_m^\ast)
      \iota({e_2})\phi_n]\cdot {e_1}\boxtimes {\bar e_2},
		\\&= \bigoplus_{lmn} \sum_{e_1,e_2}\frac1{\sqrt{d\theta_\M}}
		\sqrt[4]{\frac{d\lambda_2d\mu_2d\nu_1}{d\lambda_1d\mu_1d\nu_2}}
		\Phi^1_{\nu_1}[\cdots]
    \tikzmatht{
	    \draw[thick] (0,-1) node [below] {$\nu_1$}--(0,0) .. controls (1,1) ..
         (1,1.5) node [above] {$\mu_1$};
  		\draw[thick] (0,0) node [right] {$e_1$} .. controls (-1,1) .. 
        (-1,1.5) node [above] {$\lambda_1$};
		  \mydot{(0,0)};
  	}
    \boxtimes
    \tikzmatht{
	    \draw[thick] (0,-1) node [below] {$\bar\nu_2$}--(0,0) .. controls (1,1) ..
         (1,1.5) node [above] {$\bar\mu_2$};
	  	\draw[thick] (0,0) node [right] {$\bar e_2$} .. controls (-1,1) .. 
        (-1,1.5) node [above] {$\bar\lambda_2$};
  		\mydot{(0,0)};
	  }
  	\end{align*}
  where $l$ is considered as a multi-index
  $(\lambda_1\in\NDN,\lambda_2\in\NDN,l=1,\cdots,Z_{\lambda_1,\lambda_2})$ and
  $e_i$ stands for an ONB in $\Hom(\nu_i,\lambda_i\mu_i)$ and $\phi_l$ an ONB
  in $\Hom(\alpha^+_{\lambda_1},\alpha^-_{\lambda_2})$ with respect to the 
  induced left inverse $\Phi^1_{\lambda_1}$.
\end{defi}
The following result was conjectured in \cite{KoRu2010}.
It can be seen as the main technical result. It allows to apply a lot of 
results obtained in the categorical literature to the braided subfactor and conformal net
setting.
\begin{prop}
  \label{prop:main} Let $\NNs$ be a UMTC.
	The $\alpha$-induction  construction for $(\iota(\N)\subset \M,\NNs)$
coincides with the full center $Z(\Theta)$ of the corresponding Q-system
$\Theta$. 
\end{prop}
\begin{proof}
It is already clear that the two constructions give equivalent objects,
namely 
$$
  [Z(\theta)]=\bigoplus_{\lambda_1,\lambda_2\in\NDN} \langle \theta\lambda_2,\lambda_1\rangle_\loc 
  [\lambda_1\boxtimes\bar\lambda_2] = \bigoplus_{\lambda_1,\lambda_2\in\NDN}
  \langle \alpha^+_{\lambda_1},\alpha^-_{\lambda_2}\rangle
[\lambda_1\boxtimes\bar\lambda_2]  = [\theta_M]
$$
follows from Lemma \ref{lem:Isometry} and Lemma \ref{lem:IsomorphicHilbertSpaces}.
We have to show that the two intertwiners $Z(x)$ and $x_\M$ 
of the two respective constructions are equivalent. We decompose $Z(x)$ w.r.t.\
an ONB to show
that we obtain the same coefficients as in the $\alpha$-induction 
construction
for $x_M$.
Using Lemma \ref{lem:Isometry} we have:
\begin{align}
  \label{eq:Zx}
	{\sqrt{d\theta_\LR} \sqrt{d\theta}}Z(x)
	&=\bigoplus_{lmn}\sum_{e_2}
	\sqrt[4]{\frac{d\lambda_1}{d\lambda_2}
	\frac{d\mu_1}{d\mu_2}
	\frac{d\nu_1}{d\nu_2}}
	\,
	\tikzmatht{
		\fill[\coN,rounded corners] (-3,-3) rectangle (2.5,3.5);
		\draw[thick] (0,-1) node [below] {}--(0,0) .. controls (1,1)  .. (1,1.5);
		\draw[thick]
			(0,0) node [right] {$e_2$} .. controls	(-1,1)	.. (-1,1.5);
		\mydot{(0,0)};
		\draw [double, ultra thick, \coN] (-1.8,1.5).. controls (-2.3,1) and (-1.7,0) .. (-0.2,1.5);
		\draw [double, ultra thick, \coN] (0.2,1.5).. controls (-1.3,0) and (-1.8,0) .. (-0.8,-1);
		\begin{scope}[shift={(1,2.5)}]
			\fill[\coM,rounded corners] (-0.8,-1)--(0,-0.2)--(0,0.2)--(-1.2,-1);
			\draw[rounded corners] (-0.8,-1)--(0,-0.2)--(0,0.2)--(-1.2,-1)
				node [below] {};
			\draw[thick] (0,1) node [above] {$\mu_1$}--(0,-1) node [below]
				{};
			\mydot{(0,0)};
			\node (0,0) [right] {$m^\ast$};
		\end{scope}
		\begin{scope}[shift={(-1,2.5)}]
			\fill[\coM,rounded corners] (-0.8,-1)--(0,-0.2)--(0,0.2)--(-1.2,-1);
			\draw[rounded corners] (-0.8,-1)--(0,-0.2)--(0,0.2)--(-1.2,-1)
				node [below] {};
			\draw[thick] (0,1) node [above] {$\lambda_1$}--(0,-1) node [below]
				{};
			\mydot{(0,0)};
			\node (0,0) [right] {$l^\ast$};
		\end{scope}
		\begin{scope}[shift={(0,-2)},yscale=-1]
			\fill[\coM,rounded corners] (-0.8,-1)--(0,-0.2)--(0,0.2)--(-1.2,-1);
			\draw[rounded corners] (-0.8,-1)--(0,-0.2)--(0,0.2)--(-1.2,-1)
				node [below] {};
			\draw[thick] (0,1) node [below] {$\nu_1$}--(0,-1) node [below]
				{};
			\mydot{(0,0)};
			\node (0,0) [right] {$n$};
		\end{scope}
		\fill[\coM]
		(-2.2,1.5).. controls (-2.7,1) and (-1.7,-.5) .. (-1.2,-1)
		--	(-0.8,-1) .. controls (-1.8,0) and (-1.3,0)..(0.2,1.5)
		-- (-0.2,1.5)..controls (-1.7,0)and (-2.3,1) ..(-1.8,1.5);
		\draw (-2.2,1.5).. controls (-2.7,1) and (-1.7,-.5) .. (-1.2,-1);
		\draw (0.2,1.5).. controls (-1.3,0) and (-1.8,0) .. (-0.8,-1);
		\draw (-1.8,1.5).. controls (-2.3,1) and (-1.7,0) .. (-0.2,1.5);
		}
		\boxtimes
  \tikzmatht{
	  \fill[\coN,rounded corners] (-2.5,-3) rectangle (2.5,3.5);
	  \draw[thick] (0,-3) node [below] {$\bar\nu_2$}--(0,0) .. controls (1,1) ..
       (1,1.5)--(1,3.5) node [above] {$\bar\mu_2$};
		\draw[thick] (0,0) node [right] {$\bar e_2$} .. controls (-1,1) .. 
      (-1,1.5)--(-1,3.5) node [above] {$\bar\lambda_2$};
		\mydot{(0,0)};
	}
  \komma
\end{align}
where $l,m,n$ run over an ONB of $\Hom_\loc(\lambda_1,\theta\lambda_2)$,
$\Hom_\loc(\mu_1,\theta\mu_2)$ and $\Hom_\loc(\nu_1,\theta\nu_2)$, respectively.
We use the following expansion of an arbitrary intertwiner 
$t\in\Hom(\nu,\lambda\mu)$ with respect to an ONB $\{e\}$ of
$$
	\tikzmatht{
		\fill[\coN,rounded corners] (-2,-3) rectangle (2,3);
		\draw[thick](0,-3) node [below] {$\nu$}--(0,-.5);
		\draw[thick](-.5,3) node [above] {$\lambda$}--(-.5,.5);
		\draw[thick](.5,3) node [above] {$\mu$}--(.5,.5);
		\fill[white] (-1,-0.5) rectangle (1,0.5);
		\draw (-1,-0.5) rectangle node {$t$} (1,0.5);
	}
	=\sum_e \Phi_\nu(e^\ast t) e
	=\frac1{\sqrt{d\lambda d\mu d\nu}}\sum_{e}
	\tikzmatht{
		\fill[\coN,rounded corners] (-3,-3) rectangle (4,3);
        \mydot{(0,1)};
		\node at (0,1) [left] {$e^\ast$};
		\draw[thick] (0,-1.5) arc (0:-180:1) --(-2,1.5) arc (180:0:1)--(0,1);
		\draw[thick] (0,1) ..controls (0.5,.5) .. (0.5,0) -- (0.5,-.5);
		\draw[thick] (0,1) ..controls (-0.5,0.5) .. (-0.5,0) -- (-0.5,-.5);
		\fill[white] (-1,-0.5) rectangle (1,-1.5);
		\draw (-1,-0.5) rectangle node {$t$} (1,-1.5);
		\begin{scope}[shift={(2,0)}]
			\draw[thick] (0,-3) node [below] {$\nu$}--(0,0) .. controls (1,1)  .. (1,1.5)--(1,3)
				node [above] {$\mu$};
			\draw[thick]
				(0,0) node [right] {$e$} .. controls  (-1,1)  .. (-1,1.5)--(-1,3)
			node [above] {$\lambda$};
            \mydot{(0,0)};
		\end{scope}
	}
$$
with respect to an orthonormal basis $\{e\}$ of $\Hom(\nu,\lambda\mu)$.
The rhs of Eq.\ (\ref{eq:Zx}) becomes
\begin{align*}
	&=\bigoplus_{lmn}\sum_{e_1,e_2}
	\frac{
		\sqrt[4]{\frac{d\lambda_1}{d\lambda_2}
		\frac{d\mu_1}{d\mu_2}
		\frac{d\nu_1}{d\nu_2}}
	}{\sqrt{d\lambda_1d\mu_1d\nu_1 }}
	\tikzmatht{
		\fill[\coN,rounded corners] (-4,-4.5) rectangle (2.5,7);
		\draw[thick] (0,-1) node [below] {}--(0,0) .. controls (1,1)  .. (1,1.5);
		\draw[thick]
			(0,0) node [right] {$e_2$} .. controls	(-1,1)	.. (-1,1.5);
		\mydot{(0,0)};
		\draw [double, ultra thick, \coN] (-1.8,1.5).. controls (-2.3,1) and (-1.7,0) .. (-0.2,1.5);
		\draw [double, ultra thick, \coN] (0.2,1.5).. controls (-1.3,0) and (-1.8,0) .. (-0.8,-1);
		\begin{scope}[shift={(1,2.5)}]
			\fill[\coM,rounded corners] (-0.8,-1)--(0,-0.2)--(0,0.2)--(-1.2,-1);
			\draw[rounded corners] (-0.8,-1)--(0,-0.2)--(0,0.2)--(-1.2,-1)
				node [below] {};
			\draw[thick] (0,1) node [right] {$\mu_1$}--(0,-1) node [below]
				{};
			\mydot{(0,0)};
			\node (0,0) [right] {$m^\ast$};
		\end{scope}
		\begin{scope}[shift={(-1,2.5)}]
			\fill[\coM,rounded corners] (-0.8,-1)--(0,-0.2)--(0,0.2)--(-1.2,-1);
			\draw[rounded corners] (-0.8,-1)--(0,-0.2)--(0,0.2)--(-1.2,-1)
				node [below] {};
			\draw[thick] (0,1) node [right] {$\lambda_1$}--(0,-1) node [below]
				{};
			\mydot{(0,0)};
			\node (0,0) [right] {$l^\ast$};
		\end{scope}
		\begin{scope}[shift={(0,-2)},yscale=-1]
			\fill[\coM,rounded corners] (-0.8,-1)--(0,-0.2)--(0,0.2)--(-1.2,-1);
			\draw[rounded corners] (-0.8,-1)--(0,-0.2)--(0,0.2)--(-1.2,-1)
				node [below] {};
			\draw[thick] (0,-1)--(0,.5) node [below right] {$\nu_1$};
			\mydot{(0,0)};
			\node (0,0) [right] {$n$};
		\end{scope}
		\draw[thick](0,-2.5) arc (0:-180:1.5) --(-3,5) arc (180:0:1.5);
        \mydot{(0,5)};
		\node at (0,5) [right] {$e_1^\ast$};
		\draw[thick] (0,5) .. controls (1,4)  .. (1,3.5)--(1,2);
		\draw[thick] (0,5) .. controls (-1,4) .. (-1,3.5)--(-1,2);
		\fill[\coM]
		(-2.2,1.5).. controls (-2.7,1) and (-1.7,-.5) .. (-1.2,-1)
		--	(-0.8,-1) .. controls (-1.8,0) and (-1.3,0)..(0.2,1.5)
		-- (-0.2,1.5)..controls (-1.7,0)and (-2.3,1) ..(-1.8,1.5);
		\draw (-2.2,1.5).. controls (-2.7,1) and (-1.7,-.5) .. (-1.2,-1);
		\draw (0.2,1.5).. controls (-1.3,0) and (-1.8,0) .. (-0.8,-1);
		\draw (-1.8,1.5).. controls (-2.3,1) and (-1.7,0) .. (-0.2,1.5);
		}
		\cdot 
		\tikzmatht{
			\fill[\coN,rounded corners] (-2,-3) rectangle (2,3.5);
			\draw[thick] (0,-3) node [below] {$\nu_1$}--(0,0) .. controls (1,1)  .. (1,1.5)--(1,3.5)
			node [above] {$\mu_1$};
			\draw[thick]
				(0,0) node [right] {$e_1$} .. controls	(-1,1)	.. (-1,1.5)--(-1,3.5)
			node [above] {$\lambda_1$};
			\mydot{(0,0)};
		}
		\boxtimes
		 \tikzmatht{
			\fill[\coN,rounded corners] (-2,-3) rectangle (2,3.5);
			\draw[thick] (0,-3) node [below] {$\bar\nu_2$}--(0,0) .. controls (1,1)  .. (1,1.5)--(1,3.5)
			node [above] {$\bar\mu_2$};
			\draw[thick]
				(0,0) node [right] {$\bar e_2$} .. controls  (-1,1)  .. (-1,1.5)--(-1,3.5)
			node [above] {$\bar\lambda_2$};
			\mydot{(0,0)};
		}
    \punkt
\end{align*}
We calculate:
\begin{align*}
	\sqrt[4]{\frac{d\lambda_1}{d\lambda_2}
	\frac{d\mu_1}{d\mu_2}
	\frac{d\nu_1}{d\nu_2}}
	\,
	\tikzmatht{
		\fill[\coN,rounded corners] (-4,-4.5) rectangle (2.5,7);
		\draw[thick] (0,-1) node [below] {}--(0,0) .. controls (1,1)  .. (1,1.5);
		\draw[thick]
			(0,0) node [right] {$e_2$} .. controls	(-1,1)	.. (-1,1.5);
		\mydot{(0,0)};
		\draw [double, ultra thick, \coN] (-1.8,1.5).. controls (-2.3,1) and (-1.7,0) .. (-0.2,1.5);
		\draw [double, ultra thick, \coN] (0.2,1.5).. controls (-1.3,0) and (-1.8,0) .. (-0.8,-1);
		\begin{scope}[shift={(1,2.5)}]
			\fill[\coM,rounded corners] (-0.8,-1)--(0,-0.2)--(0,0.2)--(-1.2,-1);
			\draw[rounded corners] (-0.8,-1)--(0,-0.2)--(0,0.2)--(-1.2,-1)
				node [below] {};
			\draw[thick] (0,1) node [right] {$\mu_1$}--(0,-1) node [below]
				{};
			\mydot{(0,0)};
			\node (0,0) [right] {$m^\ast$};
		\end{scope}
		\begin{scope}[shift={(-1,2.5)}]
			\fill[\coM,rounded corners] (-0.8,-1)--(0,-0.2)--(0,0.2)--(-1.2,-1);
			\draw[rounded corners] (-0.8,-1)--(0,-0.2)--(0,0.2)--(-1.2,-1)
				node [below] {};
			\draw[thick] (0,1) node [right] {$\lambda_1$}--(0,-1) node [below]
				{};
			\mydot{(0,0)};
			\node (0,0) [left] {$l^\ast~$};
		\end{scope}
		\begin{scope}[shift={(0,-2)},yscale=-1]
			\fill[\coM,rounded corners] (-0.8,-1)--(0,-0.2)--(0,0.2)--(-1.2,-1);
			\draw[rounded corners] (-0.8,-1)--(0,-0.2)--(0,0.2)--(-1.2,-1)
				node [below] {};
			\draw[thick] (0,-1)--(0,.5) node [below right] {$\nu_1$};
			\mydot{(0,0)};
			\node (0,0) [right] {$n$};
		\end{scope}
		\draw[thick](0,-2.5) arc (0:-180:1.5) --(-3,5) arc (180:0:1.5);
        \mydot{(0,5)};
		\node at (0,5) [right] {$e_1^\ast$};
		\draw[thick] (0,5) .. controls (1,4)  .. (1,3.5)--(1,2);
		\draw[thick] (0,5) .. controls (-1,4) .. (-1,3.5)--(-1,2);
		\fill[\coM]
		(-2.2,1.5).. controls (-2.7,1) and (-1.7,-.5) .. (-1.2,-1)
		--	(-0.8,-1) .. controls (-1.8,0) and (-1.3,0)..(0.2,1.5)
		-- (-0.2,1.5)..controls (-1.7,0)and (-2.3,1) ..(-1.8,1.5);
		\draw (-2.2,1.5).. controls (-2.7,1) and (-1.7,-.5) .. (-1.2,-1);
		\draw (0.2,1.5).. controls (-1.3,0) and (-1.8,0) .. (-0.8,-1);
		\draw (-1.8,1.5).. controls (-2.3,1) and (-1.7,0) .. (-0.2,1.5);
	}
	=
	\sqrt[4]{\frac{d\lambda_1}{d\lambda_2}
	\frac{d\mu_1}{d\mu_2}
	\frac{d\nu_1}{d\nu_2}}
	(d\theta)^{-\frac32}
		\tikzmatht{
		\fill[\coN,rounded corners] (-4,-4.5) rectangle (2.5,7);
		\draw[thick] (0,-1) node [below] {}--(0,0) .. controls (1,1)  .. (1,1.5);
		\draw[thick]
			(0,0) node [right] {$e_2$} .. controls	(-1,1)	.. (-1,1.5);
		\mydot{(0,0)};
		\draw [double, ultra thick, \coN] (-1.8,1.5).. controls (-2.3,1) and (-1.7,0) .. (-0.2,1.5);
		\draw [double, ultra thick, \coN] (0.2,1.5).. controls (-1.3,0) and (-1.8,0) .. (-0.8,-1);
		\begin{scope}[shift={(1,2.5)}]
			\fill[\coM,rounded corners] (-0.8,-1)--(0,-0.2)--(0,0.2)--(-1.2,-1);
			\draw[rounded corners] (-0.8,-1)--(0,-0.2)--(0,0.2)--(-1.2,-1)
				node [below] {};
			\draw[thick] (0,1) node [right] {$\mu_1$}--(0,-1) node [below]
				{};
			\mydot{(0,0)};
			\piccircle
			\node at (.6,0) [right] {$m^\ast$};
		\end{scope}
		\begin{scope}[shift={(-1,2.5)}]
			\fill[\coM,rounded corners] (-0.8,-1)--(0,-0.2)--(0,0.2)--(-1.2,-1);
			\draw[rounded corners] (-0.8,-1)--(0,-0.2)--(0,0.2)--(-1.2,-1)
				node [below] {};
			\draw[thick] (0,1) node [right] {$\lambda_1$}--(0,-1) node [below]
				{};
			\mydot{(0,0)};
			\piccircle
			\node at (-.4,0) [left] {$l^\ast~ ~ ~ ~ ~ ~ ~ ~$};
		\end{scope}
		\begin{scope}[shift={(0,-2)},yscale=-1]
			\fill[\coM,rounded corners] (-0.8,-1)--(0,-0.2)--(0,0.2)--(-1.2,-1);
			\draw[rounded corners] (-0.8,-1)--(0,-0.2)--(0,0.2)--(-1.2,-1)
				node [below] {};
			\draw[thick] (0,-1)--(0,.5) node [below right] {$\nu_1$};
			\mydot{(0,0)};
			\piccircle
			\node at (0.6,0) [right] {$n$};
		\end{scope}
		\draw[thick](0,-2.5) arc (0:-180:1.5) --(-3,5) arc (180:0:1.5);
        \mydot{(0,5)};
		\node at (0,5) [right] {$e_1^\ast$};
		\draw[thick] (0,5) .. controls (1,4)  .. (1,3.5)--(1,2);
		\draw[thick] (0,5) .. controls (-1,4) .. (-1,3.5)--(-1,2);
		\fill[\coM]
		(-2.2,1.5).. controls (-2.7,1) and (-1.7,-.5) .. (-1.2,-1)
		--	(-0.8,-1) .. controls (-1.8,0) and (-1.3,0)..(0.2,1.5)
		-- (-0.2,1.5)..controls (-1.7,0)and (-2.3,1) ..(-1.8,1.5);
		\draw (-2.2,1.5).. controls (-2.7,1) and (-1.7,-.5) .. (-1.2,-1);
		\draw (0.2,1.5).. controls (-1.3,0) and (-1.8,0) .. (-0.8,-1);
		\draw (-1.8,1.5).. controls (-2.3,1) and (-1.7,0) .. (-0.2,1.5);
	}
	=\\
	=
	\tikzmatht{
		\fill[\coN,rounded corners] (-4,-5) rectangle (2.5,9);
		\fill[\coM] (-3.5,6.5) arc (180:0:2)-- (.5,6.5) arc (0:-90:0.5)--(-1.5,6) arc(90:180:0.5) --(-2,5.5)--(-2,4) arc (180:270:0.5) --(1.5,3.5) arc (90:0:0.5) --(2,2) arc(0:-90:.5)--(-1.5,1.5) arc(90:180:0.5)--(-2,-0.5) arc(180:270:0.5) --(0,-1) arc (90:0:0.5) --(0.5,-2.5) arc (0:-180:2)--(-3.5,6.5);
		\draw[thick] (0,-.85) node [below] {}--(0,0) .. controls (1,1)  .. (1,1.35);
		\draw[thick]
			(0,0) node [right] {$e_2$} .. controls	(-1,1)	.. (-1,1.35);
		\mydot{(0,0)};
		\begin{scope}[shift={(1,2.5)}]
				\draw[ultra thick] (0,1) node [below] {}--(0,-.85)
					node [above] {};
				\fill [white] (-0.2,-0.2) rectangle (.2,.2);
				\draw (-0.2,-0.2) rectangle (.2,.2);
			\node (0,0) [left] {$m^\ast$};
		\end{scope}
		\begin{scope}[shift={(-1,2.5)}]
				\draw[ultra thick] (0,1) node [below] {}--(0,-.85)
					node [above] {};
				\fill [white] (-0.2,-0.2) rectangle (.2,.2);
				\draw (-0.2,-0.2) rectangle (.2,.2);
			\node (0,0) [left] {$l^\ast$};
		\end{scope}
		\begin{scope}[shift={(0,-2)},yscale=-1]
				\draw[ultra thick] (0,.5) node [below] {}--(0,-.85)
					node [above] {};
				\fill [white] (-0.2,-0.2) rectangle (.2,.2);
				\draw (-0.2,-0.2) rectangle (.2,.2);
			\node (0,0) [left] {$n$};
		\end{scope}
		\draw[thick](0,6)--(0,5);
		\draw[ultra thick](0,-2.5) arc (0:-180:1.5) --(-3,6.5) arc (180:0:1.5)--(0,6);
		\mydot{(0,5)};
		\node at (0,5) [right] {$e_1^\ast$};
		\draw[thick] (0,5) .. controls (1,4)  .. (1,3.5)--(1,3);
		\draw[thick] (0,5) .. controls (-1,4) .. (-1,3.5)--(-1,3);
		\draw (-3.5,6.5) arc (180:0:2)-- (.5,6.5) arc (0:-75:0.5) (-0.15,6)-- (-1.5,6)
arc(90:180:0.5) --(-2,5.5)--(-2,4) arc (180:270:0.5)
--(-1.15,3.5) (-0.85,3.5)--(0.85,3.5) (1.15,3.5)--
(1.5,3.5) arc (90:0:0.5)
--(2,2) arc(0:-90:.5)--(-1.5,1.5) arc(90:180:0.5)--(-2,-0.5) arc(180:270:0.5) --(0,-1) arc (90:0:0.5) --(0.5,-2.5) arc (0:-180:2)--(-3.5,6.5);
	}
	=d\nu_1\sqrt{d\theta}\sqrt[4]{\frac{d\lambda_1d\lambda_2d\mu_1d\mu_2}{
		d\nu_1d\nu_2}} \Phi_{\nu_1}^1[\cdots]
        \komma 
\end{align*}
where we first use that the intertwiners $l,m,n$ are in $\Hom_\loc(\slot,\slot)$ 
and then replace by Lemma \ref{lem:IsomorphicHilbertSpaces} with an orthonormal
basis in $\Hom(\alpha^+_{\lambda_1},\alpha^-_{\lambda_2})$ and in the second
step deform the 
$\iota$ string to obtain the left inverse of  $\alpha^+_{nu_1}$ and $\Phi^1_{\nu_1}
[\cdots]$ is the expression of Def.\ \ref{defi:AlphaInductionConstruction}.
This shows that $Z(x)$ has the same coefficients as $x_M$ from the $\alpha$-induction construction.
\end{proof}
We need the following general 
result as a main tool in the following sections.

\begin{prop}[\cf \cite{KoRu2008}]
  \label{prop:KongRunkel}
        Let $\Theta_a$ and $\Theta_b$ be irreducible in a UMTC $\NNs$. Then
        $\Theta_a$ and $\Theta_b$ are Morita equivalent if and only if 
	$Z(\Theta_a)$ and $Z(\Theta_b)$ are equivalent. 
\end{prop}

\subsection{The adjoint functor of the full center}
\label{SubSec:Adjoint}
We have a tensor functor $T$ as follows: 
the map
\begin{align}
	T\left(\bigoplus_{i} \lambda_i\boxtimes \bar \mu_i\right)=
		\bigoplus_i \lambda_i\circ\bar\mu_i
\end{align}
is an extension of the monoidal product (which by definition is 
a bifunctor).

We have $T(\id_\N\boxtimes \id_\N)=\id_\N$ and 
the family of morphisms 
\begin{align}
	&\mu_{(\rho_1\boxtimes\bar\sigma_1),(\rho_2\boxtimes\bar\sigma_2)}\colon
T(\rho_1\boxtimes \bar\sigma_1)\circ T(\rho_2\boxtimes \bar\sigma_2) \longrightarrow
	T(\rho_1\rho_2\boxtimes \bar \sigma_1 \bar \sigma_2)\nonumber\\
        \label{eq:mu}
	&\mu_{(\rho_1\boxtimes\bar\sigma_1),(\rho_2\boxtimes\bar\sigma_2)}:=
	(1_{\rho_1}\otimes \varepsilon(\rho_2,\bar\sigma_1)^\ast \otimes 1_{\bar\sigma_2})
	\equiv \rho_1(\varepsilon(\rho_2,\bar\sigma_1)^\ast)
\end{align}
extends to a family 
\begin{align*}
	\mu_{(\beta_1),(\beta_2)}\colon T(\beta_1)\circ T(\beta_2) &\longrightarrow
	T(\beta_1\circ \beta_2), & \beta_1,\beta_2\in \NNs\boxtimes \overline{\NNs}
\end{align*}
and makes the following 
diagram commute:
$$	
	\begin{matrix}
	T(\beta_1)\circ T(\beta_2)\circ T(\beta_3) 
	&\longrightarrow & T(\beta_1)\circ T(\beta_2\circ\beta_3)\\
	\downarrow & & \downarrow\\
	T(\beta_1\circ\beta_2) &\longrightarrow & T(\beta_1\circ\beta_2\circ\beta_3)
	\end{matrix}.
$$
This means $T$ is a (strict with respect to the unity but in general non-strict
for associativity, \ie $\mu_{\bullet,\bullet}\neq 1$) strong monoidal functor (tensor functor).
It is  well known that strong monoidal functors map
monoids into monoids, by this we can conclude 
that 
for $\Theta_2=(\theta_2,w_2,x_2)$ a Q-system in $ \NNs\boxtimes \overline{\NNs}$ 
we obtain a (reducible) Q-system $T(\Theta_2)=(T(\theta_2),
w_{T(\Theta_2)},x_{T(\Theta_2)})$ by
\begin{align*}
	w_{T(\Theta_2)}&=T(w_2),
	&
	x_{T(\Theta_2)}&=\mu_{\theta_2,\theta_2}^\ast \cdot T(x_2)\\
\end{align*}
or explicitely by ($t_i^{jk} 
\in\Hom(\rho_i\boxtimes\bar\sigma_i,\rho_j\rho_k\boxtimes\bar\sigma_j\bar\sigma_k)$)
\begin{align*}
        \theta &=\bigoplus_i \rho_i\boxtimes\bar\sigma_i
        &
        x&=\bigoplus_{ijk} t_i^{jk}\\
        T(\theta_2) &= \bigoplus_i \rho_i\bar\sigma_i
        &
        x_{T(\Theta)} &= \bigoplus_{ijk} 
        \underbrace{
                \rho_j(\varepsilon(\rho_k,\bar\sigma_j))\cdot T(t_i^{jk})
        }_{\in\Hom(\rho_i\bar\sigma_i,\rho_j\bar\sigma_j\rho_k\bar\sigma_k)}
        \punkt
\end{align*}
We note that even if $\Theta$ is \local $T(\Theta)$ is in general not \local,
because the functor is not braided.

We introduce the notion of a direct sum for Q-systems (\cf \cite[p.\ 321]{EvPi2003}).
Let $\NNs\subset \End(\N)$ be a URFC  
and $\{\Theta_i=(\theta_i,w_i,x_i)\}_{i=1,\ldots,n}$ be 
Q-systems in $\NNs$. The direct sum Q-system $\Theta=(\theta,w,x)$ with $\theta=\bigoplus_{i=1}^n \theta_i$ 
is defined by 
\begin{align*}
	\theta&= \sum_{i=1}^n \Ad T_i \circ \theta_i\komma&
	w&= \frac{1}{\sqrt{d(\theta)}} \sum_{i=1}^n d_i \cdot T_i\cdot w_i\komma&
	x&= \sum_{i=1}^n \theta(T_i)T_ix_iT_i^\ast
	\komma 
\end{align*}
where $d_i=\sqrt{d({\theta_i})} = d({\iota_i})$ and $T_i$ are generators of the Cuntz
algebra with $n$ elements, \ie $T_i^ \ast T_j = \delta_{ij}\cdot 1$ and $\sum_i T_iT_i^\ast =1$. 
If $(\theta_i,w_i,x_i)$ corresponds to the subfactor $\N\subset \M_i$ with inclusion map $\iota_i$,
then $(\theta,w,x)$ corresponds to the 
inclusion
$\N\subset \bigoplus_{i=1}^n \M_i$. The $p_i=T_iT_i^\ast$ give a decomposition 
in the sense of Remark \ref{rmk:SubQSystem}.

The following identity has been proven on the level of objects 
in \cite[Prop.\ 3.3.]{Ev2002}. We remark that a priori
it is not clear that this ``curious identity'' holds also 
on the level of Q-systems. It is directly related to the adding the boundary 
construction in \cite{CaKaLo2013} as we discuss in
Sect.~\ref{sec:AddingBoundary}.

\begin{prop}[{\cf \cite[Prop.\ 4.3]{KoRu2008}}]
	\label{prop:DirectSumQSystem}
	Let $\NNs\subset\End(\N)$ be a UMTC and $\Theta$ a Q-system in $\NNs$ 
	with corresponding subfactor $\N\subset\M$. Then we have an equivalence 
	of Q-systems:
	$$
		T(Z(\Theta)) \cong \bigoplus_{a\in\NDM} \Theta_a.
	$$
\end{prop}
Our first aim was to prove this identity directly for the $\alpha$-induction 
construction. We had a graphical proof 
for the trivial Q-system. Because the $\alpha$-induction construction coincides 
with the full center it follows now easily from the general results of \cite{KoRu2008}.
\begin{proof}
We note (see Rem.\ \ref{rmk:Inductive}) that the Q-system $\Theta_a$ for some $a\in\NMs$ 
or equivalently $\bar a\in\s\M\N$ corresponds on the nose 
with the Q-system 
$\Phi(\bar a)^\vee\otimes_\Theta \Phi(\bar a)=\Phi(a)\otimes_\Theta
\Phi(\bar a)$ 
constructed in \cite{KoRu2008},
where $\Phi\colon \s{\M}{\N} \to  \Bim(\Theta,\id)$ is the functor in 
Prop.\ \ref{prop:EquivCat}. Then one can directly apply \cite[Prop.\ 4.3]{KoRu2008}.
\end{proof}

As a corollary this implies 
the ``curious identity'' 
which was proven in \cite[Prop.\ 3.3.]{Ev2002}
and shows that behind this identity indeed sits 
more structure.
\begin{cor}[{\cf \cite[Prop.\ 3.3.]{Ev2002}, see also \cite[Cor 6.13.]{BcEvKa1999}}]
  Let $\N\subset \M$ be a non-degenerately braided type \three subfactor
  and $Z_{\lambda\mu}=\langle \alpha^+_\lambda,\alpha^-_\mu\rangle$ 
  for $\lambda,\mu\in {}_\N\Delta_\N$. Then we have
  \begin{align}
	\bigoplus_{a\in{}_\N\Delta_\M} [a\bar a] 
	  &= \bigoplus_{\rho,\sigma\in{}_\N\Delta_\N} Z_{\rho\sigma} 
	  [\rho\bar\sigma]
  \end{align}
  and in particular the number of elements in ${}_\N\Delta_\M$ or  
        ${}_\M\Delta_\N$ is given 
  by
  \begin{align*}
  	| \NDM |=|\MDN|&= \sum_{\rho\in \NDN} Z_{\rho\rho}   	\punkt
  \end{align*}
\end{cor}

\begin{rmk} The functor $T(\slot)$ gives a (left) adjoint to the full center
$Z(\slot)$, namely  $\Theta$ is a sub-Q-system of $T(Z(\Theta))$.
\end{rmk}

\section{Modular invariance and Q-systems in
\texorpdfstring{$\NNs\boxtimes\overline{\NNs}$}{CxC'}}
\label{sec:Modular}

\subsection{Characterization of modular invariant Q-systems} 
Let $\NNs\subset \End(\N)$ be a UMTC. Given a Q-system $\Theta$ and the corresponding 
extension $\iota(\N)\subset\M$ let
$Z_{\mu\nu}=\langle\alpha^+_\mu,\alpha^-_\nu\rangle$ for $\mu,\nu\in\NDN$.
The matrix $Z=(Z_{\mu\nu})_{\mu,\nu\in\NDN}$ is a \textbf{modular invariant} \cite{BcEvKa1999}, \ie
it commutes with $S$ and $T$ from (\ref{eq:ST}).
It is called normalized because $Z_{00}=1$ and sufferable because it comes
from an inclusion $\iota(\N)\subset\M$. 
The $\alpha$-induction construction or equivalently the full center 
gives a Q-system $\Theta_2$ in $\NNs\boxtimes \overline{\NNs}$ 
with $[\theta_2]=\bigoplus_{\mu,\nu\in\NDN} Z_{\mu\nu} [\mu \boxtimes \bar\nu]$.
It is sometimes convenient to write the matrix $(Z_{\mu\nu})$ formally in character
form as $Z= \sum_{\mu,\nu\in\NDN} Z_{\mu,\nu}\chi_\mu\bar\chi_\nu$.

\begin{lem}[{\cite{BcEvKa2000}, see also \cite[Thm 4.5]{KiOs2002}}]
	\label{lem:leq}
	Let $\NNs$ be a UMTC. If $\Theta$ is an  irreducible \local Q-system in
$\NNs$, then 
        $\dim\MMs^0=\dim \NNs/(d\Theta)^2$. In particular,
	$d\Theta\leq \dim(\NNs)^\frac12$. 
\end{lem}
\begin{proof} The first statement is a combination of 
Thm.\ 4.2 and Prop.\ 3.1 in \cite{BcEvKa2000}. The second statement 
follows from the first, using $\dim \MMs^0\geq 1$.

Using Remark \ref{rmk:alphaInduction}
        and \ref{rmk:Dyslexic}, this also follows from \cite[Thm
4.5]{KiOs2002}.
\end{proof}

\begin{prop}[{\cite[Thm.\ 3.4, Prop.\ 3.22]{KoRu2009}}]
  \label{prop:QSystemModular}
  Let $\Theta_2$ be an irreducible \local Q-system 
  in $\NNs\boxtimes\overline{\NNs}$, then the following are equivalent:
  \begin{enumerate}
    \item $d\Theta_2=\dim(\NNs)$
    \item $Z=(Z_{\mu\nu})$ is a modular invariant
    \item $\Theta_2\equiv Z(\Theta)$ for some irreducible 
    Q-system $\Theta$ in $\NNs$.
  \end{enumerate}
\end{prop}
\begin{proof}
(3) are equivalent (1) by \cite[Thm.\ 3.4, Prop.\ 3.22]{KoRu2009}
(see also \cite[Thm 3.4]{Mg2010}, \cite{DaMgNiOs2013}). 

The notion of modular invariance in \cite[Thm.\ 3.4]{KoRu2009} is a bit
different. But by \cite[Appendix C]{LoRe2004} we obtain that (2) implies (1),
namely the argument shows that if $d\theta<\dim(\NNs)$ then $Z$ cannot be
modular invariant. Together with Lemma \ref{lem:leq} this gives the statement.

(3) implies (2) is clear by the fact that $Z_{\mu\nu}=\langle \alpha^+_\mu,
\alpha^-_\nu\rangle$ defines a modular invariant and that 
$Z(\Theta)$ coincides with the $\alpha$-induction construction Prop.\
\ref{prop:main}.
\end{proof}

\subsection{Permutation modular invariants}
\label{subsec:PermutationModularInvariant}
Let $\NNs \subset \End(\N)$ be a UMTC. A non-negative integer valued matrix
$Z=(Z_{\mu\nu})_{\mu,\nu\in\NDN}$ with $Z_{\id_N,\id_\N}=1$
is called a \textbf{modular invariant} if it commutes with the matrices $S$ and
$T$ constructed in Subsect.~\ref{subsec:UMTC}. It is called \textbf{realizable} (sufferable) if there exists a braided subfactor $(\iota(\N)\subset\M,\NNs)$ such that $Z_{\mu\nu}=\langle 
\alpha^+_\mu,\alpha^-_\nu\rangle$.

\begin{prop} 
	\label{prop:perm1}
	Let $\NNs\subset \End(\N)$ be a UMTC and 
	$\phi\in \Aut(\NDN)$ which only fixes the sector $[\id_\N]$ 
	and which extends to a braided automorphism of $\NNs$. 
	Then there is a braided subfactor $\N\subset \M_\phi$ in $\NNs$ 
	with 
	$$
		[\theta_\phi]
			= \bigoplus_{\nu} 
			 n_\nu [\nu],
		\qquad n_\nu=\sum_\mu \langle\mu\phi(\bar\mu),\nu\rangle 
	$$
	which
	realizes the permutation modular invariant
	$Z_{\mu\nu}=\delta_{\nu,\phi(\mu)}$.
\end{prop}
\begin{proof}
	By the Longo--Rehren construction Def.\ \ref{defi:LR} there is a Q-system
	$\Theta_\LR^\phi$ with:
	$$
		[\theta^\phi_\LR]=\bigoplus_{\mu} [\mu \boxtimes \phi(\bar\mu)]\punkt
	$$
	We define the Q-system $\Theta_\phi:=T(\Theta_\LR^\phi)$ in $\NNs$ with
	$$
		[\theta_\phi]:=\bigoplus_{\mu} [\mu\phi(\bar\mu)]
			= \bigoplus_{\nu} 
			 n_\nu [\nu],
		\qquad n_\nu=\sum_\mu \langle\mu\phi(\bar\mu),\nu\rangle 
	$$ as above 
	which is irreducible because $0=\langle \mu\phi(\bar \mu),\id_\N\rangle$ for
	$[\mu]\neq [\id_\N]$ by the assumption about $\phi$ not having non-trivial fixed
	points.
	Because $T(\slot)$ is left-adjoint to $Z(\slot)$ the subfactor 
	$\N\subset \M_\phi$ given by the Q-system $\Theta_\phi$ has the modular 
  invariant $Z_{\mu\nu}=\delta_{\nu,\phi(\mu)}$.
\end{proof}
A particular case is, if $\NNs$ has no non-trivial self-conjugate sectors besides 
the trivial sector, in this case the charge conjugation $C$ might fulfill
the assumptions and the obtained subfactor realizes the charge conjugation
modular invariant $Z=C$. We therefore can answer a particular case 
of the question how $Z=C$ is realized, namely the case that there 
are no non-trivial self-conjugate charges.

\begin{example} 
  \label{ex:E61}
  The UMTC $E_{6,1}$ for example obtained by positive energy
   representation of loop groups, has 3 sectors $\{\rho_0,\rho_1,\rho_2\}$ with $\ZZ_3$ fusion
  rules, \ie $[\rho_i\rho_j] =[\rho_{i+j \mod 3}]$ for $0\leq i,j\leq 2$,
  and the charge conjugation transposes the two non-trivial charges.
  Then Prop.\ \ref{prop:perm1} yields a Q-system with
   $[\theta]=[\rho_0] \oplus[\rho_1]\oplus[\rho_2]$ which realizes $Z=C$,
  \ie $Z=|\chi_0|^2+\chi_1\bar\chi_2+\chi_2\bar\chi_1$.
\end{example}

If there is fixed point in the permutation the same construction 
as in the proof of Prop.\ \ref{prop:perm1} is possible but we do not know how a 
dual canonical endomorphism of an irreducible Q-system giving the modular
invariant 
would look, because the ``adjoint functor'' gives a reducible Q-system.
Nevertheless, we can conclude that for a permutation matrix $Z$ of $\NDN$ which 
gives rise to a braided automorphism, there exists 
a braided subfactor $\iota(\N)\subset\M$ in $\NNs$ which has $Z$ as a
modular invariant, i.e. such permutation modular invariants are  realizable.

The category $\NNs$ is called pointed if all irreducible objects are invertible,
\ie have dimension 1 or in other words $\NNs=\Pic(\NNs)$.
\begin{lem}
\label{lem:pointed}
Let $\NNs\in \End(\N)$ be a pointed UMTC and let $\Theta_1$ and $\Theta_2$ be Q-systems. 
If $\Theta_1$ and $\Theta_2$ are Morita equivalent, then they are equivalent.
\end{lem}
\begin{proof}
  Let $\Theta_1$ and $\Theta_2$ be irreducible Q-systems in $\NNs$ which are 
  Morita equivalent. Without lost of generatlity, we may assume that 
  $\Theta_1=\Theta_{\bar\iota}$ comes from a subfactor $\iota(\N)\subset\M$ and 
  $\Theta_2=\Theta_a$ with $a\in \NMs$ irreducible.
  
  Because $\NNs$ is pointed the sectors form an abelian (due to the braiding) 
  group denoted $G$. The multiplication in $G$ is given by the  fusion rules,
\ie $\NDN=\{\lambda_g: g \in G\}$ with
  $[\lambda_g\lambda_h]=[\lambda_{gh}]$ for all $g,h\in G$ and
  $[\lambda_{g^{-1}}]=[\bar\lambda_g]$.
  We note that $\iota\lambda_g$ is irreducible, namely by Frobenius reciprocity 
  $\langle\iota\lambda_g,\iota\lambda_g\rangle = \langle
  \theta,\lambda_g\bar\lambda_g\rangle=\langle \theta,\id_\N\rangle=1$.
  Therefore $\NDM\subset \{\lambda_g\bar \iota:g\in G\}$
  (because there can be $[\lambda_g\bar\iota]=[\lambda_h\bar \iota]$).
  So we may assume that $a=\lambda_g\bar \iota$ and can conclude that 
  $[\theta_{a}]=[\lambda_g\bar\iota\iota\bar\lambda_g]=[\theta\bar\lambda_g\lambda_g]=[\theta]$.
  It is easy to check that using $\varepsilon(\lambda_g,\theta)$ we can
  construct a unitary intertwiner $\theta_a\to\theta\lambda_g\bar \lambda_{g}\to
  \theta$, which gives an equivalence of the two Q-systems.

  Alternatively, we can use that $\bar \alpha^\pm_{\lambda_{g^{-1}}}$ is an automorphism
  satisfying $a\bar \alpha^\pm_{  \lambda_{g^{-1}}}= \lambda_g
  \bar\iota\alpha^\pm_{\lambda_{g^{-1}}}=\lambda_g\lambda_{g^{-1}}\bar
  \iota=\bar \iota$. Then Prop.\ \ref{prop:Picard} gives an alternative proof of the statement.
\end{proof}
Let $\NNs\subset \End(\N)$ be a pointed UMTC and $\Theta$ be a Q-system and 
$Z_{\mu\nu}=\langle\alpha_\mu^+,\alpha_\nu^-\rangle$.
Then Lemma \ref{lem:pointed} shows that $T(Z(\Theta))$ is equivalent
to $\bigoplus_{i=1}^{\tr Z}\Theta$. 
Therefore in this case we obtain an easy formula for $\theta$ in terms of its modular invariant
matrix $Z=(Z_{\mu\nu})$: 
\begin{align*}
	[\theta] = \frac{1}{\tr Z} \bigoplus_{\rho\in \NDN}
\sum_{\mu,\nu\in\NDN}Z_{\mu\nu}N_{\mu\nu}^\rho [\rho]
  \komma
\end{align*}
see also \cite{Pi2007}.

\subsection{Maximal chiral subalgebras and second cohomology for modular
invariant Q-systems}

Let us assume that $\Theta$ is a commutative Q-system in $\NNs$ 
and $N\subset M$ the associated subfactor.

The category $\Mod(\Theta)$ forms a (non-strict) tensor category as follows.
Let $\rho$, $\sigma$ be two right $\Theta$-modules.  
Because $\Theta$ is commutative, we obtain a left action on $\rho$ and $\sigma$ using the
braiding, which makes them bimodules. Then the tensor product $\rho\otimes\sigma$ is defined to 
be the object $\rho\otimes_\Theta \sigma$ as in Remark
\ref{rmk:alphaInduction}, which we see as right module by forgetting the left
action.

Let $\Mod_0(\Theta)$ the subcategory of dyslectic modules (see
\cite{Pa1995,KiOs2002}), i.e.
modules $(\rho,r)$, such that
$r\varepsilon(\theta,\rho)\varepsilon(\rho,\theta)=r$, graphically:
\begin{align*}
    \tikzmatht[0.21]{
    \draw[thick] (2,-2.5) .. controls (2,-2) .. (0,0);
    \mydot{(0,0)};
    \mydot{(0,0)};
    \node at (0,0) [left] {$r$};
    \draw[thick] (0,1) node [above] {$\rho$}--(0,-2.5)
	.. controls (0,-4.5) and (2,-4.5) .. (2,-6)
    .. controls (2,-7.5) and (0,-7.5) .. (0,-9) node [below] {$\rho$}; 
	\draw[thick]  (2,-2.5) 
	.. controls (2,-4.5) and (0,-4.5) .. (0,-6)
    .. controls (0,-7.5) and (2,-7.5) .. (2,-9) node [below] {$\theta$}; 
	\draw [double,ultra thick, white] 
        (2,-2.6) .. controls (2,-4.5) and (0,-4.5) .. (0,-5.9)
        (2,-6.1) .. controls (2,-7.5) and (0,-7.5) .. (0,-8.9);
	\draw [thick] 
        (2,-2.5) .. controls (2,-4.5) and (0,-4.5) .. (0,-6)
        (2,-6)   .. controls (2,-7.5) and (0,-7.5) .. (0,-9);
  }
  =
    \tikzmatht[0.21]{
    \draw[thick] (2,-2.5) .. controls (2,-2) .. (0,0);
    \mydot{(0,0)};
    \mydot{(0,0)};
    \node at (0,0) [left] {$r$};
    \draw[thick] (0,1) node [above] {$\rho$}--(0,-2.5)--
    (0,-9) node [below] {$\rho$}; 
	\draw[thick]  (2,-2.5) --
    (2,-9) node [below] {$\theta$}; 
  }
  \punkt
\end{align*}

It can easily be seen that if we give the induced right $\Theta$-module $\rho\theta$
the structure of a bimodule using the braiding that it becomes equivalent to
the $\alpha$-induction $\Phi(\alpha^\pm_\rho)$ in Remark \ref{rmk:alphaInduction}, where the sign is depending on the choice of
the braiding. We obtain that $\Bim^\pm (\Theta,\Theta) \cong \Mod(\Theta)$ as tensor
categories, but we will just need the following fact.

\begin{rmk}
\label{rmk:Dyslexic}
The map obtained by restricting bimodules to right modules 
$$
        \Bim^0(\Theta,\Theta) \to \Mod_0(\Theta)
$$
is an equivalence of categories. Namely, an object in $\Bim^0(\Theta,\Theta)$
gives 
a dyslectic module, because using the fact that it is contained both, in the image of $\alpha^+$ and $\alpha^-$, 
we can ``unwind'' the
double braid. Conversely, if a module is dyslectic, the left action obtained by
the both braidings coincide, so it must come  from  $\Bim^0(\Theta,\Theta)$. 
\end{rmk}

For $\beta\in \MMs$ we define the \textbf{$\sigma$-restriction} $\sigma_\beta=\bar \iota \beta\iota\in\NNs$.

Given $\Theta_\pm$ commutative Q-systems corresponding to $\N\subset\M_\pm$ 
it follows that $\s{\M_\pm}{\M_\pm}^0$ are again UMTCs. Let us assume there is a braided
equivalence $\phi\colon\s{\M_+}{\M_+}^0\to\s{\M_-}{\M_-}^0$. 
Now we consider the Q-system $\Theta_\LR^\phi$ in 
$\s{\M_+}{\M_+}^0\boxtimes\overline{\s{\M_-}{\M_-}^0}$. By composing 
$\iota_\LR$ with $\iota_1\boxtimes\iota_2$ 
we obtain a Q-system 
$$
  \Theta_{(\Theta_+,\Theta_-,\phi)} = \Theta_{(\bar\iota_1\boxtimes \bar \iota_2)\circ\bar\iota_\LR^\phi}
$$
with
\begin{align*}
  [\theta^\phi_\LR] &= \bigoplus_{\alpha\in
\D{\M_+}{\M_+}^0,\beta\in\D{\M_-}{\M_-}^0} \tilde Z_{\alpha\beta}
    [\alpha\boxtimes\bar\beta],&
  \tilde Z_{\alpha\beta} &= \delta_{\alpha,\phi(\beta)}\\
  [\theta_{(\Theta_+,\Theta_-,\phi)}]&=\bigoplus_{\mu,\nu\in\NDN} Z_{\mu\nu}
[\mu\boxtimes \bar\nu],&
  Z_{\mu\nu} &= \sum_{\alpha\beta} Z_{\alpha\beta} \langle
    \sigma^+_\alpha,\mu\rangle\langle \sigma^-_{\bar\beta},\bar \mu\rangle
    \\&&&= \sum_\tau b^+_{\tau,\mu}b^-_{\phi(\tau),\nu}
\end{align*}
where $b^\pm_{\tau,\mu}=\langle \sigma^\pm_\tau,\mu\rangle$ for $\tau\in
\s{\M_\pm}{\M_\pm}^0$.
All maximal commutative Q-systems in $\NNs\boxtimes\overline{\NNs}$ are of this
form: 
\begin{prop}[{\cite[Prop.\ 3.7, Cor.\
3.8]{DaNiOs2013}}]
  \label{prop:ModularQSystems}
  There is a one-to-one correspondence between
  \begin{enumerate}
    \item Equivalence classes of 
    \local irreducible Q-systems
$\Theta_2$ in $\NNs\boxtimes\overline{\NNs}$ with $d\theta_2=\dim(\NNs)$.
    \item Isomorphism classes of triples $(\Theta_+,\Theta_-,\phi)$ where
$\Theta_\pm$ are \local irreducible Q-systems in $\NNs$ and   
    $\phi\colon\s{\M_+}{\M_+}^0\to\s{\M_-}{\M_-}^0$ is an equivalence of braided
    categories.
    \item Indecomposable module categories over $\NNs$.
  \end{enumerate}
\end{prop}
\begin{proof}
This statement is proven in a more general setting in \cite[Prop.\ 3.7, Cor.\
3.8]{DaNiOs2013}. They call the objects in point 1) Lagrangian algebras.
We use that by Remark \ref{rmk:alphaInduction} and \ref{rmk:Dyslexic} (see also
\cite[Thm 3.1]{Mg2010}) the category $\s{\M_+}{\M_+}^0$
is equivalent to the category of dyslectic modules.
\end{proof}

We note that there can exist inequivalent $\phi_1,\phi_2$ 
giving the same modular invariant $Z=(Z_{\mu\nu})$. Namely if 
$\langle\sigma_{\phi_1(\tau)},\mu\rangle=\langle\sigma_{\phi_2(\tau)},\mu\rangle$
holds for all $\tau \in \s{\M_+}{\M_+}^0$ and $\mu\in\NNs$ for which 
$b^+_{\tau,\mu}\neq 0$. 
Because $\phi_1$ and $\phi_2$ are inequivalent 
the Q-systems $\Theta^{\phi_1}_\LR$ and $\Theta^{\phi_2}_\LR$ are inequivalent.
This (or using Prop.\ \ref{prop:ModularQSystems}) implies that also
$\Theta_{(\Theta_+,\Theta_-,\phi_1)}$ and
$\Theta_{(\Theta_+,\Theta_-,\phi_2)}$ are inequivalent. This means that the
second cohomology (see Rem.\ \ref{rmk:2ndCom}) of
$\Theta_{(\Theta_+,\Theta_-,\phi_{1,2})}$ does not vanish in this case.

\begin{example}
 Let us consider for $\NNs$ the UMTC obtained by $\SU(3)_9$ 
  and $\Theta_+$ coming from
  the conformal inclusion $\SU(3)_9\subset E_{6,1}$.
 
  As in Ex.\ \ref{ex:E61} 
  the UMTC category $E_{6,1}$ has three sectors 
  $\D{\M_+}{\M_+}^0=\{\beta_0,\beta_1,\beta_2\}$ and
  we obtain an extension $\M_+\subset \tilde \M$ with
  $[\tilde\theta]=[\beta_0]\oplus[\beta_1]\oplus[\beta_2]$,
  which gives the permutation modular invariant interchanging
  $\beta_1\leftrightarrow\beta_2$.
  Now $\sigma^+_{\beta_1}=\sigma^+_{\beta_2}$, 
  so both inclusions $\N\subset \M_+$ and $\N\subset \tilde \M$ 
  give by the above discussion the same modular invariant 
  with respect to ${\SU(3)_9}$, which is 
  $Z=|\chi_{0,0}+\chi_{9,0}+\chi_{0,9}+
  \chi_{4,1}+\chi_{1,4}+\chi_{4,4}|^2+2|\chi_{2,2}+\chi_{5,2}+\chi_{2,5}|^2$.
This example appeared in \cite{BcEv2001}, \cf \cite{EvPu2009,EvPu2011}.

So we can conclude that  
  $\Theta_{(\Theta_+,\Theta_+,\id)}$ and 
  $\Theta_{(\Theta_+,\Theta_+,\phi)}$ in $\NNs\boxtimes\overline{\NNs}$ 
  have isomorphic endomorphisms
$[\theta_{(\Theta_+,\Theta_+,\id)}]=[\theta_{(\Theta_+,\Theta_+,\phi)}]$ but the
Q-systems are 
  not equivalent. So we have an example where 
  the second cohomology does not vanish.

  The same happens for the inclusion\footnote{This was told to us by V. Ostrik via
mathoverflow}
	$G_{2,3}\subset E_{6,1}$ where $Z=|\chi_{00}+\chi_{11}|^2+2|\chi_{02}|^2$.
\end{example}

\section{Conformal nets}
\label{sec:CNet}
We now apply the results to conformal nets.

Let $\RRc=\RR\cup\{\infty\}$ be the one-point compactification of the real line
$\RR$, which we can by the Cayley map $\RRc\ni x\mapsto z= \frac{\ima-x}{\ima+x}\in\Sc$ identify with the circle $\Sc\subset\CC$.
We denote by $\Mob$ the $\textbf{M\"obius group}$ 
which is isomorphic to both:
\begin{itemize}
  \item $\PSL(2,\RR)$, which acts naturally
    on the real line $\RRc$, and 
  \item $\PSU(1,1)$, which acts naturally on the circle 
$\Sc\subset \CC$.
\end{itemize}
The universal covering group of $\Mob$ is denoted by $\Mobc$.
We denote by $\Mob_\pm=\Mob\rtimes \ZZ_2$ where 
the action of $\ZZ_2$ is given by the reflection $r\colon z\mapsto \bar z$ on $\Sc$.
The \textbf{rotations}
$R(\vartheta)z=\e^{\ima\vartheta z}$ on $\Sc$,
the \textbf{dilations} $\delta(s)x= \e^s x$ on $\RR$, and the
\textbf{translations}
$\tau(t)x=x+t$ on $\RR$ give three distinguished
one-parameter subgroups of $\Mob$
which generate $\Mob$.

We denote by $I\in\cI$ the set of all \textbf{proper intervals} on $\Sc$,
\ie all open, connected, non-dense, non-empty intervals $I\subset \Sc$.

\begin{defi}
    A \emph{local M\"obius covariant net (\CFTnet)} $\A$ on $\Sc$ 
    is a family $\{\A(I) \}_{I\in \cI}$ of von Neumann algebras on a Hilbert 
    space $\Hil_\A$, with the following properties:
    \begin{enumerate}[{\bf A.}]
        \item \textbf{Isotony.} $I_1\subset I_2$ implies 
            $\A(I_1)\subset \A(I_2)$.
        \item \textbf{Locality.} $I_1  \cap I_2 = \emptyset$ implies 
            $[\A(I_1),\A(I_2)]=\{0\}$.
        \item \textbf{M\"obius covariance.} There is a unitary representation
            $U$ of $\Mob$ on $\Hil$ such that 
            $  U(g)\A(I)U(g)^\ast = \A(gI)$.
        \item \textbf{Positivity of energy.} $U$ is a positive energy 
            representation, i.e. the generator $L_0$ (conformal Hamiltonian) 
            of the rotation subgroup $U(R(\theta))=\e^{\ima \theta L_0}$
            has positive spectrum.
        \item \textbf{Vacuum.} There is a (up to phase) unique rotation 
            invariant unit vector $\Omega \in \Hil$ which is 
            cyclic for the von Neumann algebra $\bigvee_{I\in\cI} \A(I)$.
    \end{enumerate}
\end{defi}
The \textbf{Reeh--Schlieder property} automatically holds
\cite{FrJr1996}, \ie $\Omega$ is cyclic and separating for any $\A(I)$ with 
$I\in\cI$. 
Furthermore, we have the 
    \textbf{Bisognano--Wichmann property} 
\cite{GaFr1993,BrGuLo1993} saying that the modular operators with respect to 
$\Omega$ have geometric meaning; \eg the modular operators for the 
upper circle $I_0$ are given by the dilation 
$\Delta^{\ima t}=U(\delta(-2\pi t))$ and reflection $J=U(r)$, where here $U$ is 
extended to $\Mob_\pm$. 
For a general interval $I\in\cI$ the modular operators are given by a special
conformal transformation $\delta_I$ and a reflection $r_I$ both fixing the 
endpoints of $I$. 
The Bisognano--Wichmann property implies 
  \textbf{Haag duality} 
\begin{align*}
    \A(I)'&=\A(I') &I\in\cI 
\end{align*}
and it can be shown (see \eg \cite{GaFr1993}) that each $\A(I)$ is a type 
\threeone factor in Connes' classification \cite{Co1973}.
A conformal net is 
    \textbf{additive} 
\cite{FrJr1996}, \ie for intervals $I\in\cI$ and $I_1,\ldots, I_n \in \cI$
we have
$$
    I \subset \bigcup_i I_i \quad\Longrightarrow\quad 
    \A(I) \subset \bigvee_{i} \A(I_i)
    \punkt
$$

A local M\"obius covariant net on $\A$ on $\Sc$ is called \textbf{completely 
rational} if it 
\begin{enumerate}[{\bf A.}]
  \setcounter{enumi}{5}
  \item fulfills the 
    \textbf{split property}, \ie 
    for $I_0,I\in \cI$ with $\overline{I_0}\subset I$ the inclusion 
    $\A(I_0) \subset \A(I)$ is a split inclusion, namely there exists an 
    intermediate type I factor $M$ such that $\A(I_0) \subset M \subset \A(I)$.
  \item is 
 \textbf{strongly additive}, \ie for $I_1,I_2 \in \cI$ two adjacent intervals
obtained by removing a single point from an interval $I\in\cI$ 
the equality $\A(I_1) \vee \A(I_2) =\A(I)$ holds.
  \item for $I_1,I_3 \in \cI$ two intervals with disjoint closure and 
    $I_2,I_4\in\cI$  the two components of $(I_1\cup I_3)'$, the 
    \textbf{$\mu$-index} of $\A$
    \begin{equation}
      \mu(\A):= [(\A(I_2) \vee \A(I_4))': \A(I_1)\vee \A(I_3) ]
      \label{eq:twointerval}
    \end{equation}
    (which does not depend on the intervals $I_i$) is finite.
\end{enumerate}
\begin{example}
Examples of completely rational 
local M\"obius covariant nets are:
\begin{itemize}
  \item Diffeomorphism covariant nets with central charge $c<1$
    \cite{KaLo2004}.
  \item The nets $\A_L$ where $L$ is a positive even lattice
    \cite{DoXu2006} which contain as a special case \cite{Bi2012}
    loop group nets $\A_{G,1}$ at level 1 for $G$ a 
    compact connected, simply connected  simply-laced Lie group.
  \item The loop group nets $\A_{\SU(n),\ell}$ for $\SU(n)$ at level $\ell$.
    \cite{Xu2000}.
\end{itemize}
Further examples of rational conformal nets can be obtained from these as follows:
\begin{itemize}
  \item 
Finite index extensions and subnets of completely rational conformal
nets. Namely, let $\A\subset \cB$ be a finite subnet \ie $[\cB(I):\A(I)]<\infty$ for
some (then all) $I\in\cI$, then $\A$ is completely rational iff $\cB$ is
completely rational \cite{Lo2003}, in particular orbifolds $\A^G$ of completely
rational nets $\A$ with $G$ a
finite group are completely rational.
  \item Let $\A\subset\cB$ be a co-finite subnet , \ie
$[\cB(I),\A(I)\vee\A^\mathrm{c}(I)] <\infty$ for some (then all) $I\in\cI$,
where the \textbf{coset net} $\A^\mathrm{c}$ is defined by
$\A^\mathrm{c}(I)=\A'\cap\cB(I)$ with $\A'=(\vee_{I\in\cI}\A(I))'$. Then
$\cB$ is completely rational iff $\A$ and $\A^\mathrm{c}$ are completely
rational \cite{Lo2003}. This gives many example of completely rational nets
coming from the coset construction.
\end{itemize}
\end{example}

A \textbf{separable (non-degenerated) representation} 
of a strongly additive local M\"obius covariant net is a family 
$\pi=\{\pi_I\colon\A(I)\to \B(\Hil_\pi)\}_{I\in\cI}$ of unital 
representations ($\ast$-homomorphisms) $\pi_I$ of $\A(I)$
on a common separable Hilbert space $\Hil_\pi$, which are compatible, \ie
\begin{align*}
  \pi_{I_2}\restriction \A(I_1) &= \pi_{I_1},  & I_1\subset I_2
  \punkt
\end{align*}
Such a representation is automatically normal, 
\ie all $\pi_I$ are strongly continuous.
We denote by $\DHR(\A)$ the category of separable representations, 
where morphisms in $\Hom(\pi^1,\pi^2)$ are given by intertwiners 
$V\in\B(\Hil_{\pi^1},\Hil_{\pi^2})$,
such that $V\pi^1_I(a)=\pi^2_I(a)V$ for all $I\in \cI$ and $a\in\A(I)$.
Let us denote by $\DHR^0(\A)$ the representations $\pi$ with finite statistical
dimension $d\pi$, which is defined to be 
$$
  d\pi:= [\pi_{I'}(\A(I'))':\pi_I(\A(I))]^{\frac12}
$$
for some $I\in\cI$, where $[M:N]$ is the minimal index. The definition of $d\pi$ does not depend on the choice of $I$.

Let us from now on fix a completely rational local M\"obius covariant net
$\A$ on $\Sc$. The category $\DHR^0(\A)$ 
is a (unitary) modular tensor category \cite{KaLoMg2001}. 
Every $\pi\in\DHR^0(\A)$ is equivalent to a representation 
localized in a given $I_0\in\cI$, \ie it exists a $\rho\cong \pi$
such that $\Hil_\rho=\Hil_\A$ and $\rho_{I_0'}=\id_{\A(I_0')}$. 
Namely, $\pi_{I_0'}(\A(I_0'))$ on $\Hil_\pi$ is spatially 
isomorphic to $\A(I_0')$ on $\Hil_\A$, by the type III property. 
Let $U\colon\Hil_\pi\to\Hil_\A$ be a unitary implementing this isomorphism, then 
$\rho=\{\rho_I:=\Ad U\circ \pi_I\}_{I\in\cI}$ does the job.

This implies that the category $\DHR^{I_0}(\A)$ of 
representations with finite statistical dimensions which are 
localized in $I_0$ has the same irreducible sectors as $\DHR^0(\A)$. 

By Haag duality 
$\rho\in\DHR^{I_0}(\A)$ implies $\rho_I(\A(I))\subset\A(I)$ for every $I\supset I_0$, that means 
such a representation is an endomorphism and 
$d\rho=[\A(I_0):\rho_{I_0}(\A(I_0))]^{\frac12}$ equals the dimension of the
endomorphism.
Together with strong additivity it follows that
all intertwiners are in $\A(I_0)$. In particular, this means that
$\DHR^{I_0}(\A)$ can naturally be seen as a full subcategory of $\End(\A(I_0))$
and that $\DHR^{I_0}(\A)$ is equivalent to $\DHR^0(\A)$. We note that the family $\{\rho_I\}$ is determined by $\rho_{I_0}$ by using strong
additivity and it is really enough to consider $\DHR^{I_0}(\A)$ as a full and
replete subcategory of $\End(\A(I_0))$ and we will drop the index $I_0$. 
Repleteness is just the fact that for $U\in \A(I_0)$ also $\Ad_U\circ\rho$ is localized in $I_0$.

The \textbf{braiding} (also called statistics operator) is given by:
\begin{align*}
  \varepsilon(\rho_1,\rho_2)= \rho_2(U_1^\ast)U_2^\ast U_1\rho_1(U_2)\komma
\end{align*}
where $U_i\in\Hom(\rho_i,\tilde\rho_i)$ and $\tilde\rho_i\in[\rho_i]$ is localized
in $I_i$. Here $I_1,I_2 \subset I_0$ are two disjoint intervals 
such that $I_1>I_2$ ($I_2$ sits clockwise after $I_1$ inside $I_0$). 
We also write $\varepsilon^+$ for $\varepsilon$ 
and define the opposite braiding by
$\varepsilon^-(\rho_1,\rho_2)=\varepsilon^+(\rho_2,\rho_1)^\ast$.

We will interpret $\A$ as the chiral observables or
as chiral symmetries. For example $\A=\Vir_c$ with
$c<1$ is the net generated by the chiral stress energy tensor $T(x)$.
We want to look into CFTs on Minkowski space containing the
chiral observables $\A$ and boundary conditions on $\MM_+$ which ``preserve''
these observables. 

\subsection{Extensions and Q-systems}
Let $M$ be a spacetime, \eg Minkowski space and $\cK$ a set of open spacetime
regions in $M$, \eg the set of double cones. Let $G$ be a group acting locally 
on $M$ and let $G(O)$ be the set of all $g\in G$, such that there is a continuous path
$\gamma$ in $G$  from the identity 
to $g$ such that $\gamma(t)O\in \cK$.
\begin{defi}
    A \emph{local $G$-covariant net} $\A$ on $M$ 
    is a family $\{\A(O) \}_{O\in \cK}$ of von Neumann algebras on a Hilbert 
    space $\Hil$, with the following properties:
    \begin{enumerate}[{\bf A.}]
        \item \textbf{Isotony.} $O_1\subset O_2$ implies 
            $\A(O_1)\subset \A(O_2)$.
        \item \textbf{Locality.} 
            $[\A(O_1),\A(O_2)]=\{0\}$ for all pairwise spacelike separated
$O_1,O_2\in \cK$.
        \item \textbf{$G$-covariance.} There is a unitary positive energy representation
            $U$ of $G$ on $\Hil$, such that
            $U(g)\A(O)U(g)^\ast = \A(gO)$ for all $g\in G(O)$
        \item \textbf{Vacuum.} There is a (up to phase) unique $G$- 
            invariant unit vector $\Omega \in \Hil$ which is 
            cyclic and separating for $\A(O)$ for all $O\in \cK$.
    \end{enumerate}
\end{defi}

A \textbf{$G$-covariant DHR representation} of $\A$ is a compatible family 
$\pi=\{\pi_O\colon\A(O)\to \cB(\Hil_\pi)\}_{O\in\cK}$ of representations on a Hilbert space $\Hil_\pi$, such that for all $O \in \cK$ there exists 
a unitary $V\colon \Hil_\pi\to \Hil$, such that the representation $\rho:=\Ad
V\circ \pi$ is localized in $O$, \ie $\rho_{O_0}=\id_{\A(O_0)}$ for $O_0$
spacelike to $O$,  and that there is a unitary projective
representation $U_\pi$ of $G$, such that $\Ad U_\pi(g)\circ \pi_O = 
\pi_{gO}\circ \Ad U(g)$ for all $g\in G(O)$.

Given two local $G$-covariant nets $\A$ and $\cB$ on Hilbert spaces $\Hil_\A$
and $\Hil_\cB$, respectively,
an \textbf{arrow} $\A \to \cB$ is an isometry $V\colon \Hil_\A\to\Hil_\cB$ and a compatible family  of embeddings
(representation)
$\{\pi_O\colon \A(O)\hookrightarrow \cB(O)\}$
such that for all $O\in\cK$ we have $Va=\pi_O(a)V$,
$VU_\A(g)=U_\cB(g) V$ for all $g\in G$  and $V\Omega_\A=\Omega_\cB$.

$\A$ and $\cB$ are called \textbf{unitary equivalent} if $V$ is a unitary
and $\pi_O$ are isomorphisms.

Let us assume that we have a subnet $\A_0$ of $\cB$, \ie $\A_0(O)\subset\cB(O)$ for all $O$ and 
$U(g)\A_0(O)U(g)^\ast=\A_0(gO)$. Then $\A=\A_0e$ with $e$ the Jones projection on
$\overline{ \vee \A_0(O)\Omega }$ is a $G$-local net on $\Hil_\A:=e\Hil$, in
other words we have an arrow $\A\to \cB$ in the above sense.
We say that $\A$ is a \textbf{subnet} of $\cB$ and $\cB$ is a
\textbf{local extension} of $\A$. By abuse of notation we will not distinguish 
between the net $\A$ and its representation on the bigger Hilbert space $\H$ 
and write $\A\subset\cB$ or $\cB\supset \A$ for an inclusion/extension of nets.

For every connected region we have a subfactor $\A(O)\subset \cB(O)$.
If the 
subfactor is irreducible, we call the extension \textbf{irreducible}
and if the index is finite we call the extension \textbf{finite}.
If we have a finite irreducible extension $\cB$ of $\A$ then  
the corresponding Q-system of $\A(O)\subset \cB(O)$ is a \local irreducible Q-system in $\DHR^O(\A)$
and conversely if we have a \local irreducible Q-system $\Theta$ in $\DHR^O(\A)$
we obtain a finite local extension $\cB$ of $\A$. 
In particular we have a one-to-one
correspondence between \cite{LoRe1995}:
\begin{itemize}
        \item local finite irreducible extensions $\cB\supset \A$ up to unitary
equivalence and
        \item \local irreducible Q-systems $\Theta$ in $\DHR^O(\A)$ up to
equivalvence.
\end{itemize}

If we assume $\Theta$ to be only irreducible, we still have a relatively local 
extension, \ie $[\A(O_1),\cB(O_2)]=\{0\}$ for $O_1$ and $O_2$ spacelike
separated. We call such an extension $\cB\supset\A$ also non-local extension to stress the fact
that we do not assume locality of $\cB$. There is a one-to-one
correspondence between \cite{LoRe1995}:
\begin{itemize}
        \item finite irreducible extensions $\cB\supset \A$ up to unitary
equivalence and
        \item irreducible Q-systems $\Theta$ in $\DHR^O(\A)$ up to
equivalence.
\end{itemize}

\subsection{Representation theory of local extensions}

The following is well-known to experts \cite{Mg2010}.

\begin{prop}
\label{prop:folk}
Let $\A\subset \cB$ a finite index inclusion of local M\"obius covariant nets on
        $\Sc$ and let either net be completely rational.
Then $\A$ and $\cB$ are both completely rational and the 
inclusion is irreducible. 

Further, let $I\in\cI$ be an interval $N:=\A(I) \subset\cB(I)=:M$ and $\NNs=\DHR^I(\A)$, 
and $\Theta$ be the Q-system in $\NNs$ associated with $N\subset M$.
Then $\DHR^I(\cB)=\MMs^0$ as UMTCs and in particular $\DHR(\cB)$ is equivalent 
to $\Mod^0(\Theta)$ and $\Bim^0(\Theta,\Theta)$.
\end{prop}
\begin{proof}
Both $\MMs^0$ and $\DHR^I(\cB)$ being full and replete subcategories
of $\End(M)$, the only thing which needs to be checked is that both have the
same irreducible sectors. The braiding on $\MMs^0$ can be checked to give the braiding
on $\Rep(\cB)$ since the braiding is fixed by
the universal property $\varepsilon(\rho_1,\rho_2)=1$ if $I_2$ sits clockwise after $I_1$ inside $I$.
A sector $[\beta]\in \MDM$ is a DHR sector if and only it is in $\MDM^0$ (see
\cite{LoRe1995,BcEv1998}), which
implies $\MMs^0\subset\DHR^I(\cB)$. To see equality, we realize that global
dimensions coincide, namely 
$\dim\DHR^I(\cB)\equiv \mu(\cB)=[M:N]^{-2} \mu(\A) \equiv \dim\NNs /(d\theta)^{2}$ by
\cite{KaLoMg2001}
and 
$\dim\MMs^0= \dim\NNs/ (d\theta)^{2}$
by Lemma \ref{lem:leq}.
\end{proof}
\begin{rmk}
Commutative Q-systems $\Theta$ in a UMTC $\NNs$ are also called \textbf{quantum subgroups}, so
finding quantum subgroups in a given UMTCs $\NNs$
and finding finite index local extensions of a local M\"obius covariant $\A$ 
net with $\DHR^0(\A)\cong \NNs$ is equivalent. 
The representation theory of the extensions can be completely understood on a categorical
level.

An analogous statement for inclusions of rational VOAs appeared recently in 
\cite{HuKiLe2014}.
\end{rmk}

\subsection{Maximal 2D nets with chiral observables \texorpdfstring{$\A$}{A}}
Let $\A$ be a  local M\"obius covariant net on $\Sc\cong \RRc$.
By restriction we can and will see $\A$ as a net on
$\RR$. Then Haag duality of $\A$ on $\RR$ is equivalent to strong
additivity of $\A$. We will assume that $\A$ is completely rational, therefore
this holds automatically.

We denote by $\MM$ the two-dimensional Minkowski space and by $\cK$ the set of 
double cones $O\subset \MM$.  
Each double cone is of the form 
$$O=I\times J:= \{(t,x):t-x\in I,~t+x\in J\}, $$
where $I,J\in\cI_0$ are two intervals on the light-rays
$L_\pm=\{(t,x): t\pm x=0\}$.

The action of $\Mob\cong\PSL(2,\RR)$ on $\RRc$ gives a local action
of $\Mobc$ on $\RR$ as in \cite{KaLo2004}. 
We define $\gG_2=\Mobc\times \Mobc$ which acts locally on Minkowski space $\MM$.

For $O\in \cK$ we denote by $\gG_2(O)$ all $g\in\gG_2$ such that there is a
path $\gamma \colon [0,1]\to \gG_2$ from the identity element $e$ to $g$ 
with $\gamma(t)O\subset \M$ for all $t\in [0,1]$.

We denote by $\A_2$ the net on $\Hil_\A\otimes\Hil_\A$ given by 
$$
  \A_2(I\times J):=\A(I)\otimes \A(J).
$$
It is a local M\"obius covariant net on $\M$ as in \cite{KaLo2004}.
Every DHR representation of $\A_2$ with finite index is a direct sum 
of representations of the form $\rho\otimes\sigma$ where 
$\rho\in\DHR(\A)$ and $\sigma\in\DHR(\A)$. The braiding is given 
by $\varepsilon(\rho_1\otimes\sigma_1,\rho_2\otimes\sigma_2)=
\varepsilon^+(\rho_1,\rho_2)\otimes\varepsilon^-(\sigma_1,\sigma_2)$.
Therefore the category of DHR representations  of $\A_2$ 
with finite statistical dimensions is equivalent 
to $\DHR^I(\A)\boxtimes \overline{\DHR^J(\A)}$.

Let us write $\cB_2\supset \A_2$ for a local, M\"obius covariant, irreducible 
extension of $\A_2$, \ie a local M\"obius covariant net $\cB_2$ 
on Minkowski space $\MM$  on the Hilbert space $\H_{\cB_2}$ 
with irreducible vacuum vector $\Omega$ which is extending 
$\A_2\cong \A\otimes \A$, more precisely there is a representation $\pi$ of 
$\A_2$ on $\H_{\cB_2}$, such that 
$\pi(\A_2(O))\subset \cB_2(O)$ is an irreducible inclusion
of factors and $U(g)\pi(\A(O))U(g)^\ast = \pi(\A(gO))$
for all double cones $O\in\cK$ and all $g\in \gG(O)$.
By abuse of notation we will omit the $\pi$.

We remember that there is 
a one-to-one correspondence between 
local irreducible extensions $\cB_2\supset \A_2$ (up to unitary equivalence) and
irreducible \local Q-systems $\Theta_2$
in $\DHR^I(\A)\boxtimes \overline{\DHR^J(\A)}$ (up to equivalence).

\begin{prop}
        \label{prop:2Dmax}
	Let $\cB_2\supset \A_2$ be a local extension. Then the following statements 
	are equivalent:
	\begin{enumerate}
		\item The net $\cB_2$ is a maximal local irreducible extension,
		  \ie if $\tilde \cB_2\supset \cB_2$ is a local irreducible extension, then 
      $\cB_2= \tilde \cB_2$.		
		\item The index $[\cB_2:\A_2]=\mu_2(\A)\equiv \dim(\DHR(\A))$.
		\item The matrix $(Z_{\lambda\mu})$ is a modular invariant.
		\item The $\mu$-index of $\cB_2$ is 1.
		\item The net $\cB_2$ has no non-trivial superselection
		sectors.
	\end{enumerate}
\end{prop}
\begin{proof}
    To show (2) $\Rightarrow$ (1) let $\Theta_2$ be a Q-system in $\DHR^I(\A)\boxtimes\overline{\DHR^J(\A)}$
    	giving the extension $\A(I)\otimes\A(J)\subset \cB_2(I\times J)$ and let 
      us assume that $[\cB_2(I\times J):\A(I)\otimes\A(J)]=\mu_2(\A)$.
    	By 
    	Lemma \ref{lem:leq} we have the following inequality:
      \begin{align*}
      	d\Theta_2\equiv [\cB_2:\A_2] &\leq \dim(\DHR(\A\otimes \A))^\frac12
      	\equiv \dim\left(\DHR(\A)\boxtimes\overline{\DHR(\A)}\right)^\frac12
      	\\&=\dim(\DHR(\A))\equiv \mu_2(\A)
      	\punkt
      \end{align*}
       This implies maximality.

For showing (1) $\Rightarrow$ (2), let us assume that $[\cB_2:\A_2]<\mu_2(\A)$. 
      We need to show that there is an extension $\tilde\cB_2\supsetneq\cB_2$.
      This we obtain by adding the boundary \cite{CaKaLo2013}, \ie from $\cB_2$ 
      we obtain a possible reducible boundary net (see Subsec.\ 
      \ref{sec:AddingBoundary}) of which we choose an irreducible subnet $\cB_+$. 
      We claim $\cB_+$ cannot be Haag dual, but this follows because
      $[\cB_+:\A_+]=[\cB_2:\A_2]<\mu_2(\A)$ and then \cite[Prop.~2.13]{LoRe2004}
      implies $[\cB_+^\mathrm{d}:\cB_+]>1$.
      So we have an inclusion $\A_+\subset\cB_+\subsetneq \cB^\mathrm{d}_+$ and
      a corresponding locally isomorphic inclusion $\A_2\subset\cB_2\subsetneq
      \tilde \cB_2$ as in \cite{LoRe2004}, in particular $\cB_2$ was not 
      maximal.

      The statements (2) and (3) are equivalent by Prop.\
\ref{prop:QSystemModular} and the implication (5) $\Rightarrow$ (1) is clear.

  	(2) $\Rightarrow$ (4) follows by calculating the $\mu$ index
\cite{KaLoMg2001} and
likewise  the implication (4) $\Rightarrow$ (5) is \cite[Corollary 32]{KaLoMg2001}.
\end{proof}
\begin{prop} 
	\label{prop:class2Dmax}
        There is a one-to-one correspondence between:
  \begin{enumerate}
    \item maximal local irreducible extensions $\cB_2\supset \A_2$ 
      up to unitary equivalence.
    \item $\Theta_2$ \local irreducible Q-systems in 
      $\DHR^I(\A)\boxtimes \overline{\DHR^I(\A)}$ with $d\theta_2=\mu_2(\A)$
      up to equivalence.
    \item (Non-local) irreducible extensions $\cB\supset\A$ up to Morita equivalence.
    \item Irreducible Q-systems $\Theta$ in $\DHR^I(\A)$ up to Morita equivalence.
  	\item Indecomposable $\NNs$ module categories, 
      where $\N=\A(I)$ and $\NNs=DHR^I(\A)$.
  	\item Local chiral extensions $\A_\mathrm{L}\supset\A$, $\A_\mathrm{R}\supset\A$ together
      with a braided equivalence $\phi\colon\DHR(\A_\mathrm{L})\to\DHR(\A_\mathrm{R})$.
  \end{enumerate}
\end{prop}
\begin{proof}
  The correspondence between (1) and (2) is Prop.\ \ref{prop:2Dmax},
 the one between
  (3) and (4) \cite{LoRe1995}.
  Starting with (4) we obtain (2) by applying the full center and it is well
  defined on Morita equivalence classes and injective by
  Prop.~\ref{prop:KongRunkel}. It is surjective
  by Prop.\ \ref{prop:QSystemModular}, so (2) and (4) are equivalent.
  Equivalently, one can start with $\cB_2$ and add the boundary to obtain a Haag
  dual boundary net (as in the proof before) which correspond to a non-local
  extension. The $\alpha$-induction construction gives back the original net.

  The correspondence between (4), (5) and (6) is just 
  Prop.\ \ref{prop:ModularQSystems}, where (6) is (2) of Prop.\ \ref{prop:ModularQSystems} reformulated in the language
of nets, \cf \cite{Mg2010}.
\end{proof}
\begin{rmk}
	We know how the Morita equivalence looks like, see Subsec.\ \ref{sec:MoritaEquivalenceClass}.
\end{rmk}

\subsection{Boundary conditions}

Let $\A$ be a completely rational local M\"obius covariant net 
on $\Sc$, which we will see as a net on $\RR$ by restriction.
Let $\MM_+=\{(t,x)\in\MM: x>0\}$ be Minkowski half-plane and 
let $\cK_+$ be the set of double cones $O\Subset \MM_+$. Double cones $O\in
\cK_+$ are in 
one-to-one correspondence with pairs of proper intervals $I,J\subset \RR$ such that
$I<J$. We write $O=I\times J$.

Let $\A_+$ be the net on $\MM_+$ given by 
\begin{align*}
        \A_+(O) &= \A(I)\vee \A(J)  &O=I\times J
\end{align*}
which is locally covariant w.r.t.\ $G_+$ the universal covering of $\Mob$,
namely 
\begin{align*}
        U(g)\A_+(O )U(g)^\ast &= \A_+(gO) & g\in G_+(O)
\end{align*}
where $G_+$ acts locally on $O=I\times J \in \cK_+$
by $gO= gI\times gJ$ and $G_+(O)$ is the set of all 
$g\in G_+$ such that there is a continuous path $\gamma$ from the identity 
to $g$ such that $\gamma(t)O\in \cK_+$.

By the split property it follows that $\A_+(O)$ is spatially isomorphic 
to $\A_2(O)\equiv\A(I)\otimes\A(J)$. This implies that the net $\A_+$
is locally isomorphic to the net $\A_2$ restricted to $\MM_+$.

A boundary net $\cB_+$ associated with $\A$ is 
a local, (locally) $G_+$-covariant net $\cB_+$, which is an irreducible extension
$\cB_+\supset \A_+$.

Starting with $\cB_+\supset\A_+$, we define the generated net
$\cB_+^\mathrm{gen}\supset \A$
on $\RR$ by
$$
\cB_+^\mathrm{gen}(I)= \bigvee\limits_{\substack{O\in\cK_+\\O\subset W_I}} \cB_+(O) \supset \A(I)
    \komma
$$
where $W_I=\{(t,x) : t \pm x \in I\}$ is the left wedge, such that its 
intersection on the $t$-axis is $I$.

Conversely, given $\cB\supset\A$ a (non-local) extension on $\RR$,
we define 
$$
    \cB^\mathrm{ind}_+(O)=\cB(L)\cap\cB(K)'
    \komma
$$
where $O=I\times J$ and $L\Subset K$, such that 
$L\cap K'=I\cup J$ or equivalently $O=W_L\cap W_K'$.

The dual net is defined by $\cB_+^\mathrm{d}(O)=\cB_+(O')'$ and 
$\cB_+^\mathrm{d}=\cB_+$ if and only if $\cB_+$ is Haag dual.

Then $(\cB^\mathrm{ind}_+)^\mathrm{\mathrm{gen}} = \cB$
and $(\cB_+^\mathrm{gen})^\mathrm{ind}_+=\cB_+^\mathrm{d}=\cB_+$
provided $\cB_+$ was already Haag dual.

Together we have:
\begin{prop}[\cite{LoRe2004,LoRe1995}]
	There is a one-to-one correspondence between the equivalence classes of:
	\begin{enumerate}
		\item boundary nets  $\cB_+$ associated with $\A$, such that $\cB_+$ is Haag dual.
		\item boundary nets  $\cB_+$ associated with $\A$, such that
$\A_+\subset \cB_+$ is maximal.
		\item (Non-local) extensions $\cB\supset \A$ on $\RR$.
		\item Q-systems in $\NNs$, where $\N=\A(I)$ and $\NNs=\DHR^I(\A)$.
	\end{enumerate}
\end{prop}

\begin{defi} 
	Let $\cB_2\supset \A_2$ be local extension, \ie a CFT on Minkowski space.	
	A 
    \textbf{(M\"obius covariant) boundary condition of $\cB_2\supset \A_2$ 
            with chiral symmetry $\A$} 
  is a unitary equivalence class of boundary nets $\cB_+\supset \A_+$, where
	$\cB_2\restriction \MM_+$ is locally covariantly isomorphic to $\cB_+$, more
  precisely there is a compatible family of isomorphisms 
  $\Phi_O\colon\cB_+(O)\to \cB_2(O)$ such that it restricts to an isomorphism
	$\A_+(O)\to \A_2(O)$ for all $O\in\cK_+$	and that $\Phi$ is covariant 
  respect to the covariance $U_{\cB_+}$ of $\overline{\Mob}$ and $U_{\cB_2}$ of
  $\overline{\Mob}\times \overline{\Mob}$ (where $\overline{\Mob}$ is the 
  diagonal subgroup of $\overline{\Mob}\times\overline{\Mob}$).
\end{defi}

\begin{prop}
        \label{prop:BC}
	Let $\cB_2\supset \A_2$ maximal and let 
	$\A\subset \cB$ given by Prop.\ \ref{prop:class2Dmax}.
	Then there is a one-to-one correspondence between:
	\begin{enumerate}
		\item Boundary conditions of $\cB_2\supset \A_2$ with chiral symmetry $\A$.
		\item Unitary equivalence classes of $\cB_a\supset \A$ Morita equivalent 
      to $\cB\supset \A$.
		\item Sectors in 
			$$
				\NMs/\mathrm{Pic}(\MMs)\komma
			$$
                       where $N=\A(I)$, $M=\cB(I)$ and $\NNs=\DHR^I(\A)$. 
	\end{enumerate} 
	In particular the number of boudary conditions 
	of $\cB_2\supset \A_2$ with chiral symmetry $\A$ is 
	less or equal than 
	$$
		|\NDM|\equiv \sum_{\lambda\in\NDN} Z_{\lambda\lambda}\punkt
	$$
\end{prop}
\begin{proof}
  The following diagram commutes \cite[Cor.~2]{LoRe2009}
  $$
    \begin{tikzpicture}[node distance=2cm, auto]
    	\node (C) {$\{\cB_+\supset\A_+ \text{ maximal}\}$};
    	\node (P) [below of=C] {$\{\cB\supset\A\}$};
    	\node (Ai) [right of=P,node distance=5cm] 
        {$\{\cB_2\supset \A_2\equiv\A\otimes\A\}$};
      \draw[->] (C) to node {removing the boundary} (Ai);
      \draw[<->] (C) to node [swap] {$\sim$} (P);
      \draw[->] (P) to node [swap] {$\alpha$-induction} (Ai);
    \end{tikzpicture}
	\punkt
  $$

  Given a boundary condition, \ie a boundary net $\cB_{a,+}\supset \A_+$ let
  $\cB_a\supset \A$ be the corresponding chiral extension.  We note that
  $\cB_{a,+}$ is Haag dual (\cf \cite[App. C]{LoRe2009}), because $\cB_2$ is
  modular invariant.  If we remove the boundary we obtain $\cB_2\supset \A_2$,
  because the extensions are locally isomorphic and therefore isomorphic, see
\cite{LoRe2009}.

  We conclude by commutativity of the above diagram that $\cB\supset \A$ and
  $\cB_a\supset \A$ are Morita equivalent, namely the $\alpha$-induction 
  construction gives equivalent two-dimensional extensions, which means the 
  full centers are equivalent, which is equivalent to the Morita equivalence of
  $\cB\supset \A$ and $\cB_a\supset \A$.

  Conversely, if we have given a chiral  extension $\cB_b\supset\A$ Morita
  equivalent to $\cB\supset \A$, then $\cB_{b,+}\supset \A_+$ is locally
  equivalent to $\cB_{b,2}\supset \A_2\restriction \MM_+$ obtained by
  $\alpha$-induction. But $\cB_{2,b}\supset \A_2$ is isomorphic to $\cB_2\supset
  \A_2$ by Morita equivalence, so we get a boundary condition (this follows 
   also from \cite{LoRe2004}, realizing that the DHR orbit exhausts the Morita 
  equivalence class).

  Choosing $N=\A(I)$, $M=\cB(I)$ and $\NNs=\DHR^I(\A)$ the Q-systems $\Theta_a$
  corresponding to $\cB_a\supset\A$ which is Morita equivalent to
  $\cB\supset\A$ are in one-to-one correspondence with
  $\NMs/\mathrm{Pic}(\MMs)$ by Prop.\ \ref{prop:ClassDHROrbit}.
\end{proof}
\begin{example}
  We can give several cases as an example.
  \begin{itemize}     \item If $\A$ is holomorphic, \ie $\DHR(\A)$ just contains the vacuum
      sector or equivalently $\mu(\A)=1$, then $\cB_2=\A_2$ is maximal and the only 2D net and 
      $\A_+$ is the only boundary
      condition.
      The family of holomorphic nets contains for example the conformal nets
    $\A_L$ associated with even selfdual lattices \cite{DoXu2006} like the $E_8$
    lattice, Leech lattice etc., 
    the Moonshine net $\A_\natural$ \cite{KaLo2006} and certain framed nets 
    \cite{KaSu2012}.
    \item For $\A$ from the family of conformal nets, for which $\DHR(\A)$ is pointed, 
    it follows from Lemma \ref{lem:pointed} that there is always just one boundary condition for each $\cB_2\supset \A_2$.
    This family for example contains all conformal nets
    $\A_L$ coming from an even lattice $L$ \cite{DoXu2006}, which include
    all loop group conformal nets $\A_{G,1}$ of compact, connected, simply
    connected,  simply laced Lie groups $G$ 
    (the simple one being in one-to-one correspondence with A-D-E Dynkin
diagrams)
    at level 1 \cite{Bi2012}.
    \item If $\A$ is any completely rational net and 
    $\cB_2=\A_\LR\supset \A_2$ given by the trivial Longo-Rehren
    extension, then $\NMs\cong\NNs\cong\DHR(\A)$ and the boundary conditions 
    are given by DHR sectors of $\A$ modulo DHR automorphisms of $\A$.
    This case is sometimes also called the Cardy case.
    \item For $\A=\A_{\SU(2),k}$ the two-dimensional extensions are
    in one-to-one correspondence with
Dynkin diagrams of A-D-E type 
 with Coxeter number $k+2$. The boundary conditions are given by orbits $[\nu]$
of a marked vertex $\nu$ under the automorphism group 
of the Dynkin diagram \cf \cite{KaLoPeRe2007}.
    \item For $\A=\Vir_c$ with $c< 1$, the only possible values for $c$ are  
$c = 1 - 6/m(m + 1)$ with $m = 2, 3, 4, \ldots$.
The maximal two-dimensional extensions are in one-to-one correspondence with
pairs $(G_1,G_2)$ of Dynkin diagrams of A-D-E type 
with Coxeter number $m$ and $m+1$, respectively, \cf \cite{KaLo2004-2}.
        The boundary conditions are given by pairs $([\nu_1],[\nu_2])$ with
        $[\nu_i]$ the orbit of a marked vertex on $G_i$ under the automorphism
        group of $G_i$ ($i=1,2$). This result now follows also from \cite{KaLoPeRe2007}.
\end{itemize}
\end{example}
The invertible objects (automorphisms) in $\MMs$ have to do with invertible
defects (see for an interpretation of invertible defects in a different
framework \cite{DaKoRu2012}). 

The difference between two inequivalent $a,b\in\NMs$ related by an invertible
$\beta\in\MMs$ gets important if we also consider also reducible boundary
conditions in the next section.

\subsection{Reducible boundary conditions}

With the notation as before, let us assume $\cB_2\supset\A_2$ is a maximal extension of 
$\A_2$. Using Prop.~\ref{prop:class2Dmax} we can choose a (non-local) extension 
$\cB\supset \A$  such that $\cB_2$ is given by the $\alpha$-induction
construction of $\cB\supset \A$.

Let $I$ be an interval, $N=\A(I)$, $\NNs=\DHR^I(\A)$,
$M=\cB(I)$ and $\Theta$ the Q-system in $\NNs$ giving $N\subset M$.
Then every $a\in\NMs$ gives a in general reducible Q-system $\Theta_a$
and an extension $\cB_a\supset \A$.

We can define as before 
$$\cB_{a+}(O)=\cB_{a}(L)\cap\cB_{a}(K)'\punkt $$

This net fulfills all the properties 
of a boundary CFT in \cite{LoRe2004}, but the uniqueness 
of the vacuum and the joint irreducibility.

\begin{prop} Let $a\in\NMs$ possibly reducible. Then 
the (reducible) boundary net $\cB_{a,}+\supset\A_+$
is a (reducible) boundary condition for 
$\cB_2\supset\A_2$, which is given by the Q-system $Z(\Theta_a)$. 
\end{prop}
\begin{proof} If $a$ is irreducible this is already proven. 

Let $a$ be reducible and let 
$\Theta_a=\bar\iota\iota$ be the Q-system with inclusion $\iota(\A(I))\subset
\cB_a(I)$.
Let $\{p_i\}_{i=1}^n$ be a set of 
minimal projections in $\iota(\A(I))' \cap \cB_{a}(I)=\Hom(\iota,\iota)$ with $\sum_{i=1}^n
p_i=1$ with
corresponding morphisms
$\iota_i\prec \iota$. By the usually Reeh--Schlieder argument, the projection
do not depend on the choice of $I$.
The inclusion $\iota(\A(I))\subset\cB_a(I)$ is conjugated to
\begin{align*}
\left\{\begin{pmatrix} \iota_1(a) & &\\
                                        & \ddots &\\   
                && \iota_n(a)
        \end{pmatrix}
        : a\in\A(I)\right\} \subset \cB_a(I)\otimes M_n(\CC) \cong \cB_a(I)
        \punkt
\end{align*} 
With the same notation $\A_+(O)\subset\cB_{a,+}(O)$ is conjugated to:
\begin{align}
        \label{eq:inclusion}
        \left\{\begin{pmatrix} \iota_1(a) & &\\
                                        & \ddots &\\   
                && \iota_n(a)
        \end{pmatrix}
        : a\in\A_+(O)\right\} \subset
        \left\{\begin{pmatrix} b & &\\
                                        & \ddots &\\   
                && b
        \end{pmatrix}
        : b\in\cB_{a,+}(O)\right\} 
        \punkt
\end{align}

Because $\Theta_2:=Z(\Theta_a)$ and $Z(\Theta_{\bar \iota_i})$ are equivalent
(by Prop.\ \ref{prop:KongRunkel}) 
every $\cB_{i,+}\supset\A_+$ is a boundary condition for
$\cB_{2}\supset\A_{2}$. But then also the inclusion $\cB_2\supset\A_2$ is locally 
isomorphic to $\cB_{a,+}\supset\A_+$ by (\ref{eq:inclusion}) and the isomorphism restricted to $\A_2$
 gives a local isomorphism of $\A_2$ restricted to $\MM_+$ and $\A_+$. 
\end{proof}
Note that in the reducible case the vacuum $\Omega$ of $\cB_+$ is neither cyclic 
nor unique and that $\Omega=\sum_{i=1}^n\Omega_i$ with $\Omega_i=p_i\Omega$. The
restriction of $\cB_+$ to the subspace $\overline{\cB_+(O)\Omega_i}$ is unitarily
equivalent to the boundary condition coming from $\iota_i$.
In other words, $\NMs\ni a\mapsto \cB_{a,+}$ maps direct sums of sectors to
direct sums of boundary conditions.

\begin{example}
        Consider $a,b\in\NMs$ irreducible and mutually inequivalent but related by an automorphism
$\beta\in\MMs$, or equivalently $\Theta_a\cong\Theta_b$. 
        This means the boundary conditions
        coming from $a$ and $b$ are the same, but for example the boundary
conditions coming from
        $c:=a\oplus a$ and $d:=a\oplus b$ are different. This can be seen
        for example by regarding the relative commutants of the subfactors
        associated with $\Theta_c$ and $\Theta_d$, namely
        $\bar c(N)'\cap N\cong \CC\oplus \CC$, while $\bar d(N)'\cap N \cong
M_2(\CC)$.
\end{example}

\subsection{Adding the boundary}
\label{sec:AddingBoundary}
In \cite{CaKaLo2013} a purely operator algebraic construction of 
all boundary conditions is given. 
As a result a boundary net is obtained which is the direct sum of all boundary
conditions.

Let us consider the inclusion
$$
	\A(I)\otimes \A(J) \subset \cB_2(O) 
$$
for some fixed $O=I\times J\Subset W$
and let $\Theta_2$ be the associated Q-system in 
$\DHR^I(\A) \boxtimes \overline{\DHR^J(\A)}$.
Let $\Omega$ be the vacuum in $\H_\A$ and let us define 
the state 
$\varphi_0(x\otimes y)=(\Omega,xy\Omega)$
for $x\in \A(I)$, $y\in \A(J)$ and let $\varepsilon_O \colon \cB_2(O)\to
\A_2(O)\cong \A_+(O)$ be the conditional expectation.
This gives a state $\varphi=\varphi_0\circ \varepsilon_0$ on $\cB_2(O)$ (which can be extended to 
a state on $\mathfrak{A}_2(W)$). Using the GNS representation 
one get an inclusion $\A_+(O)\subset \cB_+(O)$ on a bigger Hilbert space and
which is by construction isomorphic to $\A_2(O)\subset\cB_2(O)$. 
This construction extends to  $\mathfrak{A}_2(W)$ and gives a (reducible) boundary net
$\{\cB_+(O)\}_{O\in\cK_+}$. 
Let us define $\cB(I)=\bigvee_{\cK_+\ni O \subset W(I)}\cB_+(O)$ where $W(I)$ is the
left wedge such that its intersection with the time axis $x=0$ is equals $I$.
This gives a non-local extension $\cB\supset\A$.
Let us fix $L \supset I \cup J$, then the Q-system of $\cB(L)\supset \A(L)$
can be chosen to be localized in $I\cup J$ and it can be in particular trivially extended
from the inclusion $\A_+(O)\subset \cB_+(O)$ using strong additivity. Let's
denote its Q-system by $\tilde \Theta$. 
\begin{prop}
	Let $\cB_2 \supset \A_2$ be a local irreducible extension with Q-system $\Theta_2$.
	The Q-system of the inclusion
    $\A(I)\subset \cB(I)$, where $\cB=\cB_+^\mathrm{gen}$ and 
    $\cB_+$ is obtained
	by adding the boundary
	is equivalent to the Q-system $T(\Theta_2)$.
\end{prop}
\begin{proof}
We have to show that $\tilde \Theta$ is equivalent to $T(\Theta_2)$, where we
see $\Theta_2$ as a Q-system by the equivalence $\NNs\boxtimes\overline\NNs\cong
\DHR^O(\A_2)$.

An endomorphism $\rho^I\boxtimes\bar\sigma^J$ gives an endomorphism 
$\rho^I\bar\sigma^J \in \End(\A(I)\vee\A(J))$ 
and this gives actually an isomorphism of tensor categories 
$$
        \End(\A(I)\otimes\A(J)) \cong \End(\A(I)\vee \A(J))
        \punkt
$$
Starting from an object in $\DHR^O(\A_2)$ the image is a localized endomorphism
of $\A(I)\vee \A(J)$ which can by strong additivity be extended to 
a localized endomorphism of $\End(\A(L))$, 
so we get a tensor functor
$$
      \tilde T\colon\DHR^I(\A_2) \to\DHR^L(\A)\equiv\NNs
$$
where we choose $N:=\A(L)$ and $\NNs=\DHR^L(\A)$. 
We note that 
the  $\mu$ from (\ref{eq:mu}) is  trivial as is 
$\varepsilon(\rho_2,\bar\sigma_1)$ because of the order of localization. 

So the functor 
$$
      \NNs\boxtimes\overline\NNs \cong \DHR^O(\A_2) \to\DHR^L(\A)\equiv\NNs
$$
is by construction equivalent to the tensor $T$ from Subsec.\
\ref{SubSec:Adjoint} 
and, in particular $\Tilde\Theta$ is equivalent to $T(\Theta_2)$.
\end{proof}
This gives as an alternative proof of Prop.\ \ref{prop:BC}.
Let us assume $\cB_2$ was modular invariant/maximal. 
All boundary conditions are obtained by the adding the boundary 
construction, and by Prop.\ \ref{prop:DirectSumQSystem}
we can conclude: 
\begin{cor}
All boundary conditions of $\cB_2$ come from an
$a\in\NDM$, where $N=\A(I)$, $M=\cB(I)$, 
$\NNs=\DHR^I(\A)$ and $\cB\subset \A$ is any 
(non-local) extension giving $\cB_2$ by the $\alpha$-induction construction.
\end{cor}

\section*{Acknowledgements}
The authors would like to thank
Karl-Henning Rehren for discussions and remarks
on the manuscript. We thank 
Ingo Runkel and Christoph Schweigert
for E-mail correspondence. 
M.B.~would like to thank
David~E.~Evans, Noah Snyder and 
Chenchang Zhu for discussions.

\renewcommand{\eprint}[1]{\href{http://arxiv.org/abs/#1}{#1}} 
\def\cprime{$'$}

\def\cprime{$'$}


\end{document}